\documentclass[12pt]{article}
\usepackage{amsmath, amssymb, amsfonts, mathtools, amsthm}
\usepackage{newpxtext, eulervm}
\usepackage[T1]{fontenc}
\usepackage[utf8]{inputenc}

\usepackage[compact]{titlesec}
\titlespacing{\section}{0pt}{2ex}{1ex}
\titlespacing{\subsection}{0pt}{1ex}{0ex}
\titlespacing{\subsubsection}{0pt}{0.5ex}{0ex}
\usepackage{algorithm}
\usepackage{algpseudocode}
\usepackage{graphicx}
\usepackage{xcolor}
\definecolor{linkcol}{HTML}{b51d3e}
\usepackage[
    colorlinks      = true,
    linkcolor       = linkcol,
    citecolor       = linkcol,
    urlcolor        = linkcol,
    backref         = page          
]{hyperref}

\usepackage{natbib}
\usepackage{setspace}
\usepackage{geometry}
\usepackage{microtype}
\usepackage{threeparttable}

\usepackage{amsmath, amssymb, amsfonts, mathtools}
\usepackage{bbm}            
\usepackage{bm}             

\usepackage{graphicx}
\usepackage{booktabs}
\usepackage{caption}
\usepackage{subcaption}

\graphicspath{{figures/}{output/}{./}}

\usepackage{enumitem}
\setlist{nosep}

\usepackage{xcolor}

\usepackage{booktabs,tabularx,longtable, makecell}
\usepackage{multirow,array,siunitx}
\usepackage{tikz}
\usetikzlibrary{arrows.meta,positioning,calc,decorations.pathreplacing,shapes.geometric,fit,intersections,patterns}
\usepackage{pgfplots}
\usepgfplotslibrary{statistics,fillbetween}
\pgfplotsset{compat=1.18}
\usepgfplotslibrary{groupplots}

\pgfmathdeclarefunction{normcdf}{2}{
  \pgfmathparse{0.5*(1 + tanh(1.5957691216*(#1)/(#2) + 0.0705659995*((#1)/(#2))^3))}%
}
\pgfmathdeclarefunction{PhiApprox}{2}{
  \pgfmathparse{1/(1 + exp(-1.702*(#1)/(#2)))}%
}

\pgfmathdeclarefunction{LHSsym}{7}{
  \pgfmathparse{ (#2)*(#1) + (#3)*(#4) * ( 1 - PhiApprox( ((1-#6)/#6)*((#1)-(#5)) , #7 ) ) }%
}

\newcommand{\SolveMstarSafe}[9]{%
\begingroup
  \pgfmathsetmacro{\mL}{-2.0}
  \pgfmathsetmacro{\mR}{ 4.0}
  \pgfmathsetmacro{\fL}{ LHSsym(\mL,#2,#3,#4,#5,#6,#7) - #8 }
  \pgfmathsetmacro{\fR}{ LHSsym(\mR,#2,#3,#4,#5,#6,#7) - #8 }

  \foreach \k in {1,...,12}{%
    \pgfmathparse{(\fL<0) ? -1 : 1}\let\sL\pgfmathresult
    \pgfmathparse{(\fR<0) ? -1 : 1}\let\sR\pgfmathresult
    \ifnum\sL=\sR
      \pgfmathsetmacro{\mL}{\mL - 2}
      \pgfmathsetmacro{\mR}{\mR + 2}
      \pgfmathsetmacro{\fL}{ LHSsym(\mL,#2,#3,#4,#5,#6,#7) - #8 }
      \pgfmathsetmacro{\fR}{ LHSsym(\mR,#2,#3,#4,#5,#6,#7) - #8 }
    \else
      \breakforeach
    \fi
  }%

  \foreach \iter in {1,...,#9}{%
    \pgfmathsetmacro{\mM}{ (\mL+\mR)/2 }
    \pgfmathsetmacro{\fM}{ LHSsym(\mM,#2,#3,#4,#5,#6,#7) - #8 }
    \pgfmathparse{(\fL*\fM<0) ? 1 : 0}%
    \ifnum\pgfmathresult=1
      \xdef\mR{\mM}
      \xdef\fR{\fM}
    \else
      \xdef\mL{\mM}
      \xdef\fL{\fM}
    \fi
  }%
  \xdef#1{\mM}%
\endgroup
}

\definecolor{computeColor}{RGB}{33,113,181}   
\definecolor{egressColor}{RGB}{200,54,53}     
\definecolor{stdColor}{RGB}{49,163,84}        
\definecolor{baselineColor}{RGB}{117,107,177} 
\definecolor{actionColor}{RGB}{253,174,97}    
\definecolor{stateColor}{RGB}{44,162,95}      
\usepackage{subcaption} 
\usepackage[font={small,stretch=1.1},labelfont=bf]{caption}
\newtheorem{theorem}{Theorem}[section]
\newtheorem{proposition}{Proposition}[section]
\newtheorem{lemma}{Lemma}[section]

\renewcommand{\Pr}{\mathbb{P}}

\begin{document}

\title{Incentivizing High Quality Entrants When Creators Are Strategic}
\author{Felicia Nguyen\thanks{Department of Marketing. \href{mailto:pnguy38@emory.edu}{pnguy38@emory.edu}.}\\Emory University}
\date{\today}
\maketitle
\thispagestyle{empty}
\begin{abstract}
\onehalfspacing
We study how a platform should design early exposure and rewards when creators strategically choose quality before release. A short testing window with a pass/fail bar induces a pass probability, the slope of which is the key sufficient statistic for incentives. We derive three main results. First, a closed-form ``implementability bounty'' can perfectly align creator and platform objectives, correcting for incomplete revenue sharing. Second, front-loading guaranteed impressions is the most effective way to strengthen incentives for a given attention budget. Third, when impression and cash budgets are constrained, the optimal policy follows an equal-marginal-value rule based on the prize spread and certain exposure. We map realistic ranking engines (e.g., Thompson sampling) into the model's parameters and provide telemetry-based estimators. The framework is simple to operationalize and offers a direct, managerially interpretable solution for platforms to solve the creator cold-start problem and cultivate high-quality supply.
\end{abstract}

\newpage
\setcounter{page}{1}
\section{Introduction}\label{sec:introduction}

Digital platforms increasingly confront a design problem that is as practical as it is conceptual: how to discover promising new creators while simultaneously giving those creators sufficient confidence that effort invested before launch will be noticed and rewarded \citep{bhargava2022creator, peres2024creator}. When early exposure is sparse or poorly timed, when selection at the promotion threshold is either too forgiving or too stringent, and when monetary incentives are spread thinly across low-information outcomes, the marginal return to pre-entry quality falls and a familiar cold-start externality emerges. This paper proposes a simple, implementable solution built on two instruments that platform designers already recognize: a guaranteed testing window that is paced early in the lifecycle, and a one-time, outcome-contingent bounty that is paid if a transparent graduation threshold is crossed. We show that these levers can be tuned to raise pre-entry quality, to concentrate scarce attention and cash where they buy the most incentive, and, under mild conditions, to replicate the first-best benchmark in which the platform could dictate investment directly.

Our starting point is the growing economics and marketing literature on information design and attention \citep{ke2022information}. A central insight of that work is that platforms can shape behavior by choosing what signals to release and how to structure the environment in which agents process those signals. In the canonical persuasion model, a designer commits to a signaling scheme that induces desired actions; in the rational inattention tradition, decision makers optimize subject to information costs and attention constraints, which makes the designer’s choice of \emph{what} to highlight and \emph{when} a primitive of the problem rather than an afterthought \citep{KamenicaGentzkow2011,BergemannMorris2019}. The empirical and theoretical evidence that limited attention and selective disclosure affect choices is now well established, including formal links between attention costs and discrete choice, and behavioral refinements that emphasize salience and the shrouding of non-salient attributes \citep{MatejkaMcKay2015,BordaloGennaioliShleifer2013}. In marketing, for example, \cite{shi2023design} posit through a Bayesian learning game that partial information disclosure in online platforms' reputation system can make sellers invest more in quality and maximize platform's profit. These ideas matter for platforms because early metrics and promotion rules are precisely the signals that creators and algorithms attend to, when designed thoughtfully they align incentives, and when designed poorly they mute them \citep{GabaixLaibson2006,GentzkowShapiro2010}.

A second strand relevant to our question is the economics and marketing of digital platforms, media, and advertising. The auction and ranking mechanisms that intermediate content and attention have been characterized in models of sponsored search, position auctions, and two-sided competition, with a particular emphasis on how platform rules shape the distribution of exposure and surplus across sides \citep{EdelmanOstrovskySchwarz2007,AtheyEllison2011}. The structure of platform markets and the possibility of bias or preferential integration have been analyzed in two-sided market frameworks and in models that make the platform’s editorial position explicit \citep{RochetTirole2006,Armstrong2006,DeCorniereTaylor2014}. Complementary work in marketing has documented how privacy, creative, and targeting choices affect the reach and persuasiveness of digital advertising, and how product and content ranking decisions shape discovery and conversion \citep{GoldfarbTucker2011,GhoseYang2009, ren2024advertising}. These contributions establish, both theoretically and empirically, that small design changes in ranking, pricing, and disclosure can move large amounts of traffic and revenue; they also make clear that platforms must manage scarce resources, impressions and cash, under institutional and fairness constraints \citep{YaoMela2011,LambrechtTucker2013,GordonZettelmeyerBhargava2019}.

A third, complementary strand studies exploration and learning. Classical analyses of bandit problems and strategic experimentation show that myopic allocation rules are often inefficient when information has option value; dynamic mechanisms can restore efficiency but are challenging to implement at scale \citep{Gittins1979,BoltonHarris1999}. More recent work considers environments in which a planner induces experimentation by sequential agents, and dynamic mechanisms that align private and social marginal contributions over time \citep{KremerMansourPerry2014,BergemannValimaki2010}. In a similar spirit to our paper, but under different setting, \citet{immorlica2020incentivizing} show that selective data disclosure in recommendation systems can incentivize exploration. . These ideas are directly applicable to platform discovery: the platform is the planner; new creators are the alternatives to be explored; and the allocation of early exposure is the instrument that creates the option value of learning. What is missing, and what this paper provides, is a tractable bridge between those normative insights and a set of parameters that can be shipped in a product stack and governed in a budgeting process \citep{AralWalker2011}.

The contribution is methodological and substantive. On the methodological side, we specify a transparent environment that distills the design problem to four primitives that map one-for-one to platform design choices: a guaranteed testing window, a posted graduation threshold, a continuation value that summarizes post-graduation exposure under the ranking engine, and a hit-based bounty that pays once if and only if the threshold is crossed. The model is intentionally parsimonious, binary engagement outcomes, a single index of pre-entry quality, and a reduced-form continuation scale, so that comparative statics can be derived in closed form and then read as instructions for the timing and sizing of early slots and targeted transfers. The focus on supply side incentives is not meant to diminish user side information design; it is meant to complement it. A platform can improve creators’ investment incentives without changing the user interface at all, simply by adjusting how much early exposure is guaranteed, where the graduation bar sits, how valuable graduation is, and whether there is a small payment at the moment when selection is most diagnostic.

On the substantive side, the paper delivers three design principles and a budgeting rule. First, timing matters: for a fixed number of guaranteed opportunities, the earliest feasible pacing weakly dominates alternatives. When early slots are front-loaded, creators face a larger certain return to effort without any change in the informativeness of the graduation event. Second, targeting matters: because the pass-or-fail event at the graduation bar is most informative near the middle of the cohort distribution, the bar should be placed where pass rates are neither vanishingly small nor trivially large, and small outcome-contingent payments should be concentrated on that event rather than spread across low-information outcomes. Third, alignment is feasible: a simple formula pins down a one-time bounty that eliminates the wedge between the social marginal value of pre-entry quality and the private marginal return implied by the platform’s revenue share, testing window, and continuation scale. Finally, the budgeting rule is to treat impressions and cash as tradable resources and to equalize the shadow-price-adjusted marginal value of an extra early impression and an extra expected dollar of payout. This “balanced exploration” prescription converts a theory problem into a dashboard problem: estimate the marginal lift in constrained welfare per early slot and per expected payout dollar, and then adjust the two instruments until each last unit is worth exactly its shadow price.

Our approach complements and extends the literatures summarized above in three ways. Relative to information design and attention, we focus on how signals and thresholds aimed at creators, not consumers, shape investment incentives, and we show how to use a posted bar and a one-time payment to enact the persuasion logic in a setting where the platform cannot observe pre-entry effort directly \citep{KamenicaGentzkow2011,BergemannMorris2019}. Relative to the platform and advertising literature, we treat early impressions as an incentive instrument rather than merely as an experimental resource, and we formalize the trade-off between attention and cash in a way that is compatible with campaign budgets and governance constraints \citep{EdelmanOstrovskySchwarz2007,DeCorniereTaylor2014}. Relative to exploration and dynamic mechanism design, we replace rich but fragile mechanisms with two parameters that have negligible implementation cost: front-load a small incubation window and pay a small bounty when the clearly defined bar is crossed; under familiar regularity conditions this pair attains the normative benchmark while remaining explainable to creators and auditable by policy teams \citep{KremerMansourPerry2014,BergemannValimaki2010}.

The empirical spirit of recent marketing work also informs the paper’s design emphasis. Studies of digital ad effectiveness and measurement have underscored how small adjustments in timing, targeting, and creative can move outcomes and budgets, and how governance disciplines should focus on marginal value per unit of resource rather than on average rates \citep{GoldfarbTucker2011,GordonZettelmeyerBhargava2019}. Parallel work on discovery, recommendations, and social diffusion has highlighted the role of early exposure and social proof in shaping eventual reach, as well as the heterogeneity of creators’ monetization opportunities across surfaces \citep{RutzBucklin2011,AralWalker2011}. By framing early slots as a scarce resource that purchases incentive strength and by making the bounty contingent on a high-diagnosticity event, the proposed policy directly targets those margins. It also admits clean segmentation: when segments differ in monetization, continuation value, or baseline pass rates, a common set of shadow prices can be used to skew early slots and bounties toward the cohorts where they are most effective \citep{YaoMela2011,LambrechtTucker2013}.

The findings can be summarized succinctly. In the baseline with no guarantees and no targeted payments, the marginal return to quality is weak for typical entrants; raising the graduation bar without fixing early exposure often reduces incentives by making passes rarer without making them more informative. Adding a small guaranteed testing window and front-loading it strengthens incentives mechanically by making a portion of effort certain to be seen; positioning the bar where pass events are diagnostic then makes the graduation margin decisive; and concentrating cash on that event buys a large increase in effort at low expected cost because each dollar is spent where the pass probability is most responsive to quality. Under standard convexity and smoothness conditions, there is a unique equilibrium level of pre-entry quality, it rises monotonically with each lever, and a one-line bounty eliminates the residual wedge between the social and private marginal values when revenue sharing is incomplete. Treating impressions and cash as budgets then yields a simple rule: increase early slots until the marginal lift per slot equals the shadow price of attention; increase the bounty until the marginal lift per expected payout dollar equals the shadow price of cash; and retune the bar if the pass event becomes too rare or too common for targeting to be meaningful.

The remainder of the paper formalizes these ideas and connects them to the platform design toolkit. The model section specifies the primitives and establishes how early exposure, the graduation bar, continuation value, and a hit-based bounty map to creators’ investment choices. The main results establish existence and uniqueness, deliver comparative statics, prove the optimality of front-loading, and show how a small, posted bounty implements the planner’s benchmark. The resource-constrained section casts impressions and cash as budgets and derives an “equalize marginal value per unit resource” rule that scales to segments and multi-winner discovery. The extensions show that the prescriptions are robust to over-dispersed outcomes, to index or sampling engines in which continuation is history-dependent, and to modest integrity risks. Throughout, the emphasis is on simplicity and implementability: the recommended policy relies only on guarantees and payouts tied to a threshold that the platform can publish in creator documentation and can monitor with standard telemetry. In short, when attention and cash are scarce resources and selection is a slope, good policy spends the former where the latter is steep.

\section{Model}\label{sec:model}

This section develops a tractable environment for studying how a platform can induce pre-entry quality investment by new creators through its early-stage exposure and payment policies. The model is intentionally parsimonious, outcomes are binary, effort is summarized by a single quality index, and exposure is summarized by a history-dependent allocation rule, so that the primitives map transparently to design parameters used by product and operations teams. Throughout, we interpret “engagement” broadly (e.g., a completed view, a click-through, or a purchase proxy), and treat it as a measurable success under the platform’s instrumentation.

\subsection{Players, Timing, and Signals}\label{subsec:players_timing_signals}

\begin{table}[!ht]
\centering
\caption{Notation and Primitives in the Model}
\label{tab:notation_correct}
\begin{tabularx}{\textwidth}{@{}>{\raggedright\arraybackslash}p{2.5cm} X@{}}
\toprule
\textbf{Symbol} & \textbf{Meaning} \\
\midrule
\multicolumn{2}{l}{\textit{Creator Primitives}} \\
$\mu$ & Creator's pre-entry quality choice, $\mu \in [\underline{\mu}, \overline{\mu}] \subset (0,1)$. \\
$c(\mu)$ & Cost of choosing quality $\mu$, assumed strictly convex. \\
$\alpha$ & Creator's revenue share from realized engagement, $\alpha \in (0,1]$. \\
$\Pi_C$ & The creator's discounted expected payoff function. \\
\addlinespace
\multicolumn{2}{l}{\textit{Platform Policy Instruments}} \\
$q$ & Discounted number of guaranteed impressions in the testing window. \\
$\bar{\mu}$ & The posted graduation bar, representing a target success rate. \\
$s$ & The integer threshold of successes required for graduation, $s = \lceil q\bar{\mu} \rceil$. \\
$B$ & One-time discovery bounty paid if and only if the creator graduates ($S \ge s$). \\
$H$ & Discounted continuation value (expected exposure) if the creator graduates. \\
\addlinespace
\multicolumn{2}{l}{\textit{Model Mechanics and Outcomes}} \\
$S$ & Random variable for the number of successes in $q$ trials, $S \sim \mathrm{Binomial}(q, \mu)$. \\
$P(\mu)$ & Probability of graduation for a creator of quality $\mu$; $P(\mu) = \Pr(S \ge s)$. \\
$P'(\mu)$ & Slope of the pass probability w.r.t. quality; the key measure of incentive strength, equal to a Beta PDF. \\
$\Xi(\mu)$ & Total discounted expected exposure for a creator of quality $\mu$. \\
$\gamma$ & Per-period discount factor, $\gamma \in (0,1)$. \\
\addlinespace
\multicolumn{2}{l}{\textit{General Model and Resource Constraints}} \\
$H_0, H_1$ & Discounted continuation value upon failing ($H_0$) or passing ($H_1$) the bar. \\
$\Delta H$ & The prize spread of graduating, $\Delta H = H_1 - H_0$. \\
$R, M$ & Per-entrant budgets for guaranteed impressions ($R$) and expected cash payout ($M$). \\
$\lambda_{\mathrm{imp}}, \lambda_{\$}$ & Shadow prices for the impression and cash budget constraints, respectively. \\
\bottomrule
\end{tabularx}
\end{table}
Time is discrete, $t=1,2,\ldots$, and the platform discounts future payoffs at rate $\gamma\in(0,1)$. In each period a fresh user arrives, the platform selects content to display, and an engagement outcome is realized. New creators enter over time; we analyze the incentives faced by a representative entrant $i$ at her entry date and suppress the entry index when this creates no confusion.

Upon entry, the creator chooses a quality level $\mu\in[\underline{\mu},\overline{\mu}]\subset(0,1)$ at cost $c(\mu)$, where $c$ is continuously differentiable, strictly convex, and satisfies $c'(\underline{\mu})=0$. Economically, $\mu$ captures all pre-launch investments that raise the probability that a randomly matched user finds the content valuable, scripting and editing effort, creative iteration, thumbnail and titling craft, or, in longer-form contexts, research and post-production. Convexity reflects the familiar idea that it is easy to reach basic adequacy but increasingly costly to extract additional probability points of success. From a managerial perspective, the function $c(\cdot)$ indexes the “quality elasticity” of the creator population; policies that shift the marginal return to $\mu$ will translate into predictable moves along this schedule.

When the platform surfaces the entrant in period $t$, a binary outcome $Y_t\in\{0,1\}$ is realized with
\[
Y_t \sim \mathrm{Bernoulli}(\mu)\quad\text{independently across $t$ conditional on }\mu,
\]
and $Y_t$ is observed by the platform and recorded by its analytics stack.\footnote{Assuming platform-measured outcomes rules out direct manipulation and keeps the model focused on real quality rather than reporting games. This assumption is standard in analytical work that studies incentive provision rather than fraud.} We adopt the Bernoulli specification for clarity; Section~\ref{sec:extensions} discusses how the analysis extends to richer signal technologies as long as the early-success probability is differentiable in $\mu$.

Let $x_t\in[0,1]$ denote the period-$t$ exposure probability for the entrant (e.g., the share of relevant impressions or the probability of receiving the focal slot in a feed). An \emph{allocation rule} is a measurable mapping
\[
\mathcal{A}: \mathcal{H}_t \longrightarrow [0,1],\qquad x_t=\mathcal{A}(H_t),
\]
from the observable history $H_t$ of past exposures and outcomes to the entrant’s current exposure. This notation nests many concrete ranking and pacing policies used in practice: a greedy policy that favors the highest posterior-mean item, a Thompson sampler that draws proportionally to posterior optimism, or a rule that opens a testing window and then either graduates or shelves the entrant. We define the \emph{discounted expected exposure} (or “discounted pull count”) induced by $\mathcal{A}$ at quality $\mu$ as
\[
\Xi(\mu;\mathcal{A}) \;\equiv\; \mathbb{E}_\mu\Bigg[\sum_{t\ge 1}\gamma^{t-1} x_t \Bigg],
\]
where the expectation is taken over the outcome process and any randomized elements of $\mathcal{A}$.

The platform monetizes engagement directly (ads, subscriptions) or indirectly (retention, referral). To keep the connection to managerial KPIs explicit, we let the platform’s per-period objective be expected engagement net of any transfers it commits to pay creators under announced policies. In the main analysis of Section~\ref{sec:mechanisms_main} we allow for a transparent, outcome-contingent “discovery bounty”; for now we keep transfers implicit and focus on the exposure engine.

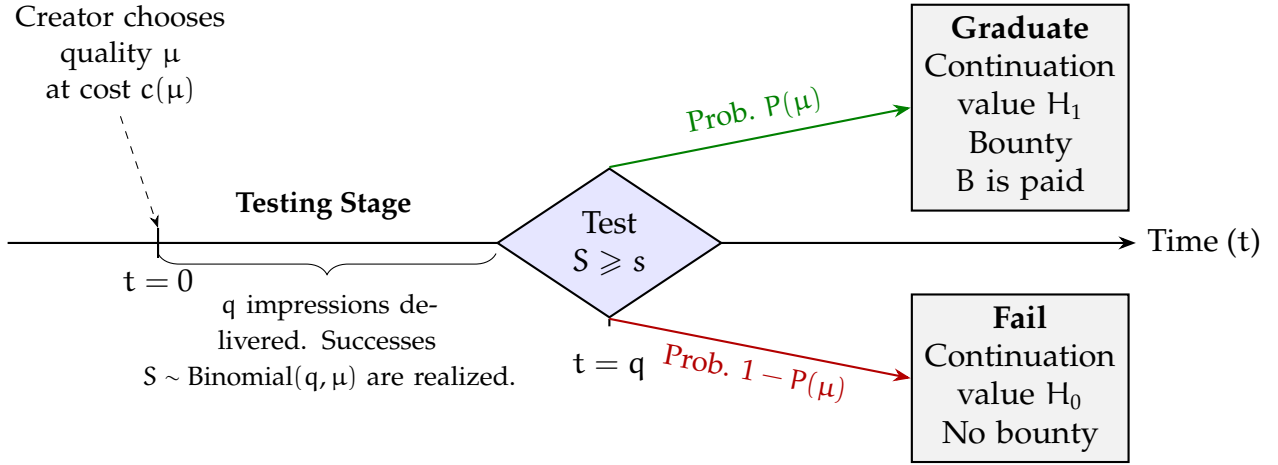
\begin{figure}[!ht]
\centering
\begin{tikzpicture}[
    node distance=1.2cm and 2cm, 
    timeline/.style={-Stealth, thick},
    stage/.style={
        rectangle, 
        draw, 
        thick, 
        fill=gray!10, 
        minimum height=1.2cm, 
        minimum width=2.8cm, 
        text width=2.6cm, 
        align=center
    },
    decision/.style={
        diamond, 
        draw, 
        thick, 
        fill=blue!10, 
        aspect=1.5, 
        minimum size=1.5cm, 
        align=center
    },
    event/.style={
        text width=3.5cm, 
        align=center, 
        font=\small
    },
    pass/.style={-Stealth, color=green!50!black, thick},
    fail/.style={-Stealth, color=red!70!black, thick}
]


\draw[timeline] (-1,0) -- (14,0) node[right] {Time ($t$)};

\node[event, align=center] (choice) at (0.5, 2.5) {Creator chooses quality $\mu$ \\ at cost $c(\mu)$};

\node[decision] (test) at (7, 0) {Test \\ $S \ge s$};

\node[stage, right=2.5cm of test, yshift=1.8cm] (graduate) {
    \textbf{Graduate} \\
    Continuation value $H_1$ \\
    Bounty $B$ is paid
};
\node[stage, right=2.5cm of test, yshift=-1.8cm] (fail) {
    \textbf{Fail} \\
    Continuation value $H_0$ \\
    No bounty
};


\draw[thick] (1,0.2) -- (1,-0.2) node[below] {$t=0$};
\draw[thick] (test.south) ++(0,-0.1) -- (test.south) ++(0, -0.3) node[below] {$t=q$};

\node[font=\small] at (3.2, 0.5) {\textbf{Testing Stage}};

\draw[decorate, decoration={brace, amplitude=8pt, mirror, raise=4pt}] (1,0) -- (test.west) 
    node[midway, below=15pt, font=\small, text width=5cm, align=center] 
    {\footnotesize $q$ impressions delivered. Successes $S \sim \text{Binomial}(q,\mu)$ are realized.};

\draw[-Stealth, dashed] (choice.south) -- (1,0.2);

\draw[pass] (test.north) -- (graduate.west) 
    node[midway, above, sloped, font=\small] {Prob. $P(\mu)$};
\draw[fail] (test.south) -- (fail.west) 
    node[midway, below, sloped, font=\small] {Prob. $1-P(\mu)$};

\end{tikzpicture}
\caption{Timeline of Creator Incentives and Platform Actions}
\label{fig:timeline_revised}
\end{figure}
On the creator side, we assume the entrant receives a revenue share $\alpha\in(0,1]$ tied to realized engagement (for example, an ad-revenue split or an affiliate rate). Because $\alpha$ is a pass-through of the platform’s marginal revenue scale, it is natural to treat it as a parameter; in applications, lower-monetized surfaces (or early-stage pilots) correspond to smaller $\alpha$, which will matter when we study how cash and exposure instruments substitute for each other. Given an allocation rule $\mathcal{A}$, the entrant’s discounted expected payoff is
\[
\Pi_C(\mu;\mathcal{A}) \;=\; \alpha\,\mu\,\Xi(\mu;\mathcal{A}) \;-\; c(\mu),
\]
while the platform’s discounted expected objective (before transfers) is
\[
W_0(\mu;\mathcal{A}) \;=\; \mu\,\Xi(\mu;\mathcal{A}).
\]
Two conceptual remarks are useful at this point. First, if creators were non-strategic arms with fixed $\mu$ drawn from a common prior, the platform would simply choose $\mathcal{A}$ to maximize $W_0$; classical index policies provide a clean benchmark for that problem. Second, once $\mu$ becomes a choice variable, the allocation rule indirectly shapes investment incentives through $\Xi(\mu;\mathcal{A})$ and its sensitivity to $\mu$. The manager’s design problem is therefore not only to learn and route traffic efficiently but also to engineer the \emph{marginal return to quality} at the moment where creators decide how much to invest before they are known to the system.

\subsection{Baseline Without Instruments and the Cold-start Externality}\label{subsec:baseline_coldstart}

We now describe a baseline in which the platform uses a standard learning-to-rank engine without any explicit early-stage guarantees or outcome-contingent payments. This captures a widely used operational approach: new creators are ``thrown into the pool,'' the algorithm opportunistically tries them when uncertainty is high enough or incumbents underperform, and those who happen to look promising early receive more traffic while others fade out quickly. The baseline is efficient when arms are exogenous; it is not generally incentive-compatible when arms invest.

To make the baseline’s implications transparent without committing to a particular algorithm, we summarize its early-stage logic by two primitives. First, there is an \emph{organic testing window} of discounted expected size $m\ge 0$ that a typical entrant receives ``for free,” generated by the algorithm’s intrinsic appetite to sample uncertain options. Second, conditional on the $m$ early exposures, the entrant is \emph{graduated} to a larger exploitation stream with discounted size $H_0\ge 0$ if her realized early performance crosses a data-driven bar. Let $P_0(\mu)\in[0,1]$ denote the probability of graduation as a function of true quality under the baseline rule.\footnote{In Section~\ref{sec:mechanisms_main} we will specialize this graduation event to an explicit binomial test and characterize $P_0(\mu)$ and its slope in closed form. For the present discussion it suffices that $P_0(\cdot)$ is increasing and differentiable on $(\underline{\mu},\overline{\mu})$.} The entrant’s discounted expected exposure under the baseline is then
\[
\Xi_0(\mu)\;=\; m \;+\; H_0\, P_0(\mu).
\]
Intuitively, $m$ plays the role of a default ``look,” while $H_0 P_0(\mu)$ captures the high-stakes graduation margin: small differences in early success rates translate into large differences in subsequent traffic. The creator’s discounted expected payoff becomes
\[
\Pi_C(\mu;\Xi_0) \;=\; \alpha\,\mu\,\big(m + H_0 P_0(\mu)\big) \;-\; c(\mu).
\]
Under standard regularity conditions, the best-response quality $\mu^\star$ solves the first-order condition
\begin{equation}\label{eq:baseline_foc}
\alpha\Big(m + H_0 P_0(\mu^\star)\Big)\;+\;\alpha\,\mu^\star\,H_0\,P_0'(\mu^\star)\;=\; c'(\mu^\star),
\end{equation}
with, under Assumption 1 below, the left-hand side strictly increasing in $\mu$ and the right-hand side strictly increasing by convexity. Equation~\eqref{eq:baseline_foc} organizes the baseline’s incentive effects into two economically distinct channels. The term $\alpha m$ is a level effect: even if graduation were impossible, a larger organic window raises the private return to quality proportionally. The term $\alpha\,\mu H_0 P_0'(\mu)$ is a slope effect: holding the expected size of the exploitation stream fixed, incentives are stronger when early performance is a \emph{steep} separator of types (high $P_0'(\mu)$), because a marginal increase in $\mu$ substantially changes the odds of “being picked” for graduation.

The cold-start externality appears when both channels are weak at the relevant portion of the type distribution. If operational constraints or conservative ranking heuristics keep $m$ small, and if the baseline bar is tuned so that $P_0'(\mu)$ is shallow around the mass of entrants (for example, because the test is noisy or implicitly too forgiving), then the left-hand side of~\eqref{eq:baseline_foc} is low throughout the interior of $[\underline{\mu},\overline{\mu}]$. In that case the unique solution $\mu^\star$ sits close to $\underline{\mu}$, and a large measure of otherwise promising creators under-invest. From a product vantage point, this is precisely the pattern one sees when new-creator cohorts are asked to “prove themselves” quickly in crowded feeds: many entrants never receive enough early, high-leverage exposure for their effort to matter, and anticipating this, they exert less effort ex ante. The platform then rationally interprets the cohort as low quality and reduces testing intensity further, a self-reinforcing loop.

Equation~\eqref{eq:baseline_foc} also clarifies why naïvely raising the graduation bar can backfire. Moving the bar upward may indeed reduce false positives, but if it simultaneously flattens $P_0'(\mu)$ at the typical entrant’s $\mu$ (for instance, by placing the bar in a region where the test has little power), then the slope effect weakens and incentives fall. Conversely, a small increase in the organic window $m$ can have a disproportionately large impact on investment if it shifts the equilibrium $\mu^\star$ into a region where $P_0'(\mu)$ is steeper, thereby activating the graduation margin. These interactions are central to the policy analysis that follows: the instruments we study in Section~\ref{sec:mechanisms_main}, front-loaded exploration slots and hit-based bounties, are designed to raise the level and the slope terms in a cost-effective, implementable way.
\subsection{Design Choices: Early Exposure and Outcome-contingent Rewards}\label{subsec:policy_instruments}

Having established the baseline case, we now introduce the two instruments that will anchor the subsequent analysis: a guaranteed early-stage exposure window for entrants and a simple outcome-contingent transfer that is paid only when early performance clears a transparent bar. The first instrument acts directly on the level of the private return to quality by ensuring that pre-entry effort will be “seen.” The second concentrates monetary incentives exactly where selection is sensitive, so that each budgeted dollar purchases a large increase in the probability of discovery. Both instruments are already familiar in practice under various names, \emph{new-creator slots}, \emph{incubation windows}, \emph{graduation thresholds}, and \emph{launch bonuses}, and are straightforward to operationalize within standard experimentation and ranking stacks.

Formally, fix four policy primitives $(q,\bar{\mu},B,H)$. The platform commits to a \emph{testing stage} that grants a newly entering creator $q\in\mathbb{N}$ expected discounted impressions.\footnote{For clarity we model $q$ as a discounted count; in implementation it corresponds to a paced sequence of guaranteed opportunities across the earliest relevant contexts. As in Section~\ref{subsec:players_timing_signals}, exposure can be interpreted as a probability of receiving a focal slot.} During this stage the entrant’s content is surfaced irrespective of early realizations, so that the ex ante marginal return to quality reflects not only the possibility of graduation but also a nontrivial mass of certain views. At the end of testing, the platform applies a \emph{graduation rule} that compares the entrant’s realized early performance to a threshold $\bar{\mu}\in(0,1)$. If the bar is cleared, the entrant is admitted into a larger exploitation stream whose discounted expected size is summarized by $H\ge 0$; otherwise testing concludes and the entrant reverts to organic exposure only.

The outcome-contingent instrument is a \emph{discovery bounty} $B\ge 0$ that is paid once if and only if the entrant clears the bar. The bounty is best viewed as a small, transparent top-up to the revenue share that is concentrated on the event that the entrant is selected for graduation. Because it attaches cash to the same observable that triggers exposure, it requires no additional data or auditing infrastructure. From a manager’s perspective, the two design choices serve distinct but complementary purposes: $q$ allocates a scarce attention resource in a way that is guaranteed to be felt by every entrant, while $B$ focuses the monetary budget on the pivotal cases where the return to effort is steepest. Section~\ref{sec:mechanisms_main} will show that a simple choice of $B$ exactly aligns private and social marginal returns.

To make the graduation rule operational, it is useful to translate the bar $\bar{\mu}$ into a test on the realized number of early successes. Let $S$ denote the number of positive outcomes observed during testing. Because outcomes are Bernoulli and exposures are paced deterministically during testing, the natural statistic is a binomial count. For a given $q$ and bar $\bar{\mu}$, the platform chooses the smallest integer threshold
\[
s \;\equiv\; \left\lceil q\,\bar{\mu}\right\rceil \in \{1,\ldots,q\},
\]
and graduates the entrant if and only if $S\ge s$. The event $\{S\ge s\}$ is exactly the one that triggers both the exposure “unlock” and the bounty payment. This pairing keeps the rule legible to creators and makes it easy to explain and govern: the platform publicizes the size of the incubation window, the success threshold, and the one-time bonus for crossing it. In the background, $H$ captures the expected discounted pull count that the exploitation engine would allocate to a promising item after graduation; in later sections we will reinterpret $H$ as the continuation value induced by a general index or Thompson policy.

The policy $(q,\bar{\mu},B,H)$ thus generates a simple two-stage exposure path. The entrant enjoys $q$ discounted impressions in expectation regardless of outcomes. Conditional on early realizations, the entrant either enters a larger exposure stream of discounted size $H$ and receives a bounty $B$, or exits the testing stage and receives no additional transfer. As we will see next, the key object that mediates incentives is the probability of clearing the bar as a function of true quality, together with its slope. These two primitives organize the entire analysis, including the comparative statics for $q$ and the design of $B$.

\subsection{Early Success Probabilities}\label{subsec:early_success}

Let $S\sim \mathrm{Binomial}(q,\mu)$ denote the number of successes realized during the testing stage when the entrant’s true quality is $\mu$. The probability of graduation is
\begin{equation}\label{eq:Pmu}
P(\mu)\;\equiv\; \Pr\!\big[S\ge s\big] \;=\; \sum_{k=s}^{q} \binom{q}{k}\,\mu^k(1-\mu)^{\,q-k},
\end{equation}
where $s=\lceil q\bar{\mu}\rceil$ as defined above. This tail probability collects, in reduced form, the entire mapping from unobserved quality to the chance of being “picked” after testing. Two analytic properties are especially useful for design. First, $P(\mu)$ is strictly increasing in $\mu$ on $(0,1)$ whenever $1\le s\le q$; better content is more likely to clear the bar. Second, and more importantly for incentives, the \emph{slope} of $P(\mu)$ with respect to $\mu$ is the density of a Beta distribution:
\begin{equation}\label{eq:Pprime}
P'(\mu) \;=\; \frac{\mu^{\,s-1}(1-\mu)^{\,q-s}}{B\!\left(s,\,q-s+1\right)} \;=\; f_{\mathrm{Beta}\left(s,\,q-s+1\right)}(\mu),
\end{equation}
where $B(\cdot,\cdot)$ is the Beta function. Equation~\eqref{eq:Pprime} follows from the identity $P(\mu)=I_{\mu}(s, q-s+1)$, where $I_{\mu}$ denotes the regularized incomplete Beta function, and the fact that $\partial_\mu I_{\mu}(a,b)=\mu^{a-1}(1-\mu)^{b-1}/B(a,b)$ for $a,b>0$.

Figure~\ref{fig:sec2-P-slope} visualizes the graduation technology. The pass probability
$P(\mu)=\Pr[S\ge s]$ is S‑shaped in quality, while the slope $P'(\mu)$ is the density
$\mathrm{Beta}(s,q{-}s{+}1)$ with a single interior peak near $(s{-}1)/(q{-}1)$.
This peak summarizes the \emph{diagnosticity} of the bar: policies that place the induced
equilibrium near this region convert small changes in effort into large changes in pass rates.
Managerially, dollars paid upon passing “buy” the slope $P'(\mu)$, whereas guaranteed
impressions “buy” the level term. The figure makes clear why tuning $(q,s)$ to keep
the cohort near the steep part of the curve is central to cost‑effective incentives.

\begin{figure}[!ht]
  \centering
  \includegraphics[width=.7\linewidth]{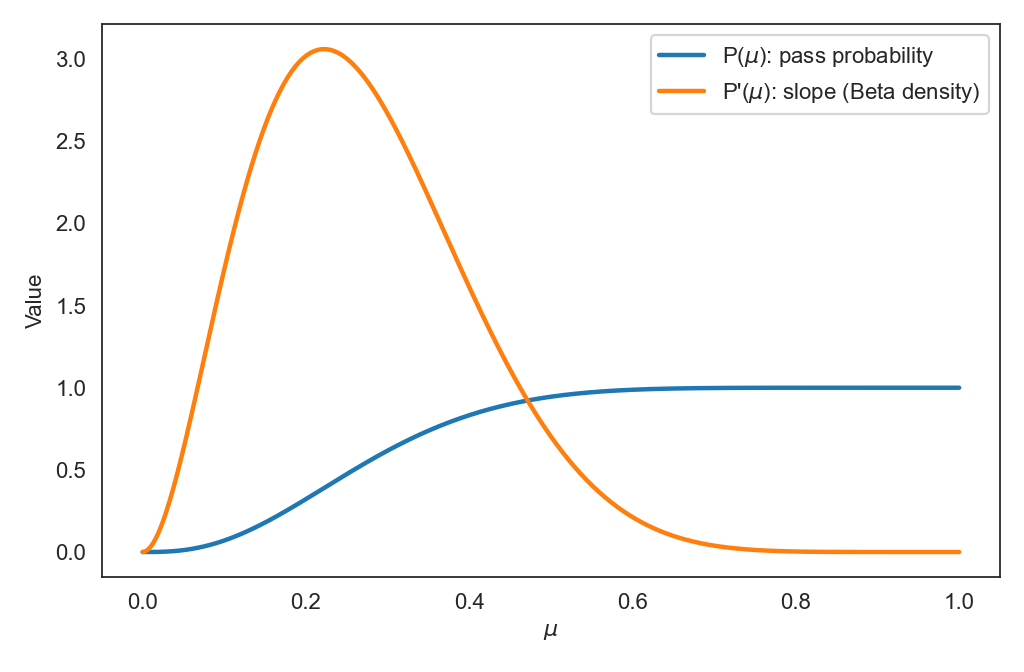}
  \caption{\textbf{Pass probability and its slope at $q_s$.} The S-curve shows $P(\mu)=\Pr[S\ge s]$ for $S\sim\mathrm{Bin}(q,\mu)$. The bell-shaped curve is $P'(\mu)$, the Beta density $\mathrm{Beta}(s,q-s+1)$. Diagnostic power at the bar is governed by the height of $P'(\mu)$; locating the equilibrium near this peak makes graduation dollars and early attention most effective.\\
\footnotesize\emph{Primitives used:} $q{=}10$, $s{=}3$.}
  \label{fig:sec2-P-slope}
\end{figure}

The interpretation of \eqref{eq:Pprime} is immediate and managerially powerful. The derivative $P'(\mu)$ measures how much a marginal increase in true quality changes the probability of graduation. Because it coincides with a Beta density, its shape is unimodal and peaks near $(s-1)/(q-1)$ when $1<s<q$; it decays to zero at the boundaries $\mu\to 0$ and $\mu\to 1$. In practical terms, this means that the graduation rule is most \emph{informative}, and therefore most potent as an incentive device, precisely in the region of the quality distribution where the bar sits. If the bar is placed too low relative to typical entrants, almost everyone passes and the slope is flat; effort is not rewarded at the margin. If the bar is placed too high, almost no one passes and the slope is again flat; effort is discouraged because the chance of graduation barely responds. By contrast, setting the bar so that a substantial fraction of entrants are on the margin produces a steep $P'(\mu)$, which in turn amplifies the effect of either guaranteed exposure $q$ or a small bounty $B$ on pre-entry investment.

It is convenient to summarize expected exposure and transfers under the policy using $P(\mu)$. Because the testing window guarantees $q$ discounted impressions and graduation unlocks a continuation stream of expected size $H$, the entrant’s discounted expected exposure is
\begin{equation}\label{eq:Xi_policy}
\Xi(\mu; q,\bar{\mu},H) \;=\; q \;+\; H\,P(\mu).
\end{equation}
If the platform pays the discovery bounty only upon graduation, the expected transfer is $B\,P(\mu)$. These two expressions permit a compact representation of objectives that will be central in Section~\ref{sec:mechanisms_main}. The platform’s discounted expected objective, net of transfers, is
\begin{equation}\label{eq:platform_objective}
W(\mu;\,q,\bar{\mu},B,H)\;=\; \mu\,\big[q + H P(\mu)\big]\;-\;B\,P(\mu),
\end{equation}
while the entrant’s discounted expected payoff, given a revenue share $\alpha\in(0,1]$, is
\begin{equation}\label{eq:creator_objective}
\Pi_C(\mu;\,q,\bar{\mu},B,H) \;=\; \alpha\,\mu\,\big[q + H P(\mu)\big] \;+\; B\,P(\mu) \;-\; c(\mu).
\end{equation}
Equations~\eqref{eq:platform_objective}-\eqref{eq:creator_objective} make clear how the instruments operate. The testing window $q$ raises payoffs linearly even absent graduation; it guarantees that a portion of the creator’s effort will be exposed. The continuation size $H$ scales the strategic value of being selected; larger $H$ increases the stakes and, through $P(\mu)$, the convexity of exposure with respect to quality. The bounty $B$ leverages the steep region of $P'(\mu)$ to target marginal incentives cost-effectively: because the transfer occurs only when graduation is imminent, a small $B$ can deliver a large increase in effort when $P'(\mu)$ is large. Finally, the bar $\bar{\mu}$ shapes the location and height of the $P'(\mu)$ peak and therefore governs the trade-off between discoverability and selectivity. The four parameters thus admit a natural division of labor: $q$ makes investment pay even for those who will not immediately graduate, $H$ and $\bar{\mu}$ determine the intensity and locus of selection, and $B$ aligns private incentives with the platform’s long run objective without paying for low-information outcomes.

A final observation concerns robustness to the exact mechanics of exploitation. In \eqref{eq:Xi_policy} we have summarized the continuation by a single parameter $H$. This simplification is without loss for the comparative statics logic that follows. In a fully specified ranking engine, $H$ should be interpreted as the expected discounted pull count that an entrant would receive conditional on graduation under the chosen index or sampling policy. All subsequent expressions carry through with that interpretation. The advantage of the present summary is that it keeps the object of design visible to managers: $H$ is ``how much graduation is worth,'' $q$ is ``how much we guarantee up front,'' $\bar{\mu}$ is “where we set the bar,” and $B$ is “how much we are willing to pay when the bar is cleared.''

\paragraph{Equilibrium investment under the policy.}
Given the primitives in \eqref{eq:Xi_policy}-\eqref{eq:creator_objective}, the entrant chooses $\mu$ to maximize $\Pi_C(\mu;q,\bar{\mu},B,H)$. Differentiating with respect to $\mu$ and invoking \eqref{eq:Pprime} yields the first-order condition for an interior solution:
\begin{equation}\label{eq:policy_foc}
\alpha\Big[q + H P(\mu^\star)\Big]\;+\;\alpha\,\mu^\star\,H\,P'(\mu^\star)\;+\;B\,P'(\mu^\star)\;=\;c'(\mu^\star).
\end{equation}
The structure of \eqref{eq:policy_foc} mirrors the economic decomposition already visible in the baseline. The guaranteed testing window $q$ operates through a level effect that does not depend on the entrant’s realized early outcomes; it ensures that some portion of effort is monetized with certainty. The continuation size $H$ scales both the level and the slope terms: by making graduation more valuable, it raises incentives even for creators who are confident of graduation (the $H P(\mu)$ term) and, crucially, it steepens the marginal return to effort for those on the cusp (the $H P'(\mu)$ term). The bounty $B$ reinforces the slope term by concentrating transfers exactly where the graduation probability is sensitive. Equation \eqref{eq:policy_foc} therefore already encodes the central design message for managers: front-load some exposure to make effort pay on average, set the bar where the slope is steep, and reserve payments for the pivotal boundary where a small improvement in quality most changes future allocation.

To make the optimization problem well-posed without imposing heavy structure, we maintain a mild regularity assumption that is satisfied by the standard functional forms used in applications.

\noindent\textbf{Assumption 1 (Regularity and single crossing).} The cost function $c$ is twice continuously differentiable and strictly convex on $[\underline{\mu},\overline{\mu}]$. Moreover, $c''(\mu)$ is large enough relative to $\sup_{\mu\in(\underline{\mu},\overline{\mu})}\big|\,\alpha \mu H P''(\mu)+B P''(\mu)\,\big|$ so that the difference
\[
\Delta(\mu)\;\equiv\;c'(\mu)\;-\;\Big(\alpha[q+H P(\mu)]\;+\;\alpha \mu H P'(\mu)\;+\;B P'(\mu)\Big)
\]
is strictly increasing in $\mu$.

Assumption~1 is weak in the sense that it holds whenever $c$ is, for example, quadratic or has increasing marginal costs with slope that dominates the local curvature of $P$; it does not require global concavity of the benefit term. Under Assumption~1, $\Delta(\mu)$ crosses zero at most once, which implies existence and uniqueness of $\mu^\star$ solving \eqref{eq:policy_foc}. Two further observations are immediate and will be formalized in Section~\ref{sec:mechanisms_main}. First, because the left-hand side of \eqref{eq:policy_foc} shifts upward with $q$, $H$, $\alpha$, and $B$ at every $\mu$, the unique solution $\mu^\star$ is weakly increasing in each policy parameter. Second, the comparative statics are economically interpretable: expanding $q$ raises the base return to quality, expanding $H$ raises the strategic stakes of being selected, increasing $\alpha$ lets creators internalize more of the engagement surplus, and increasing $B$ targets additional marginal returns where they are most effective.

\paragraph{Timing of guaranteed exposure and the case for front-loading.}
The testing window has been represented thus far as a discounted count $q$. How that count is scheduled over calendar time matters for incentives because future exposures are discounted while the selection statistic $S$ and the graduation probability $P(\mu)$ depend only on the \emph{number} of testing impressions, not their timing, under the maintained i.i.d.\ outcome structure. To see this, fix $q$, $B$, and the graduation rule $(s,H)$, and consider two schedules that differ only in when the $q$ guaranteed impressions occur. The direct marginal benefit of effort that does not pass through selection is $\alpha\,\mu \sum_{t\in\mathcal{E}}\gamma^{t-1}$, where $\mathcal{E}$ indexes the testing slots. This expression is maximized by placing guaranteed impressions as early as feasible. By contrast, both $P(\mu)$ and $P'(\mu)$ depend on the distribution of $S$ conditional on the number of trials and are therefore invariant to the schedule. Consequently, for a fixed $q$ the earliest feasible pacing of testing impressions weakly increases the left-hand side of \eqref{eq:policy_foc} at every $\mu$; by single crossing, it weakly increases $\mu^\star$. 

From an operational standpoint, this argument justifies the common practice of “incubation windows” or “launch ramps” that are concentrated near entry. Early pacing is not merely generous; it is the cost-effective way to convert a given testing budget into stronger pre-entry incentives, precisely because it leaves the selection statistic unchanged while increasing the certain component of the return to effort. The design guidance is robust to the exact form of the exploitation engine and, as emphasized earlier, only requires that testing outcomes be aggregated by counts rather than by time.

\paragraph{Mapping primitives to implementation and measurement.}
Equations \eqref{eq:policy_foc}-\eqref{eq:creator_objective} also clarify how the model’s primitives translate to parameters and metrics that product teams can own. The guaranteed window $q$ corresponds to a paced probability of occupying a focal slot across the earliest relevant contexts (e.g., first sessions or first $N$ candidate exposures), and its execution can be audited in the same way experimentation platforms audit treatment delivery. The bar $\bar{\mu}$ is operationalized as a threshold $s$ on observed successes; its effective location in the entrant distribution can be monitored by tracking the empirical pass rate and the empirical analogue of $P'(\mu)$, namely the change in pass rate with respect to small improvements in early-stage quality proxies. The continuation value $H$ is simply the expected discounted pull count that the ranking policy allocates after graduation; it can be estimated from holdout routing maps or by simulating the index/sampling process on historical posteriors. Finally, the bounty $B$ is a one-time payment that triggers on the same event as graduation, which makes it administratively straightforward and easy to explain in creator-facing documentation. 

This mapping matters in two ways. First, it makes the instruments testable and governable: a platform can run policy pilots that vary $q$ and $B$ while holding $s$ fixed near the empirical margin where $P'(\mu)$ is steep, and then read off the shift in pre-entry effort from creator-side process data (e.g., production time, revision counts, asset quality). Second, it connects the theory to budgeting: the attention cost of $q$ is measured in guaranteed impressions, the cash cost of $B$ is measured in expected payouts $B P(\mu^\star)$, and both can be allocated within the same planning framework as ad inventory and incentive programs. In Section~\ref{sec:mechanisms_main} we formalize these trade-offs through a simple resource-constrained program that delivers a clear “equalize marginal value per unit resource” rule.

\paragraph{Summary and link to the main results.}
The model developed in this section delivers a compact but expressive representation of supply side incentives under platform control. The key analytical object is the binomial tail $P(\mu)$ and, especially, its slope $P'(\mu)$, which pin down how much a marginal improvement in quality changes the odds of discovery. Guaranteeing a modest testing window $q$ makes effort pay even without graduation; concentrating payments on the graduation event $B$ buys marginal incentives exactly where they matter; tuning the bar $\bar{\mu}$ determines where the slope is steep; and the continuation value $H$ sets the stakes of selection. Under weak regularity, the entrant’s best response is unique and responds monotonically to these levers. The next section exploits this structure to deliver comparative statics results, a constructive bounty that implements the planner’s first best, and a resource-aware rule for balancing exposure and cash when both are scarce.

\section{Main Results}\label{sec:mechanisms_main}

This section develops the main analytical results that connect the primitives introduced in Section~\ref{sec:model} to implementable design rules. We begin by characterizing the entrant’s best response under the policy $(q,\bar{\mu},B,H)$ and establishing existence and uniqueness of the equilibrium quality choice. We then derive comparative statics that link each managerial option to the induced shift in pre-entry investment. The results are stated and interpreted in the body; full proofs and technical lemmas are deferred to the appendix.

\subsection{Creator Best Response}\label{subsec:best_response}

Given the testing window $q$, bar $\bar{\mu}$ (equivalently the integer threshold $s=\lceil q\bar{\mu}\rceil$), bounty $B$, and continuation value $H$, the entrant chooses $\mu$ to maximize the discounted payoff in \eqref{eq:creator_objective}. Differentiating and applying the Beta-density identity in \eqref{eq:Pprime} yields the first-order condition for an interior optimum,
\begin{equation}\label{eq:foc_main}
\alpha\Big[q + H P(\mu^\star)\Big]\;+\;\alpha\,\mu^\star\,H\,P'(\mu^\star)\;+\;B\,P'(\mu^\star)\;=\;c'(\mu^\star),
\end{equation}
where $P(\mu)=\sum_{k=s}^{q}\binom{q}{k}\mu^k(1-\mu)^{q-k}$ and $P'(\mu)=\mu^{s-1}(1-\mu)^{q-s}/B(s,q-s+1)$ for $1\le s\le q$. The left-hand side displays two economically distinct channels through which the platform’s policy acts. The \emph{level term}, $\alpha[q+H P(\mu)]$, captures the certain return from guaranteed testing and the expected return from graduation at the prevailing probability. The \emph{slope term}, $\alpha \mu H P'(\mu)+ B P'(\mu)$, captures how a marginal improvement in quality changes the odds of graduation and therefore the expected continuation exposure and transfer.\footnote{Throughout we assume outcomes are platform-measured and not manipulable; this keeps the slope term focused on real improvements in content rather than reporting.}

To formalize well-posedness, recall Assumption~1 in Section~\ref{sec:model}, which imposes strict convexity of $c$ and a mild single-crossing condition ensuring that the marginal cost dominates any local curvature of the benefit term. Define the gap function
\[
\Delta(\mu; q,\bar{\mu},B,H,\alpha)\;\equiv\;c'(\mu)\;-\;\Big(\alpha[q+H P(\mu)]\;+\;\alpha \mu H P'(\mu)\;+\;B P'(\mu)\Big).
\]
Under Assumption~1, $\Delta(\cdot)$ is strictly increasing on $[\underline{\mu},\overline{\mu}]$, and therefore crosses zero at most once.

\begin{proposition}[Existence and uniqueness]\label{prop:exist_unique}
Suppose Assumption~1 holds. Then there exists a unique solution $\mu^\star\in[\underline{\mu},\overline{\mu}]$ to the optimality condition \eqref{eq:foc_main}. If $\Delta(\underline{\mu})>0$ the solution is $\mu^\star=\underline{\mu}$; if $\Delta(\overline{\mu})<0$ the solution is $\mu^\star=\overline{\mu}$; otherwise $\mu^\star$ is the unique interior point satisfying $\Delta(\mu^\star)=0$.
\end{proposition}

The proposition highlights a simple but important message. Because the testing window $q$ and the bounty $B$ shift the right-hand side of \eqref{eq:foc_main} upward at every $\mu$, they do not merely increase effort for the “almost-graduating”; they also protect against corner solutions at the low end by ensuring that even modest quality is monetized with certainty in testing. Conversely, if $q$ is negligible and the graduation rule is either too forgiving or too stringent, the slope term $P'(\mu)$ is flat where most entrants reside, and the unique optimum may sit near the lower bound. In practice this manifests as under-investment by new creators who rationally anticipate that their effort will not be observed at a margin where it affects allocation.

Two boundary cases illustrate the roles of the policy levers. When $B=0$ and $H=0$, the right-hand side of \eqref{eq:foc_main} collapses to the constant $\alpha q$, and the optimal $\mu^\star$ is pinned down entirely by the trade-off between a guaranteed exposure stipend and convex effort costs. When $q=0$ but the continuation value $H$ is large and the bar is calibrated so that $P'(\mu)$ is steep near the mass of entrants, the slope term dominates and a small change in $\mu$ produces a large swing in graduation odds; in that case the platform is relying purely on selection to create incentives, which may be efficient if the feed is not capacity-constrained but is brittle if early outcomes are noisy. The general policy in \eqref{eq:foc_main} allows the platform to balance the two mechanisms: use $q$ to ensure that effort is not wasted, and use $(H,\bar{\mu})$, optionally reinforced by $B$, to make the graduation margin decisive where it should be.

A simple parametric example helps fix ideas. Take $c(\mu)=\tfrac{\kappa}{2}(\mu-\underline{\mu})^2$ with $\kappa>0$. Then the left-hand side of \eqref{eq:foc_main} is continuous and strictly increasing in $\mu$, while the right-hand side is linear in $\mu$. The intersection exists and is unique, and comparative statics are transparent: increasing $q$, $H$, $\alpha$, or $B$ shifts the benefit schedule up and therefore increases $\mu^\star$. The role of the bar is to re-locate the peak of $P'(\mu)$; moving $s$ so that the implied $(s-1)/(q-1)$ aligns with the mass of entrants steepens the slope term at the relevant margin and lowers the required increase in $q$ or $B$ needed to attain a given target quality.

\subsection{Comparative Statics}\label{subsec:comparative_statics}

We now establish that the entrant’s equilibrium quality is monotone in each parameter under weak conditions. The result formalizes the intuition that guaranteed exposure increases the base return to effort, that a more valuable continuation raises the stakes of being selected, that a higher revenue share lets creators internalize more of the surplus, and that a discovery bounty adds targeted marginal incentives at the graduation threshold.

\begin{theorem}[Monotone comparative statics]\label{thm:mcs}
Suppose Assumption~1 holds. Then the unique best response $\mu^\star(q,\bar{\mu},B,H,\alpha)$ is weakly increasing in $q$, $H$, $B$, and $\alpha$. In particular,
\[
\frac{\partial \mu^\star}{\partial q}\;\ge\;0,\qquad 
\frac{\partial \mu^\star}{\partial H}\;\ge\;0,\qquad 
\frac{\partial \mu^\star}{\partial B}\;\ge\;0,\qquad
\frac{\partial \mu^\star}{\partial \alpha}\;\ge\;0,
\]
whenever the corresponding partial derivatives exist; at corners the monotone dependence holds in the sense of set-valued supermodular comparative statics.
\end{theorem}

The proof proceeds by differentiating the gap condition $\Delta(\mu^\star)=0$ and invoking the implicit function theorem. Under Assumption~1, $\partial \Delta/\partial \mu>0$, while $\partial \Delta/\partial q=-\alpha<0$, $\partial \Delta/\partial H=-(\alpha P(\mu)+\alpha \mu P'(\mu))<0$, $\partial \Delta/\partial B=-P'(\mu)<0$, and $\partial \Delta/\partial \alpha=-(q+H P(\mu)+\mu H P'(\mu))<0$. The signs imply that an increase in any of the policy levers requires a compensating increase in $\mu^\star$ to restore the equality. Economically, the theorem says that each lever rotates the private return to effort upward, albeit through different channels, and therefore induces greater pre-entry investment.

\begin{figure}[!ht]
  \centering
  \includegraphics[width=.72\linewidth]{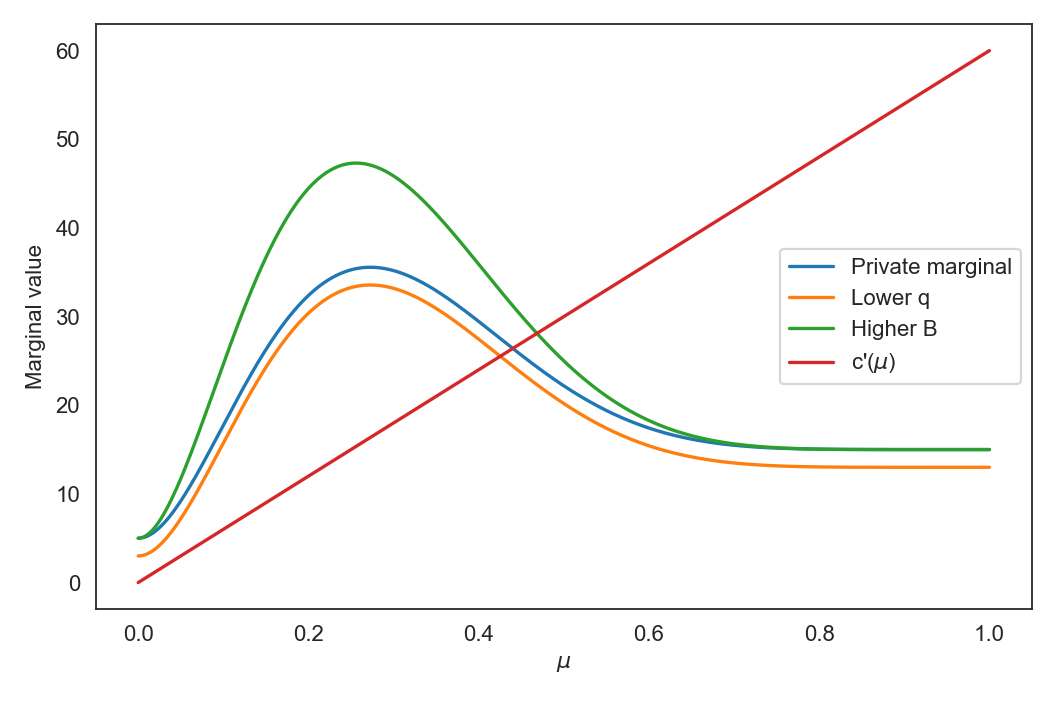}
  \caption{\textbf{Best-response crossing and comparative statics.}
  The creator’s private marginal benefit crosses $c'(\mu)$ at $\mu^\star$. Lowering $q$ shifts the level term down; raising $B$ adds slope weight via $P'(\mu)$, moving the crossing to the right.\\
\footnotesize\emph{Primitives used:} $q{=}10$, $s{=}3$, $\alpha{=}0.5$, $H_0{=}0$, $\Delta H{=}20$, $\kappa{=}60$, $\mu_0{=}0$.}
  \label{fig:sec3-crossing}
\end{figure}

The qualitative differences between the parameters are managerially important. An increase in $q$ is a broad-based stipend that lifts all boats; it is most effective when the mass of entrants is located in a region where $P'(\mu)$ is already nontrivial, because the induced rightward move in $\mu^\star$ then activates the slope term and creates a positive feedback. By contrast, an increase in $B$ is surgical: it buys marginal incentives at the bar and is therefore most efficient when $P'(\mu)$ is steep, precisely the configuration created by setting $\bar{\mu}$ near the empirical margin where many entrants just pass or just fail. The continuation value $H$ combines both effects; it increases the expected value of graduation and magnifies the slope term. Finally, a higher revenue share $\alpha$ substitutes for $B$ in the sense that it strengthens both the level and slope channels without any transfers contingent on graduation; this substitution will be made precise in Section~\ref{subsec:implement_fb} when we derive the bounty that implements the planner’s first best.

\subsection{Discovery Bounty}\label{subsec:implement_fb}

The analysis so far has treated the platform’s policy as shaping the private return to quality. It is useful to benchmark these incentives against a normative target in which the platform could directly dictate the entrant’s investment level. Let the planner’s objective aggregate the platform’s discounted engagement from the entrant and the entrant’s effort cost, abstracting from transfers:
\begin{equation}\label{eq:planner_problem}
\max_{\mu\in[\underline{\mu},\overline{\mu}]}\ \ \mu\,\big[q+H P(\mu)\big]\;-\;c(\mu).
\end{equation}
The first-order condition for an interior optimum, which we denote $\mu^{FB}$, is
\begin{equation}\label{eq:planner_foc}
q+H P(\mu^{FB})\;+\;\mu^{FB} H P'(\mu^{FB})\;=\;c'(\mu^{FB}),
\end{equation}
where $P(\mu)$ and $P'(\mu)$ are given by \eqref{eq:Pmu}-\eqref{eq:Pprime}. Equation \eqref{eq:planner_foc} says that the marginal social benefit of increasing quality, comprising the guaranteed impressions, the expected continuation exposure, and the change in graduation odds times the value of graduation, should equal the marginal effort cost. Because the planner internalizes both the certain and selection-driven components of exposure, \eqref{eq:planner_foc} has the same structure as the private condition \eqref{eq:foc_main} but, crucially, without the revenue-share distortion and without any need for transfers.

The key question is whether the first best can be achieved by a policy that respects the informational and operational constraints of the problem; in particular, whether a simple outcome-contingent transfer can align the entrant’s private marginal return with \eqref{eq:planner_foc}. The next result shows that a one-line bounty does the job.

\begin{theorem}[First-best implementability]\label{thm:fb_implement}
Fix $(q,\bar{\mu},H,\alpha)$ and let $\mu^{FB}$ satisfy \eqref{eq:planner_foc}. Define the \emph{discovery bounty}
\begin{equation}\label{eq:Bstar}
B^\star \;\equiv\; \frac{\big[q+H P(\mu^{FB})+\mu^{FB} H P'(\mu^{FB})\big]\,(1-\alpha)}{P'(\mu^{FB})}\;\;\ge 0.
\end{equation}
Under the policy $(q,\bar{\mu},B^\star,H)$, the entrant’s unique best response is $\mu^\star=\mu^{FB}$.
\end{theorem}

\noindent\emph{Proof (sketch).} Evaluate the private optimality condition \eqref{eq:foc_main} at $\mu^{FB}$ and substitute $B^\star$ from \eqref{eq:Bstar}. The left-hand side becomes
\[
\alpha\big[q+H P(\mu^{FB})\big]\;+\;\alpha \mu^{FB} H P'(\mu^{FB})\;+\;(1-\alpha)\big[q+H P(\mu^{FB})+\mu^{FB}H P'(\mu^{FB})\big],
\]
which simplifies to $q+H P(\mu^{FB})+\mu^{FB} H P'(\mu^{FB})$. By \eqref{eq:planner_foc} this equals $c'(\mu^{FB})$, so $\mu^{FB}$ satisfies the private FOC. Assumption~1 implies the private objective is strictly quasi-concave in $\mu$; hence the solution is unique. \qed

Figure~\ref{fig:sec3-implement} illustrates the result of this theorem. The planner’s marginal value
coincides with the creator’s once we add the bounty $B^\star$ derived in the text. The overlay
is exact when $P'(\mu^{FB})>0$; If the bar is in a position where the tail is flat, the denominator vanishes
and no finite bounty can implement the planner. For exposition, the figure uses an interior
configuration with $P'(\mu^{FB})>0$; the formula applies in general.

\begin{figure}[!ht]
  \centering
  \includegraphics[width=.72\linewidth]{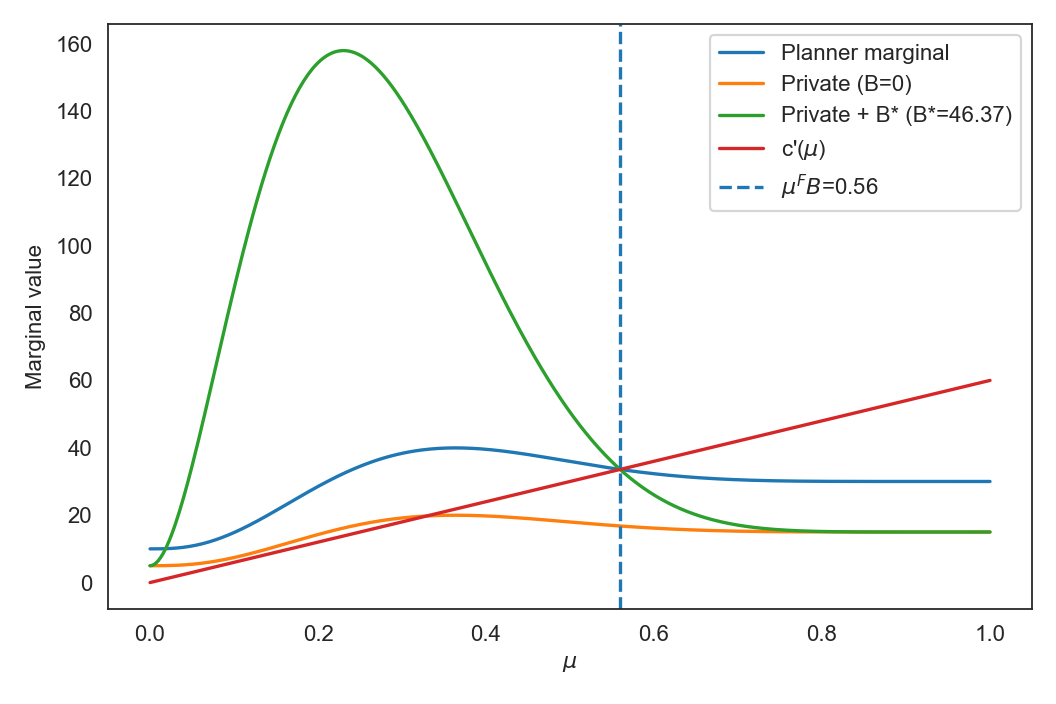}
\caption{\textbf{Implementability bounty $B^\star$.}
Adding $B^\star$ makes the private marginal coincide with the planner marginal at $\mu^{FB}$.\\
\footnotesize\emph{Primitives used:} $q{=}10$, $s{=}3$, $\alpha{=}0.5$, $H_0{=}0$, $\Delta H{=}20$, $\kappa{=}60$, $\mu_0{=}0$.}
  \label{fig:sec3-implement}
\end{figure}

The bounty in \eqref{eq:Bstar} has a transparent interpretation. The factor $(1-\alpha)$ is the wedge between the social and private scales of engagement: when the entrant already receives a high revenue share, transfers are less necessary. The numerator in brackets is the planner’s marginal value of quality at $\mu^{FB}$, the right-hand side of \eqref{eq:planner_foc}. The division by $P'(\mu^{FB})$ concentrates the top-up on the informative event that the entrant is at the margin of graduation. In other words, $B^\star$ pays only for outcomes that are most diagnostic of true quality; it purchases a large increase in the probability of discovery per dollar spent because it targets the steep part of the selection frontier. This is precisely the reason to pair the bounty with a visible bar: the platform commits to paying when the bar is cleared, and because the slope $P'(\cdot)$ is largest near the bar, the same dollar is more potent there than anywhere else.

Three comparative statics insights follow immediately from \eqref{eq:Bstar}. First, $B^\star$ is weakly decreasing in $\alpha$ and vanishes as $\alpha\to 1$. When creators already internalize most of the engagement surplus, guaranteed exposure and a well-calibrated graduation rule suffice to implement the first best; cash is redundant. Second, $B^\star$ is inversely related to $P'(\mu^{FB})$. If the bar $\bar{\mu}$ is tuned so that typical entrants sit near the steep region of the Beta density, very small bounties deliver large incentive effects; if the bar is too low or too high, $P'$ is flat and the same target requires a large transfer. This logic provides an operational guide for setting $\bar{\mu}$: target the margin where the pass rate is neither too low nor too high so that the incentive leverage of $B$ is maximal. Third, holding the bar fixed, $B^\star$ increases with $H$ and $q$ through the bracketed term, reflecting the fact that when graduation is more valuable and guaranteed testing is more substantial, the planner’s marginal value of quality is larger and thus requires a larger private top-up when $\alpha<1$.

It is also useful to compute the expected payout at the first best. Because the bounty is paid only upon graduation,
\begin{equation}\label{eq:expected_payout}
\mathbb{E}\big[\text{payout}\,\big|\mu^{FB}\big]\;=\;B^\star\,P(\mu^{FB})\;=\;\frac{\big[q+H P(\mu^{FB})+\mu^{FB} H P'(\mu^{FB})\big]\,(1-\alpha)\,P(\mu^{FB})}{P'(\mu^{FB})}.
\end{equation}
Expression \eqref{eq:expected_payout} makes transparent the budget implication of bounty targeting: for a fixed normative target $\mu^{FB}$, the expected transfer is decreasing in $P'(\mu^{FB})/P(\mu^{FB})$, the local hazard-type ratio of the binomial tail. Policies that put the bar where the tail is steep, so that a small increase in quality flips many entrants from failing to passing, achieve the same incentive alignment with lower expected cash spend. This highlights the substantive implication: do not just set a bar, set it where the diagnosticity of the pass event is highest relative to its frequency.

The implementability result bears directly on instrument choice. A natural alternative to a hit-based bounty is a flat per-view subsidy during testing. Such a subsidy increases the level term in \eqref{eq:foc_main} but is insensitive to whether the entrant is near the graduation margin. When $P'(\mu)$ is steep around $\mu^{FB}$, \eqref{eq:Bstar}-\eqref{eq:expected_payout} imply that concentrating dollars on the pass event buys more marginal $\mu$ per unit of expected transfer than spreading dollars uniformly across testing impressions. This targeting advantage is the supply side analogue of the “pay-for-information” principle in exploration design: incentives should be strongest where the signal is most responsive to effort. In implementation, the advantage manifests as lower spend to reach the same pre-entry quality target, provided the pass event is measured reliably and the bar is communicated clearly.

Finally, the simplicity of \eqref{eq:Bstar} is an asset for governance. The platform need not estimate or publish counterfactuals about what would have happened had the entrant not been tested; it need only commit ex ante to $(q,\bar{\mu},H)$ and to a bounty $B$ that is a deterministic function of those primitives. The payment rule is explainable, auditable, and easy to parameterize in budgeting terms: expected impression cost is $q$, expected cash cost is $B P(\mu^\star)$, and both map directly into standard planning dashboards. In Section~\ref{sec:resources_bwk} we will fold these costs into a resource-constrained program and show how to balance exposure and cash by equating marginal welfare per unit resource.
\subsection{Scheduling of Exploration}\label{subsec:frontloading}

The preceding results treat the testing window as a discounted exposure mass $q$, abstracting from its exact calendar placement. In practice, product teams schedule a \emph{raw} number of testing impressions and control when those impressions occur relative to entry. It is therefore useful to formalize the timing choice and show that, for any fixed number of testing impressions, the earliest feasible pacing weakly dominates alternatives in terms of investment incentives and induced welfare.

Let $Q\in\mathbb{N}$ denote the (undiscounted) number of guaranteed testing impressions the platform commits to deliver to a new entrant, and let $\tau=(t_1,\ldots,t_Q)$ be a strictly increasing sequence of calendar slots at which these impressions will be delivered, with $t_j\in\{1,2,\ldots\}$. Define the \emph{discounted testing mass} induced by schedule $\tau$ as
\[
q(\tau)\;=\;\sum_{j=1}^{Q}\gamma^{t_j-1}.
\]
During testing, the graduation rule compares the realized number of successes $S$ to the integer threshold $s=\lceil Q\bar{\mu}\rceil$; the graduation probability and its slope are therefore
\[
P_Q(\mu)\;=\;\Pr\!\big[S\ge s\;\big|\;S\sim\mathrm{Binomial}(Q,\mu)\big],
\qquad
P'_Q(\mu)\;=\;\frac{\mu^{\,s-1}(1-\mu)^{\,Q-s}}{B\!\left(s,\,Q-s+1\right)}.
\]
Crucially, because outcomes are i.i.d.\ conditional on $\mu$, $P_Q(\cdot)$ and $P'_Q(\cdot)$ depend only on the \emph{count} $Q$ and the bar $s$, not on the specific calendar placement of testing impressions. The entrant’s discounted exposure under schedule $\tau$ is therefore
\[
\Xi(\mu;\tau,H)\;=\;q(\tau)\;+\;H\,P_Q(\mu),
\]
and the private first-order condition becomes
\begin{equation}\label{eq:foc_schedule}
\alpha\Big[q(\tau) + H P_Q(\mu^\star)\Big]\;+\;\alpha\,\mu^\star\,H\,P'_Q(\mu^\star)\;+\;B\,P'_Q(\mu^\star)
\;=\;c'(\mu^\star).
\end{equation}

We compare schedules by the induced discounted testing mass $q(\tau)$. Write $\tau^\uparrow$ for the \emph{earliest feasible} schedule that places the $Q$ testing impressions in the first $Q$ available slots (so $t_j=j$ in the absence of capacity constraints). Because $\gamma\in(0,1)$, $\tau^\uparrow$ maximizes $q(\tau)$ among all schedules with $Q$ impressions; any delay replaces an early factor $\gamma^{t_j-1}$ with a strictly smaller $\gamma^{t_j}$.

\begin{theorem}[Front-loading is weakly optimal]\label{thm:frontload}
Fix $Q$, $\bar{\mu}$ (equivalently $s=\lceil Q\bar{\mu}\rceil$), $H$, $B$, and $\alpha$. Among all testing schedules $\tau$ with $Q$ impressions, the earliest schedule $\tau^\uparrow$ maximizes the entrant’s equilibrium quality $\mu^\star(\tau)$ and the platform’s equilibrium objective $W(\mu^\star(\tau);\tau,H,B)$.\footnote{If per-period feasibility caps restrict testing to at most one guaranteed impression per period, $\tau^\uparrow$ is defined as “as early as allowed by the cap.” The argument below is unchanged.}
\end{theorem}

\noindent\emph{Proof (sketch).} For any two schedules $\tau$ and $\tilde{\tau}$ with the same $Q$, $P_Q(\cdot)$ and $P'_Q(\cdot)$ coincide. The only difference in \eqref{eq:foc_schedule} is the term $\alpha q(\tau)$ versus $\alpha q(\tilde{\tau})$. Since $q(\tau^\uparrow)\ge q(\tau)$ for all $\tau$, the left-hand side of \eqref{eq:foc_schedule} evaluated at $\tau^\uparrow$ is weakly larger at every $\mu$. Under Assumption~1, the gap function is strictly increasing in $\mu$, so the unique solution $\mu^\star(\tau)$ is weakly increasing in $q(\tau)$; hence $\mu^\star(\tau^\uparrow)\ge \mu^\star(\tau)$. Because $W(\mu;\tau,H,B)=\mu[q(\tau)+H P_Q(\mu)]-B P_Q(\mu)$ is increasing in both $\mu$ and $q(\tau)$, it follows that $W(\mu^\star(\tau^\uparrow);\tau^\uparrow,H,B)\ge W(\mu^\star(\tau);\tau,H,B)$. \qed

The theorem formalizes an intuitive but important operational lesson. When a platform has fixed the \emph{number} of new-creator opportunities $Q$ it can afford, the way to turn those opportunities into the strongest possible pre-entry incentives is to deploy them as close to entry as feasible. Early deployments raise the certain component of the return to effort without affecting the diagnosticity of the graduation event, which depends only on how many early outcomes are observed, not on when. This “as-early-as-possible” guidance is straightforward to implement: in experimentation platforms it maps to pacing rules that consume the guaranteed impressions budget at the start of the eligible horizon, subject to safety caps and eligibility filters.

Two clarifications smooth the path from theory to practice. First, if policy is parameterized directly by the discounted testing mass $q$ rather than by the raw count $Q$, then \emph{by definition} different schedules that deliver the same $q$ are incentive-equivalent in \eqref{eq:foc_schedule}; in that case the front-loading recommendation is already baked into the calibration of $q$. The theorem should then be read as a reminder that, when teams manage raw impression counts, the effective $q$ is \emph{endogenous} to timing and is maximized by front-loading. Second, if match quality drifts over calendar time (for example, if audiences expand or contract over a daypart), front-loading remains weakly optimal when such drift affects both the certain and selection components symmetrically through the same discount factor, for instance, when the planner’s objective discounts future outcomes for opportunity-cost reasons. In environments with novelty decay or attention scarcity that is stronger later in the session, the case for early pacing is, if anything, stricter: the same number of testing impressions is worth more, both in certainty and in selection potency, when concentrated near entry.

From a budgeting standpoint, Theorem~\ref{thm:frontload} also rationalizes a common governance practice: reserving a small, protected pool of early slots for each entrant that cannot be traded away for short-term revenue. Because those slots convert directly into a higher effective $q(\tau)$, the policy pays for itself in stronger investment incentives and, downstream, in higher discovery of genuinely high-quality creators. The next section takes this budgeting logic one step further by treating impressions and cash as explicit resources and deriving an equal-marginal-value rule for balancing $q$ (delivered early) against the bounty $B$ when both are scarce.

\section{Resource-Constrained Design}\label{sec:resources_bwk}

The policy instruments introduced so far, guaranteed testing impressions and discovery bounties, consume tangible resources that platforms budget in parallel with experimentation slots and incentive programs. This section treats impressions and cash as explicit constraints and derives a simple “equalize marginal value per unit resource” rule that governs how much to invest in early exposure versus targeted transfers. The analysis uses a static planning horizon both for transparency and because the online, sequential version of the problem maps to the same shadow-price logic under standard primal-dual arguments.

\subsection{Two Resources and a Constrained Program}\label{subsec:two_resources}

\begin{figure}[!ht]
  \centering
  \includegraphics[width=.65\linewidth]{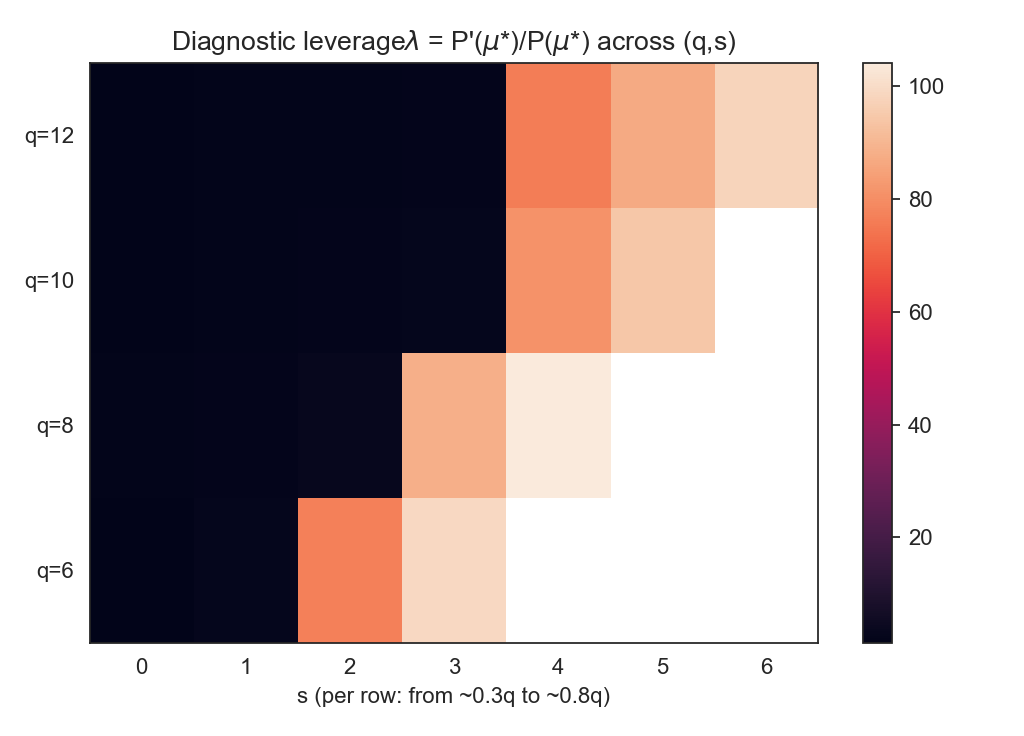}
  \caption{\textbf{Diagnostic leverage $\Lambda(\mu^\star)=P'(\mu^\star)/P(\mu^\star)$ across bar/window pairs.}
  Warm cells indicate settings where the graduation margin is steep \emph{at the induced equilibrium}; these are the most effective regions for bounty targeting and for allocating scarce attention.\\
\footnotesize\emph{Primitives used:} $q{=}10$, $s{=}3$, $\alpha{=}0.5$, $H_0{=}0$, $\Delta H{=}20$, $\kappa{=}60$, $\mu_0{=}0$.}
  \label{fig:sec4-heatmap}
\end{figure}

Consider a planning window (e.g., a week) during which a large cohort of new creators enters. For ease of exposition, suppose the platform adopts a stationary policy that assigns the same testing window $q$ and bounty $B$ to each entrant, together with a common graduation rule summarized by $(\bar{\mu},H)$. The platform has an \emph{impression budget} $R>0$ that caps the total number of \emph{discounted} guaranteed testing impressions it is willing to allocate to the cohort,\footnote{In practice, $R$ is the number of early, high-quality impressions product leadership reserves for incubation across the window, converted to discounted units by applying the same $\gamma$ used in Section~\ref{sec:model}. The front-loading result in Section~\ref{subsec:frontloading} implies that managers should pace these impressions as early as feasible within each entrant’s eligibility horizon.} and a \emph{cash budget} $M>0$ that caps the total expected payout on discovery bounties. If $n$ entrants arrive in the window, feasibility requires $n q \le R$ and $n\,B\,P(\mu^\star(q,\bar{\mu},B,H,\alpha)) \le M$, where $\mu^\star(\cdot)$ is the unique best response characterized in Section~\ref{sec:mechanisms_main}. 

Because $n$ scales all terms proportionally, it is convenient to write the problem in per-entrant units and regard $R$ and $M$ as \emph{budgets per entrant}. Using the platform objective \eqref{eq:platform_objective}, the per-entrant design problem is
\begin{equation}\label{eq:perentrant_program}
\max_{q\ge 0,\;B\ge 0}\;\; W\!\left(\mu^\star(q,\bar{\mu},B,H,\alpha);\;q,\bar{\mu},B,H\right)
\quad
\text{s.t.}\quad 
q \le R,\ \ B\,P\!\left(\mu^\star(q,\bar{\mu},B,H,\alpha)\right) \le M.
\end{equation}
The first constraint allocates a scarce attention resource; the second allocates expected cash. Both are \emph{endogenous} to incentives: changing $q$ or $B$ shifts $\mu^\star$, which in turn changes the graduation probability and therefore the expected payout. The endogeneity is the economic heart of the trade-off. Spending attention early strengthens private incentives and reduces the amount of cash needed later to align behavior; conversely, a well designed bounty can economize impressions by increasing effort exactly at the margin where selection is sensitive.

Program \eqref{eq:perentrant_program} is separable across entrants under the stationarity assumption, which means that the optimal \emph{per-entrant} policy coincides with the policy one would obtain by solving a large-scale allocation across heterogeneous entrants with the same shadow prices. This observation will be useful for implementation: it legitimizes thinking in terms of a single pair $(q,B)$ tuned to budgets $R$ and $M$, with the understanding that heterogeneity can be incorporated by allowing the pair to vary across segments while keeping the same shadow-price tests.\footnote{Formally, with heterogeneous segments $s$ and segment shares $\pi_s$, the planner chooses $(q_s,B_s)$ to maximize $\sum_s \pi_s W(\mu^\star_s;q_s,\bar{\mu}_s,B_s,H_s)$ subject to $\sum_s \pi_s q_s \le R$ and $\sum_s \pi_s B_s P_s(\mu^\star_s)\le M$. The stationarity case corresponds to a single segment; the primal-dual characterization below carries through verbatim with segment-specific marginal values.}

Two preliminary facts follow directly from \eqref{eq:perentrant_program}. First, \emph{budget monotonicity}: if $R$ and $M$ increase (weakly), the optimal pair $(q^\dagger,B^\dagger)$ weakly increases componentwise. Intuitively, loosening either constraint cannot make it optimal to cut back on the corresponding instrument because both $q$ and $B$ raise incentives and therefore raise the platform’s objective at the induced equilibrium. Second, when $\alpha$ is high and the revenue share already lets creators internalize most of the engagement surplus, the cash budget is less binding: the optimal bounty is small or zero and the platform allocates most of the scarce resource to early impressions. When $\alpha$ is low, a modest $B$ is a cost-effective substitute for $q$ because it targets the steep region of the graduation frontier where $P'(\mu)$ is largest.

\subsection{A Pimal-dual Characterization and the Balanced Exploration Rule}\label{subsec:balanced_rule}

To make the trade-off operational, we pass to a Lagrangian representation with \emph{shadow prices} for impressions and cash. Let $\lambda_{\mathrm{imp}}\ge 0$ be the multiplier (per discounted impression) on the testing budget and let $\lambda_{\$}\ge 0$ be the multiplier (per expected dollar of payout) on the cash budget. The Lagrangian for \eqref{eq:perentrant_program} is
\begin{equation}\label{eq:lagrangian}
\mathcal{L}(q,B;\lambda_{\mathrm{imp}},\lambda_{\$})
\;=\;
W\!\left(\mu^\star(q,\bar{\mu},B,H,\alpha);\;q,\bar{\mu},B,H\right)
\;-\;
\lambda_{\mathrm{imp}}\,(q-R)
\;-\;
\lambda_{\$}\,\Big(B\,P\!\left(\mu^\star(q,\bar{\mu},B,H,\alpha)\right)-M\Big).
\end{equation}
At an interior optimum the first-order conditions set the marginal gain in constrained welfare from an extra discounted testing impression equal to its shadow price, and analogously for an extra expected dollar of bounty payout. Writing $\mathrm{MB}_q$ and $\mathrm{MB}_{\$}$ for these marginal gains, we obtain
\begin{equation}\label{eq:equal_marginal}
\mathrm{MB}_q(q^\dagger,B^\dagger)\;=\;\lambda_{\mathrm{imp}},
\qquad
\mathrm{MB}_{\$}(q^\dagger,B^\dagger)\;=\;\lambda_{\$}.
\end{equation}
The exact algebraic expressions for $\mathrm{MB}_q$ and $\mathrm{MB}_{\$}$ collect both \emph{direct} effects (e.g., one more early impression directly raises expected engagement by $\mu^\star$) and \emph{indirect} effects that operate through the induced change in $\mu^\star$. The indirect terms depend on the planner’s marginal value of quality at the induced equilibrium, $q+H P(\mu^\star)+\mu^\star H P'(\mu^\star)-B P'(\mu^\star)$, which equals zero if the bounty is set to the implementability level $B^\star$ in Theorem~\ref{thm:fb_implement}. In that edge case, the marginal values reduce to familiar objects: the value of one more discounted early impression is $\mu^{FB}$, while the value of one more expected dollar of payout is $-1$ (because transfers are a pure resource cost). Away from the implementability level, the indirect terms tilt the trade-off toward the instrument that moves $\mu^\star$ the most per unit of resource at the prevailing slope $P'(\mu^\star)$.

\begin{figure}[!ht]
  \centering
  \begin{subfigure}{.48\linewidth}
    \centering
    \includegraphics[width=\linewidth]{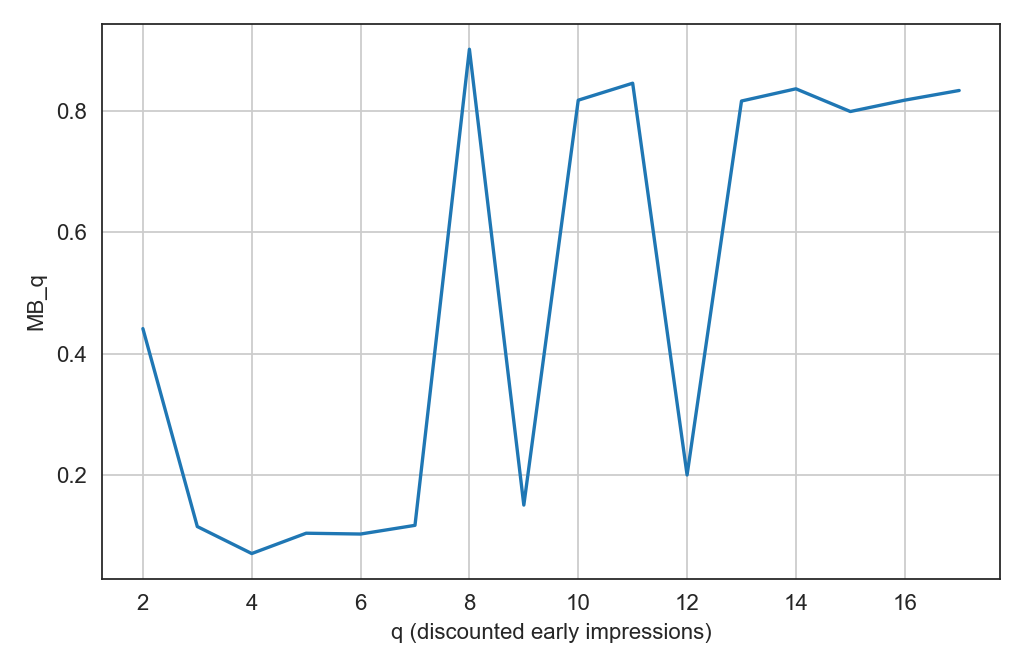}
    \caption{$\mathrm{MB}_q$ vs $q$}
  \end{subfigure}\hfill
  \begin{subfigure}{.48\linewidth}
    \centering
    \includegraphics[width=\linewidth]{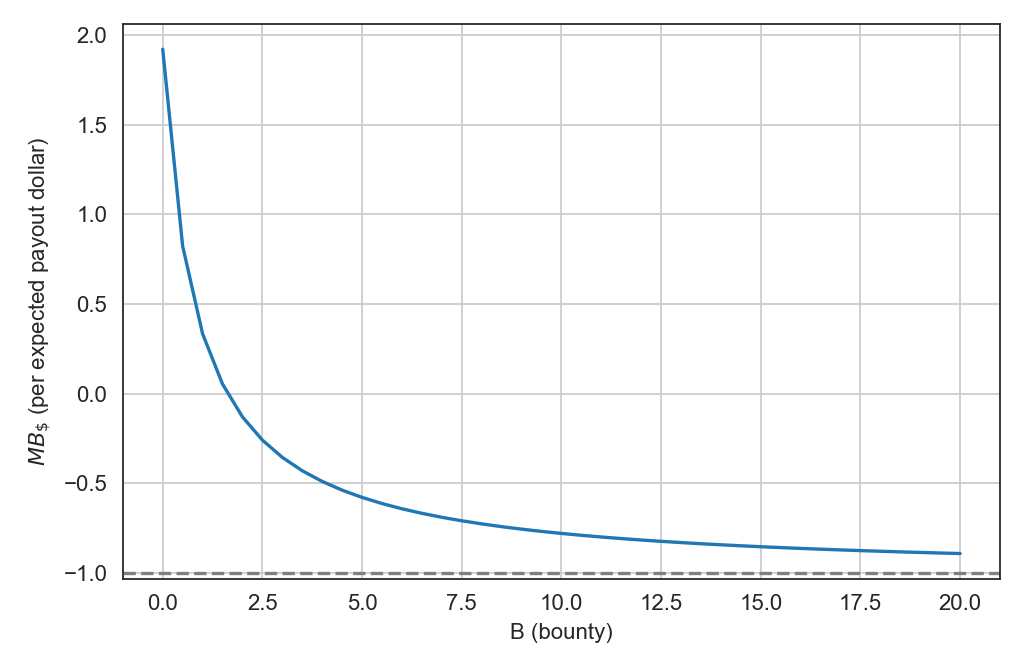}
    \caption{$\mathrm{MB}_{\$}$ vs $B$ (baseline)}
  \end{subfigure}
  \caption{\textbf{Equal–marginal–value objects.}
  (a) The marginal gain per discounted early impression mixes a direct term ($\mu^\star$) with an induced-effort term. Kinks occur where the integer threshold $s(q) = \bar{\mu}q$ increments, reflecting the discrete change in the binomial tail. (b) The marginal gain \emph{per expected payout dollar} is decreasing in $B$ and converges to $-1$; when $P(\mu^\star)$ is tiny, the per-dollar normalization can be numerically unstable.\\
\footnotesize\emph{Primitives used:} $q{=}10$, $s{=}3$, $\alpha{=}0.5$, $H_0{=}0$, $\Delta H{=}20$, $\kappa{=}60$, $\mu_0{=}0$.}
  \label{fig:sec4-MB}
\end{figure}

Condition \eqref{eq:equal_marginal} is the \emph{balanced exploration rule}. It prescribes increasing $q$ up to the point where the next discounted early impression is worth exactly $\lambda_{\mathrm{imp}}$ in constrained welfare, and increasing the bounty up to the point where the next expected dollar of payout is worth exactly $\lambda_{\$}$. In practice, the multipliers are not chosen directly; they are revealed by budgets. If the impression budget is tight relative to cash, $\lambda_{\mathrm{imp}}$ is high and the optimizer leans on $B$ to buy incentives at the graduation margin; if cash is scarce relative to early slots, $\lambda_{\$}$ is high and the optimizer leans on $q$ to raise the base return to quality. Either way, the prescription is implementable with dashboards managers already maintain: plot the estimated marginal lift in constrained welfare per additional early impression and per expected dollar of payout as functions of $(q,B)$, and adjust until the two curves hit their respective shadow-price lines.

Two further clarifications help with execution. First, the rule speaks in \emph{per-unit} terms and therefore scales to segmentation without change. If the platform maintains segment-specific bars $(\bar{\mu}_s)$ or continuation values $(H_s)$, the same shadow prices $\lambda_{\mathrm{imp}}$ and $\lambda_{\$}$ apply across segments, and the allocation of impressions and bounties should be skewed toward segments with the highest marginal gains per unit of resource. Second, the rule is robust to the choice of exploitation engine. If $H$ is interpreted as the expected discounted pull count under a Thompson or index policy conditional on graduation, then $q+H P(\cdot)$ remains the correct exposure aggregator, and the same Lagrangian logic yields \eqref{eq:equal_marginal}. The conceptual takeaway is unchanged: treat attention and cash as tradable budgets and push each instrument to the point where its next unit yields the same shadow-price-adjusted value as the other.

\subsection{Tuning the Graduation Bar: Selection versus Discoverability}\label{subsec:tuning_bar}

The bar $\bar{\mu}$ plays a dual role in the design problem. By shifting the integer threshold $s=\lceil q\bar{\mu}\rceil$, it governs both the \emph{location} and the \emph{steepness} of the selection frontier. On the one hand, a higher bar improves selectivity and raises the expected quality of those who graduate into the continuation stream. On the other hand, it reduces the pass rate $P(\mu)$ and, if moved too far, can place most entrants in regions where the pass event is nearly unresponsive to effort, thereby dulling incentives. The planner’s task is to position the bar so that the graduation event is (i) neither too rare nor too common and (ii) maximally \emph{diagnostic} around the intended quality margin.

Two analytic facts guide this choice. First, from \eqref{eq:Pprime} the sensitivity of the pass probability with respect to quality is a Beta density,
\[
P'(\mu)\;=\; \frac{\mu^{\,s-1}(1-\mu)^{\,q-s}}{B(s,q-s+1)},
\]
which is unimodal with mode at $\mu^\dagger=(s-1)/(q-1)$ when $1<s<q$. Thus, moving $s$ shifts the point of maximal diagnosticity along the unit interval. If the entrant distribution (or the induced $\mu^\star$ under current policy) has most mass near some $\tilde{\mu}$, the steepest incentive leverage is obtained by aligning the bar so that $\mu^\dagger\approx \tilde{\mu}$, i.e., $s\approx 1+(q-1)\tilde{\mu}$. Second, for a fixed $(q,s)$ the \emph{diagnostic leverage} of the pass event can be summarized by the ratio
\begin{equation}\label{eq:leverage_ratio}
\Lambda(\mu)\;\equiv\;\frac{P'(\mu)}{P(\mu)}.
\end{equation}
This local hazard-type ratio measures how many percentage points the pass probability moves per marginal point of quality, \emph{per unit of pass probability}. It is large when the tail is steep relative to its level (pass events are neither vanishingly rare nor ubiquitous), and small when the bar is far from the mass of the distribution.

These objects translate directly into managerial guidance. If the bar is set so low that most entrants pass ($P(\mu^\star)$ close to one), then $P'(\mu^\star)$ and $\Lambda(\mu^\star)$ are mechanically small; effort barely changes graduation odds and transfers tied to graduation do little work. If the bar is set so high that almost no one passes ($P(\mu^\star)$ near zero), $P'(\mu^\star)$ is again small relative to $P(\mu^\star)$; the event is too rare for targeted payments to be effective, and discoverability collapses. In contrast, aligning $s$ so that the pass rate sits in an intermediate region, operationally, a cohort pass rate in the neighborhood of 30-70\%, places the mode of $P'(\cdot)$ near $\mu^\star$ and maximizes $\Lambda(\mu^\star)$. This configuration is precisely the one in which a small change in effort meaningfully changes graduation odds and, downstream, expected continuation exposure.

The continuation value $H$ interacts with the bar choice in two ways. A larger $H$ increases the stakes of selection, thereby magnifying both the level term $\alpha H P(\mu)$ and the slope term $\alpha \mu H P'(\mu)$ in \eqref{eq:foc_main}. Holding budgets fixed, this raises the marginal return to moving $s$ toward the mass of the distribution because the prize for crossing the bar is richer. Conversely, if governance or fairness constraints require capping $H$ (for example, to avoid excessive concentration on a small set of graduates), the bar must do more of the incentive work: one should tilt $s$ toward regions where $P'(\cdot)$ is naturally steeper so that graduation remains decisive even when the continuation stream is modest. In either case, the practical diagnostic is the same: track the empirical pass rate and an empirical analogue of $\Lambda(\mu)$ (for instance, the elasticity of pass probability with respect to observable quality proxies) and adjust $s$ until the diagnostic leverage is high around the observed $\mu^\star$.

Finally, tuning the bar has budget consequences. From \eqref{eq:expected_payout}, the expected bounty spend at $\mu^\star$ is $B\,P(\mu^\star)$, while the per-unit incentive potency of the bounty is proportional to $P'(\mu^\star)$. Moving $s$ to increase $\Lambda(\mu^\star)=P'(\mu^\star)/P(\mu^\star)$ therefore \emph{simultaneously} strengthens incentives and reduces expected spend for a given desired uplift in the private first-order condition. The bar is thus a true “multi-plier”: appropriately placed, it lowers the cash needed to induce a given shift in effort and reduces the number of impressions needed to create the same shift via $q$.

\subsection{Targeting Advantage of Hit-based Bounties versus Flat Testing Subsidies}\label{subsec:bounty_vs_flat}

A natural alternative to a hit-based bounty is a uniform testing-stage subsidy that pays a small amount per successful testing outcome, regardless of graduation. Such a subsidy is equivalent to temporarily increasing the revenue share on testing impressions; it lifts the level term in \eqref{eq:foc_main} by a constant. We show that, locally around a given $\mu^\star$, buying a prescribed increase in marginal incentives is generically cheaper with a hit-based bounty whenever the pass event is sufficiently diagnostic. The result sharpens the budget logic behind targeted transfers and pins down the role of the leverage ratio $\Lambda(\mu^\star)$.

Formally, consider two infinitesimal policy changes designed to increase the left-hand side of \eqref{eq:foc_main} at $\mu^\star$ by a common amount $\varepsilon>0$:

\emph{(i) Hit-based increment.} Increase the bounty by $\Delta B$ and leave other primitives fixed. The marginal increase in the first-order condition is $(\Delta B)\,P'(\mu^\star)$. To achieve the target $\varepsilon$ we set
\[
\Delta B \;=\; \frac{\varepsilon}{P'(\mu^\star)}.
\]
The associated expected transfer at $\mu^\star$ is
\begin{equation}\label{eq:cost_bounty}
\text{Cost}_{\text{bounty}} \;=\; (\Delta B)\,P(\mu^\star) \;=\; \varepsilon\,\frac{P(\mu^\star)}{P'(\mu^\star)} \;=\; \frac{\varepsilon}{\Lambda(\mu^\star)}.
\end{equation}

\emph{(ii) Flat testing subsidy.} Pay a small per-success bonus $\sigma$ on testing impressions only; expected testing-stage bonuses equal $\sigma\,\mathbb{E}[S]=\sigma\,\mu q$. This adds $\sigma q$ to the left-hand side of \eqref{eq:foc_main}. To achieve the same target $\varepsilon$ we set $\sigma=\varepsilon/q$. The associated expected transfer at $\mu^\star$ is
\begin{equation}\label{eq:cost_flat}
\text{Cost}_{\text{flat}} \;=\; \sigma\,\mu^\star q \;=\; \varepsilon\,\mu^\star.
\end{equation}

Comparing \eqref{eq:cost_bounty} and \eqref{eq:cost_flat}, the hit-based instrument is locally more cost-effective than the flat testing subsidy if and only if
\begin{equation}\label{eq:dominance_condition}
\frac{\varepsilon}{\Lambda(\mu^\star)} \;\le\; \varepsilon\,\mu^\star
\quad\Longleftrightarrow\quad
\Lambda(\mu^\star) \;\ge\; \frac{1}{\mu^\star}.
\end{equation}
Condition \eqref{eq:dominance_condition} has an immediate interpretation: when the pass event is steep relative to its level, so that a unit of quality produces a large relative change in the pass probability, the targeted bounty buys more marginal incentive per expected dollar than a flat testing subsidy. Conversely, if $\Lambda(\mu^\star)$ is small (for example, because the bar is set so high that almost nobody passes, or so low that almost everybody passes), the targeted event ceases to be informative and the advantage of hit-based payments disappears. In that regime, the planner should first reposition the bar to increase $\Lambda(\mu^\star)$; only then do targeted bounties resume their advantage.

Two extensions reinforce the practical message. First, if the flat subsidy pays per \emph{testing impression} regardless of success, its effect on the first-order condition is zero and it has no incentive value at all; any dollars spent in that way should be reallocated to the bounty or to additional early impressions. Second, if the flat subsidy pays per success in both testing and continuation, its marginal effect mimics a temporary increase in the revenue share $\alpha$; while this can substitute for $B$ in \eqref{eq:foc_main}, it also dilutes targeting by paying for low-information outcomes far from the graduation margin. In budget terms, it raises expected spend by a multiple of total engaged views, not by engaged views that pivot graduation; unless continuation is very small ($H\approx 0$), the hit-based design will typically dominate on \emph{both} incentive potency and cash efficiency grounds whenever \eqref{eq:dominance_condition} holds.

Taken together, Sections~\ref{subsec:tuning_bar} and \ref{subsec:bounty_vs_flat} deliver a concrete playbook. Use early slots parsimoniously but front-load them; place the bar where the pass event is diagnostic for the observed entrant distribution; and, conditional on that placement, prefer a hit-based bounty to flat testing subsidies because it concentrates spend on the pivotal event. In the next section we examine robustness and segmentation and translate these rules into a manager-facing dashboard that reports shadow prices, diagnostic leverage, and the implied balanced exploration policy for each cohort.

\section{Extensions, Robustness, and Managerial Implications}\label{sec:extensions}

The core model isolates a small set of primitives that map directly to design levers: a guaranteed testing window, a transparent graduation bar, a continuation value summarizing downstream allocation, and a hit-based bounty that concentrates transfers on diagnostic outcomes. This section extends the analysis along two dimensions that matter in practice. First, we allow for heterogeneity across creator segments and for discovery mechanisms that graduate more than one entrant at a time. Second, we relax the stylized “block” form of post-graduation exposure and show how the results carry over to index or sampling engines in which continuation is itself history-dependent.

\subsection{Heterogeneity and Multi-winner Discovery}\label{subsec:heterogeneity}

Creator populations are rarely homogeneous. Monetization rates, production technologies, and audience match vary across surfaces and verticals, and these differences appear in the primitives $(\alpha, c, \bar{\mu}, H)$ even before policy is set. A platform that uses the same instruments across segments should expect different elasticities of investment and different budget footprints; a platform willing to tailor instruments gains additional efficiency but must still reason with a common currency for impressions and cash. The framework in Sections~\ref{sec:mechanisms_main}-\ref{sec:resources_bwk} remains valid segment-by-segment and admits a clean aggregation.

Fix a finite set of segments $s\in\mathcal{S}$ with shares $\pi_s$. Segment $s$ is characterized by $(\alpha_s, c_s, \bar{\mu}_s, H_s)$ and receives a policy $(q_s, B_s)$. Let $\mu^\star_s$ be the unique best response in $s$ under its policy, and let $P_s$ and $P'_s$ denote the binomial tail and slope induced by $(q_s,\bar{\mu}_s)$. The per-entrant platform objective in $s$ is $W_s(\mu;q_s,\bar{\mu}_s,B_s,H_s)=\mu[q_s+H_s P_s(\mu)]-B_s P_s(\mu)$, and the per-entrant resource uses are $q_s$ discounted impressions and $B_s P_s(\mu^\star_s)$ expected dollars. If the platform manages aggregate budgets $(R,M)$, the shadow-price characterization in \eqref{subsec:balanced_rule} applies with the same multipliers $(\lambda_{\mathrm{imp}},\lambda_{\$})$ across all segments. Consequently, the optimal cross-segment schedule sets $(q_s,B_s)$ where the \emph{segment-specific} marginal gains per unit of resource equal the \emph{common} shadow prices. Operationally, this means skewing early slots and bounties toward segments with the steepest $P'_s(\mu^\star_s)$ and the highest planner marginal value of quality at the induced equilibrium. The principle is simple but powerful: use a common set of prices for impressions and cash, and let segments with the strongest incentive leverage buy more of each resource until their marginal values fall to the common benchmark.

Discovery is also frequently \emph{multi-winner}. Feeds, shelves, and carousels promote a set of items at a time, and editorial programs often seek to graduate a cohort rather than a single entrant. The analysis adapts with minimal change if the multi-winner mechanism can be represented as a threshold on early performance: graduate an entrant if her testing successes exceed a score $s_K$ chosen so that, in expectation, $K$ seats are filled. For a large pool of entrants, this rule is equivalent to setting $s_K$ to solve $\mathbb{E}[\#\{i: S_i\ge s_K\}]=K$, holding fixed the distribution of peers’ qualities; the graduation probability for a focal entrant with quality $\mu$ is then
\[
P^{(K)}(\mu)\;=\;\Pr\!\big[S\ge s_K\big]\;=\;\sum_{k=s_K}^{q}\binom{q}{k}\mu^k(1-\mu)^{q-k},
\]
with slope
\[
{P^{(K)}}'(\mu)\;=\; \frac{\mu^{\,s_K-1}(1-\mu)^{\,q-s_K}}{B\!\left(s_K,\,q-s_K+1\right)}.
\]
The structure is identical to \eqref{eq:Pmu}-\eqref{eq:Pprime}, with the integer threshold now pinned down by the cohort size $K$ and the competitive environment.\footnote{When the peer distribution is heterogeneous, $s_K$ is the integer that implements the desired fill rate under that mixture; if competition is tight, the mapping from $K$ to $s_K$ steepens, which shifts the mode of ${P^{(K)}}'(\cdot)$ upward. In steady usage, the pair $(K,s_K)$ can be calibrated from fill-rate telemetry without modeling the full joint distribution of peers.} All investment and budgeting results carry through by replacing $P$ and $P'$ with $P^{(K)}$ and ${P^{(K)}}'$, and by interpreting $H$ as the discounted continuation value per graduate seat. Managerially, the key difference is that $K$ becomes a policy lever: increasing $K$ lowers the implied threshold $s_K$, raises pass rates, and flattens the slope around the margin; decreasing $K$ has the opposite effect. The same leverage ratio $\Lambda(\mu)=P'(\mu)/P(\mu)$ organizes the targeting advantage of the bounty: for fixed $K$, tune the bar so that $\Lambda(\mu^\star)$ is large around the induced $\mu^\star$; for fixed bar, choose $K$ to avoid the extremes in which almost everyone or almost no one graduates.

These two extensions, heterogeneity and multiple seats, invite one final governance observation. When different segments or cohorts face different bars and continuation values, transparency matters. The instruments are easiest to communicate when each segment has a posted testing window, a posted bar (or cohort size), and a posted one-time bounty; creators then internalize both the certainty and the stakes of early performance. The balancing logic with common shadow prices guarantees that this transparency is also efficient: it selects segments and cohorts on the basis of marginal value per unit resource, not on opaque editorial judgment.

\subsection{General Exploitation Engines}\label{subsec:general_engines}

The analysis in Sections~\ref{sec:model}-\ref{sec:mechanisms_main} summarized post-graduation exposure by a single parameter $H$, interpreted as the discounted size of a continuation block. Real ranking systems allocate continuation stochastically, with exposure responding to posterior beliefs that update as fresh outcomes arrive. The framework accommodates such engines with a change of notation that preserves the core structure and the implementability result.

Let $H_1$ denote the expected discounted pull count \emph{conditional on passing} the graduation bar and entering the promising pool, and let $H_0$ denote the expected discounted pull count \emph{conditional on failing} and remaining in organic rotation. For a focal entrant with quality $\mu$, early testing produces a pass event with probability $P(\mu)$ and a fail event with probability $1-P(\mu)$. By the law of iterated expectations, the entrant’s discounted expected continuation exposure is
\[
\mathbb{E}\big[\text{continuation exposure}\,|\,\mu\big]
\;=\;
H_0\;+\;(H_1-H_0)\,P(\mu).
\]
Adding the guaranteed testing window yields the exposure aggregator
\begin{equation}\label{eq:exposure_twoarm}
\Xi(\mu)\;=\;q \;+\; H_0 \;+\; \Delta H \cdot P(\mu),\qquad \Delta H \equiv H_1-H_0\ \ge 0.
\end{equation}
Expression \eqref{eq:exposure_twoarm} nests the block model with $H_0=0$ and $H_1=H$. It also nests index or Thompson engines in which “failure” still earns some tail exposure, e.g., through long-tail sampling or periodic resets, provided that admission to a promising pool (or to a higher-priority band) is governed by a threshold on the same testing statistics. The private and planner objectives become
\[
\Pi_C(\mu)\;=\;\alpha\,\mu\,\big[q+H_0+\Delta H\,P(\mu)\big]\;+\;B P(\mu)\;-\;c(\mu),
\qquad
W(\mu)\;=\;\mu\,\big[q+H_0+\Delta H\,P(\mu)\big]\;-\;B P(\mu),
\]
and the private first-order condition reads
\begin{equation}\label{eq:foc_generalH}
\alpha\Big[q+H_0+\Delta H\,P(\mu^\star)\Big]\;+\;\alpha\,\mu^\star\,\Delta H\,P'(\mu^\star)\;+\;B\,P'(\mu^\star)\;=\;c'(\mu^\star).
\end{equation}
Relative to \eqref{eq:foc_main}, the only change is that the continuation scale entering both the level and slope terms is $\Delta H$, the \emph{incremental} value of passing the bar, while the baseline continuation $H_0$ acts like an additive contribution to $q$. This decomposition is both analytically convenient and intuitive: guarantees that apply regardless of graduation ($H_0$) behave like extra testing impressions, whereas the \emph{prize spread} $\Delta H$ determines how decisive graduation is and therefore how powerful the slope term is as an incentive device.

The planner’s first-order condition has the familiar form
\[
q+H_0+\Delta H\,P(\mu^{FB})\;+\;\mu^{FB}\,\Delta H\,P'(\mu^{FB})\;=\;c'(\mu^{FB}),
\]
and the implementability result carries over with a cosmetic change:
\begin{equation}\label{eq:Bstar_generalH}
B^\star \;=\; \frac{\big[q+H_0+\Delta H\,P(\mu^{FB})+\mu^{FB}\,\Delta H\,P'(\mu^{FB})\big]\,(1-\alpha)}{P'(\mu^{FB})}.
\end{equation}
The bounty again concentrates transfers on the pass event, but the magnitude of the top-up now scales with the incremental continuation value $\Delta H$. When engines are “soft” in the sense that $H_0$ is nontrivial (e.g., a Thompson sampler that occasionally resurfaces failed entrants), the private return to quality already contains a level component from $H_0$; the implementable bounty focuses on aligning incentives at the margin that separates soft continuation from hard graduation. From a design angle, this suggests two complementary levers in index engines: make the \emph{spread} between priority bands meaningful (increase $\Delta H$), and, holding that spread fixed, set the bar where $P'(\cdot)$ is steep and use the bounty to close any residual wedge caused by low monetization $\alpha$.

Two practical remarks help with calibration. First, $H_0$ and $H_1$ are not free parameters; they are induced by the chosen index or sampling policy and can be estimated ex ante by simulating the engine on historical posterior paths. The resulting $(H_0,\Delta H)$ summarize the continuation landscape for policy design without requiring the main text to track the full dynamics. Second, if product teams prefer to avoid discrete bars altogether, the same machinery goes through by defining a \emph{virtual bar} as a posterior-index threshold for entry into a higher-priority band. Because index thresholds translate into integer thresholds on binomial counts in finite testing windows, the pass event remains a binomial tail, and $P(\cdot)$ and $P'(\cdot)$ keep their closed forms. The economic conclusions are unchanged: incentives ride on the certainty created by $q+H_0$, the decisiveness created by $\Delta H$, and the diagnosticity of the pass event, which the bounty can target with precision.

\subsection{Robustness}\label{subsec:robustness}

\begin{figure}[t]
  \centering
  \includegraphics[width=.72\linewidth]{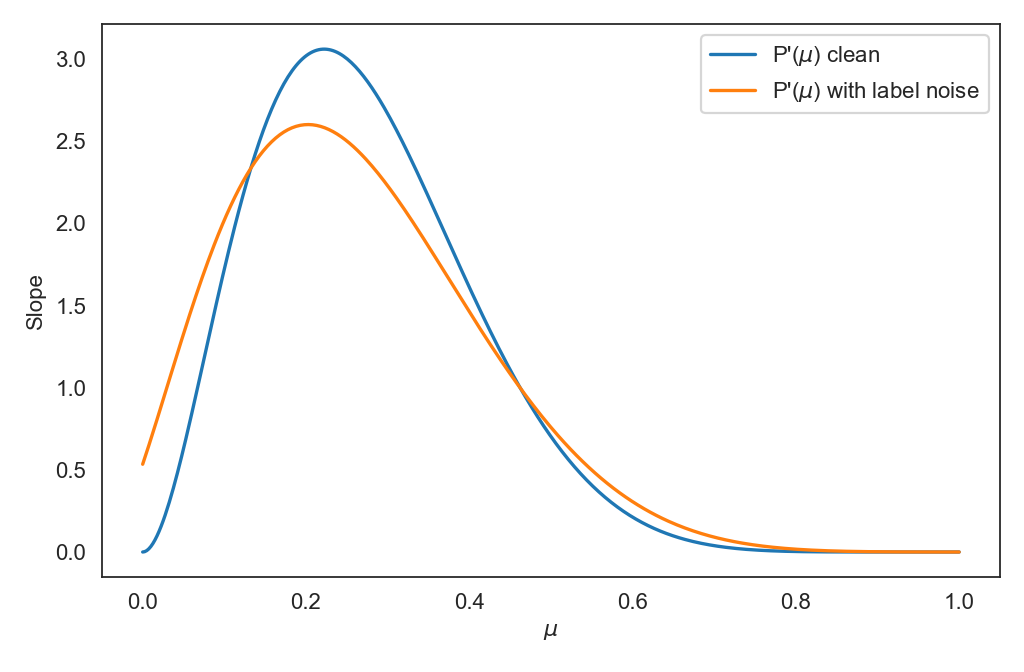}
  \caption{\textbf{Label noise attenuates the slope.}
  With misclassification rates $(\eta_0,\eta_1)=(0.10,0.05)$, the slope scales by $(1-\eta_0-\eta_1)$ and evaluates at the shifted $\tilde p$.
  Lower $P'$ implies a larger $B^\star$ to maintain incentives at the bar.\\\footnotesize\emph{Primitives used:} $q{=}10$, $s{=}3$, $\eta_0 = 0.1$, $\eta_1 = 0.05$.}
  \label{fig:sec5-noise}
\end{figure}

The analytical structure that underpins the results is deliberately light: engagement outcomes during testing are i.i.d.\ conditional on quality, the graduation event is defined by a count threshold, and the private and planner problems depend on the binomial tail $P(\mu)$ and its slope $P'(\mu)$ through the exposure aggregator $q+H P(\mu)$ (or $q+H_0+\Delta H P(\mu)$ in \eqref{eq:exposure_twoarm}). This section argues that the main prescriptions are robust to richer outcome technologies, modest forms of measurement noise and manipulation risk, and time-variation in audience conditions. The common thread is that what matters for incentives is the \emph{diagnosticity} of the pass event around the induced $\mu^\star$, not the fine details of the signal process.

A first robustness margin concerns outcome dispersion. Many surfaces exhibit over-dispersed engagement relative to the Bernoulli baseline, either because underlying propensities are heterogeneous across user subtypes or because bursts of attention generate correlated successes. If the testing statistic remains a sum of independent draws conditional on an \emph{effective} quality parameter, then $P(\mu)$ is the tail of a familiar family (e.g., Beta-Binomial or Poisson-Binomial) and $P'(\mu)$ is its corresponding density with respect to $\mu$. The form of \eqref{eq:foc_main} is unchanged: the level term scales with $P(\mu)$ and the slope term scales with $P'(\mu)$. The comparative statics, the implementability bounty in Theorem~\ref{thm:fb_implement}, and the balanced exploration rule in \eqref{subsec:balanced_rule} carry through verbatim with $(P,P')$ reinterpreted accordingly. In practice, over-dispersion simply widens confidence bands for early statistics; the design implication is to choose $q$ and the bar so that the empirical leverage ratio $\Lambda(\mu)=P'(\mu)/P(\mu)$ remains large at the observed $\mu^\star$ despite noise.

A second margin is measurement error and reporting integrity. The main text assumed platform-measured outcomes to focus on real quality rather than reporting games. If creators can occasionally inflate measured success with probability $\eta\in[0,1)$ that is independent of $\mu$, the pass probability becomes $\tilde{P}(\mu)=(1-\eta)P(\mu)+\eta$ and the slope becomes $\tilde{P}'(\mu)=(1-\eta)P'(\mu)$. All incentive expressions simply scale by $(1-\eta)$ in the slope term and shift in the level term; the implementability bounty \eqref{eq:Bstar} remains well-defined after replacing $(P,P')$ by $(\tilde{P},\tilde{P}')$, and is larger by a factor that compensates for the reduced diagnosticity. If manipulation is deterred by random audits with penalty $\phi$, the effective $\eta$ falls endogenously with $\phi$; setting $\phi$ to keep $(1-\eta)$ high is operationally equivalent to keeping $\Lambda(\mu^\star)$ high. The budget takeaway is that dollars and impressions spent on maintaining clean measurement buy back leverage in the slope term and reduce the bounty required for alignment. 

A third margin is strategic \emph{composition} of effort. Creators may allocate resources to ``desirable'' improvements that raise long run user utility and to ``undesirable'' tactics that raise early clicks without persistence (e.g., click-bait thumbnails). A simple two-effort extension partitions the quality index into $\mu=\mu^{\mathrm{d}}+\mu^{\mathrm{u}}$ with convex costs $c^{\mathrm{d}}$ and $c^{\mathrm{u}}$, while the continuation engine discounts “undesirable” contributions in the long run so that only $\mu^{\mathrm{d}}$ meaningfully affects $H_1$.\footnote{Equivalently, one can model the testing statistic as a weighted sum in which desirable effort loads more heavily into the post-graduation posterior than into the testing posterior.} The optimal policy retains the same structure but \emph{coarsens} the testing statistic or pairs it with content-quality audits so that the pass event’s slope with respect to $\mu^{\mathrm{d}}$ is large while the slope with respect to $\mu^{\mathrm{u}}$ is small. Operationally, this means designing bars around signals that are predictive of sustained satisfaction (e.g., long-view shares or delayed likes) rather than immediate clicks; doing so preserves the targeting advantage of $B$ in \eqref{eq:Bstar} while muting incentives to allocate effort to low-persistence tactics.

Time variation in audience conditions is also common: novelty decays within a session; user intent and inventory quality vary across dayparts. If such variation can be summarized by a time-dependent discount factor $\gamma_t$ that multiplies both the certain and selection components symmetrically, the front-loading result in Theorem~\ref{thm:frontload} strengthens: placing testing early increases the discounted mass $q(\tau)=\sum_{j}\prod_{u< t_j}\gamma_u$ while leaving the distribution of the count statistic $S$ unchanged. If novelty decay lowers the effective per-trial success probability later in the horizon, then early testing also raises $P(\mu)$ and $P'(\mu)$ at a given $\mu$; both effects reinforce the case for early pacing. The practical guidance is to parameterize $q$ in discounted units and to consume that budget as early as eligibility allows, reserving later slots for organic exploration that is not intended to do incentive work.

Capacity constraints and eligibility filters are easily accommodated. If at most one guaranteed impression can be consumed per period, the definition of the earliest schedule $\tau^\uparrow$ in Theorem~\ref{thm:frontload} is “as early as allowed by the cap,” and the weak dominance of front-loading holds as stated. If eligibility depends on coarse segment-level safety screens, the design should treat those screens as part of the guaranteed window $q$ only if passing them is under the entrant’s control and predictive of long run value; otherwise they should be treated as exogenous gating that does not enter the incentive calculus.

Finally, even though our focus is on pre-entry investment, the same logic extends to repeated entry and identity resets. If identities can be reset at low cost, a stricter bar or a slightly larger $q$ may be optimal to preserve diagnosticity and deter churn. Conversely, if resets are costly or reputations are portable, the planner can safely sharpen the bar and lean more on $B$ without encouraging cycling. The common diagnostic remains the same: monitor the empirical pass rate and leverage ratio around the induced $\mu^\star$, and adjust the policy so that targeted dollars fall on informative events.

\subsection{Substantive Applications}\label{subsec:dashboard}

Translating theory into durable operational practice requires a dashboard that expresses the policy in the same currency as planning and experimentation systems. The primitives $(q,\bar{\mu},B,H)$ (or $(q,\bar{\mu},B,H_0,\Delta H)$ in \eqref{eq:exposure_twoarm}) are naturally paired with a small set of empirical diagnostics that can be estimated from routine telemetry and lightweight pilots, and that map one-for-one to the comparative statics and budget identities in Sections~\ref{sec:mechanisms_main}-\ref{sec:resources_bwk}.

The first diagnostic is the \emph{delivered} discounted testing mass per entrant, $\hat{q}$, computed from realized pacing using the same discounting scheme used to value engagement. Because Theorem~\ref{thm:frontload} speaks in discounted units, it is essential to monitor $\hat{q}$ rather than just the raw count of early impressions; a cohort whose guaranteed impressions slip late in the session effectively receives a smaller $q$ and therefore weaker incentives. The second diagnostic is the \emph{cohort pass rate} $\widehat{P}$ at the current bar and testing window, together with a local estimate of the leverage ratio $\widehat{\Lambda}$ around the observed $\mu^\star$. The pass rate is visible in standard funnel reports; the leverage ratio can be estimated by inducing small, randomized, creator-facing nudges that shift observable quality proxies and then measuring the induced change in pass probability, or by fitting a smooth pass-probability curve and reading off its slope at the cohort’s median. These two numbers implement the bar-tuning logic in \S\ref{subsec:tuning_bar}: move the bar when $\widehat{\Lambda}$ is low because pass events are either too rare or too common, and stop when the pass event is most diagnostic.

The third diagnostic is the continuation landscape $(\widehat{H}_0,\widehat{\Delta H})$, estimated by simulating or replaying the exploitation engine on historical posteriors and measuring the expected discounted pull count conditional on passing versus failing. Reporting the \emph{spread} $\widehat{\Delta H}$ next to the bar makes incentives tangible for both managers and creators: it is the operational value of crossing the threshold. When $\widehat{\Delta H}$ is small, the dashboard should flag that either the graduation bar is misaligned with the engine or that the engine’s priority bands should be made more distinct; in either case the slope term in \eqref{eq:foc_generalH} will be weak unless $\widehat{\Lambda}$ is very large.

On the budgeting side, two lines close the loop. The expected impression spend per entrant is simply $\hat{q}$, and the expected cash spend is $\widehat{B}\,\widehat{P}$, where $\widehat{B}$ is the posted bounty. With these in hand, the balanced exploration rule in \eqref{subsec:balanced_rule} can be operationalized by plotting, for a grid of $(q,B)$ values around the status quo, the estimated marginal lift in constrained welfare per additional discounted impression and per additional expected dollar of payout. The intersection with the implied shadow-price lines from the impression and cash budgets identifies the recommended adjustments. Because these marginal curves depend on the induced $\mu^\star$, they should be recomputed periodically as the entrant population responds; in steady state, the dashboard will converge to a regime in which early slots are consumed as early as feasible, the bar keeps the pass event diagnostic, and the posted bounty is just large enough to offset any residual wedge caused by a low monetization share.

Governance and communication are as important as calibration. The instruments are easiest to explain, and least likely to elicit gaming, when they are posted in creator-facing documentation: a guaranteed testing window expressed in opportunities and time, a clear graduation criterion, a schematic of the continuation engine that makes the stakes apparent, and a one-time bounty that triggers on the same event as graduation. Transparency does not weaken incentives; it strengthens them by allowing creators to plan their effort against a stable rule. Internally, the same transparency disciplines policy changes: bar moves should be justified by movements in the leverage ratio, bounty moves by movements in monetization or budgets, and changes in early-slot allocations by movements in the shadow prices of impressions.

Finally, dashboards should track a small set of outcome KPIs that tie back to the theory: the estimated \emph{quality lift} $\Delta \mu^\star$ (proxied by pre-entry process metrics or by post-graduation persistence), the \emph{discovery rate} $\widehat{P}$, the \emph{spend efficiency} measured as expected dollars per basis point of $\mu$ at the margin (a function of $\widehat{\Lambda}$), and the \emph{shadow prices} of impressions and cash that justify current allocations. When these KPIs move in the directions predicted by the model, higher $\widehat{\Lambda}$, larger $\widehat{\Delta H}$, and stable or declining expected spend per unit of induced incentive, the platform can be confident that early exposure and targeted transfers are being used where they have the largest payoff. When they do not, the corrective actions are the ones the analysis has emphasized throughout: front-load more of $q$, retune the bar to restore diagnosticity, or rebalance toward the instrument whose marginal gain per unit of resource is currently higher.

\section{Conclusion}\label{sec:conclusion}

This paper develops a tractable theory of supply side exploration on digital platforms when creators are strategic about pre-entry quality. The starting point is a simple observation: the allocation engine that learns which content to promote is also the instrument that shapes whether effort is worth undertaking before the platform knows anything about a new creator. In the baseline, when early exposure is sparse and selection around graduation is either too forgiving or too stringent, the marginal return to quality is weak and a cold-start externality emerges. The analysis shows how a small set of design levers, a guaranteed testing window $q$, a transparent graduation bar $\bar{\mu}$, an explicit continuation value summarized by $H$ (or $(H_0,\Delta H)$), and a one-time, hit-based bounty $B$, can be assembled into an implementable policy that restores strong incentives while respecting impression and cash budgets.

Two analytical objects organize the results: the binomial tail $P(\mu)$, which gives the probability that a creator of quality $\mu$ clears the bar after $q$ tests, and its slope $P'(\mu)$, which measures the diagnosticity of the pass event. With these in hand, the entrant’s best response is characterized by a single first-order condition \eqref{eq:foc_main}. Under weak regularity, a unique equilibrium quality $\mu^\star$ exists (Proposition~\ref{prop:exist_unique}) and moves monotonically with each parameter (Theorem~\ref{thm:mcs}). The comparative statics are economically transparent. A larger testing window raises the certain return to effort even for creators who will not immediately graduate; a richer continuation stream increases the stakes of being selected and strengthens incentives at the margin; a higher revenue share lets creators internalize more of the engagement surplus; and a hit-based bounty concentrates transfers exactly where small improvements in quality most change graduation odds.

Three design principles follow. First, timing matters. When testing opportunities are managed as raw impressions rather than discounted units, the earliest feasible pacing weakly dominates all others (Theorem~\ref{thm:frontload}). Front-loading increases the certain component of the return without changing the distribution of the graduation statistic; for a fixed count of trials, nothing is lost by moving them forward and much is gained in incentive strength. Second, targeting matters. Because $P'(\mu)$ is a Beta density that peaks near $(s-1)/(q-1)$, the bar should be placed where the pass event is neither too rare nor too common for the entrant cohort; the leverage ratio $\Lambda(\mu)=P'(\mu)/P(\mu)$ provides a useful diagnostic. When $\Lambda(\mu^\star)$ is large, a dollar tied to the pass event buys more marginal incentive than a dollar spread uniformly over testing, and the same lift in the private first-order condition can be achieved with lower expected spend. Third, alignment is feasible. A closed-form bounty $B^\star$ implements the planner’s first best (Theorem~\ref{thm:fb_implement}) by topping up the private marginal return at the graduation margin; it vanishes as monetization approaches full revenue sharing and shrinks as the pass event becomes more diagnostic.

Budgets enter the picture naturally. Guaranteed impressions consume attention; bounties consume cash. Treating both as explicit constraints leads to a primal-dual characterization in which the optimal policy equalizes the marginal gain in constrained welfare per discounted testing impression and per expected dollar of payout. The resulting balanced exploration rule (Section~\ref{sec:resources_bwk}) is operational: for a given pair of shadow prices, increase $q$ until the next discounted early impression is worth its shadow price and increase $B$ until the next expected payout dollar is worth its own shadow price, adjusting the bar so that the pass event remains diagnostic around the induced equilibrium. This rule scales cleanly to segmentation with a common set of shadow prices and segment-specific marginal gains, and it nests multi-winner discovery by treating cohort size as a choice of threshold.

While the model is stylized, the prescriptions are robust and managerially legible. The results extend to over-dispersed or correlated outcomes as long as the graduation statistic admits a smooth tail and density; they accommodate exploitation engines implemented by indices or sampling by reinterpreting $H$ as an expected pull count and focusing on the prize spread $\Delta H$; and they tolerate modest measurement noise and integrity risks by recognizing that safeguarding diagnosticity is itself a budgeted activity that buys back slope in the graduation frontier. The framework also clarifies where care is needed. If early statistics are driven by low-persistence tactics, the pass event should be built from signals that predict long run  satisfaction rather than immediate clicks; if identity resets are cheap, the bar and window should be tuned to preserve diagnosticity without encouraging churn; if fairness or safety constraints limit continuation concentration, the bar must carry more of the incentive load and the bounty can be used to close the remaining wedge.

There are natural directions for further work. A dynamic analysis with repeated creative choices would allow for intertemporal substitution in effort and reputation effects, as well as the design of rolling bars and renewal of testing windows. Endogenous audience composition and cross-creator spillovers would connect discovery incentives to the platform’s broader ecology, including diversity and fairness constraints. Richer reporting games, including partial manipulability of testing statistics, would link the present approach to strategic classification models of explanation and auditing. Finally, the interaction between supply side incentives and user side information design, for instance, surfacing previews or uncertainty disclosures that improve the diagnosticity of early outcomes, promises a unified treatment of content strategy across both sides of the market.

In summary, the intuition is simple. Attention and cash are scarce resources; selection is a slope. Good policy spends the former where the latter is steep. A small number of early, well placed opportunities and a small, well-targeted bounty can move investment behavior meaningfully, provided the graduation event is designed to be informative at the right margin. Because the instruments are transparent and their effects are easy to monitor with standard telemetry, the theory translates into a playbook that platform designers can implement, audit, and iterate: front-load the incubation window, place the bar where the pass event is diagnostic, make the prize spread meaningful, and use targeted transfers to close any residual wedge. When these conditions hold, the cold-start externality recedes, high-quality entrants surface more reliably, and the platform’s exploration policy advances both efficiency and growth.

\newpage
\bibliographystyle{chicago}
\bibliography{ref}
\newpage
\appendix
\section*{ONLINE APPENDIX}
\section{Proofs and Technical Lemmas}\label{app:proofs}

This appendix collects formal statements and proofs for the results in the main text. The goal is clarity over generality: we keep the primitives as introduced in Sections~\ref{sec:model}-\ref{sec:mechanisms_main}, specialize where it simplifies exposition, and highlight where arguments rely on standard analytic facts. Throughout, $q\in\mathbb{N}$ denotes the (discounted) size of the testing window, $s\in\{0,1,\ldots,q\}$ is the integer threshold implied by the posted bar, and $\mu\in(0,1)$ is the entrant’s pre-entry quality. When convenient, we write $H\ge 0$ for the continuation scale and, in the generalized engine of Section~\ref{subsec:general_engines}, $(H_0,\Delta H)$ with $\Delta H\equiv H_1-H_0\ge 0$. The revenue share is $\alpha\in(0,1]$ and the one-time bounty is $B\ge 0$. The cost function $c:[\underline{\mu},\overline{\mu}]\to\mathbb{R}_+$ is twice continuously differentiable and strictly convex.

\subsection*{A.1 Preliminaries and notation}

During testing, the entrant’s content is displayed $q$ times in expectation, independently across trials, and each trial yields a binary success with probability $\mu$. Let
\[
S\ \sim\ \mathrm{Binomial}(q,\mu)
\qquad\text{and}\qquad
P(\mu)\ \equiv\ \Pr[S\ge s]\ =\ \sum_{k=s}^{q}\binom{q}{k}\mu^k(1-\mu)^{q-k}.
\]
The function $P(\cdot)$ is the \emph{binomial tail} at threshold $s$; it maps true quality into the probability of clearing the bar after $q$ tests. We will frequently use its derivative $P'(\cdot)$ and, occasionally, $P''(\cdot)$. Two special-function identities keep the analysis compact. The \emph{Beta function} is $B(a,b)=\int_{0}^{1} t^{a-1}(1-t)^{b-1}\,dt$ for $a,b>0$. The \emph{regularized incomplete Beta function} is
\[
I_{x}(a,b)\ \equiv\ \frac{1}{B(a,b)}\int_{0}^{x} t^{a-1}(1-t)^{b-1}\,dt,\qquad a,b>0,\ \ x\in[0,1].
\]
A standard representation (see, e.g., Abramowitz-Stegun 26.5.24) gives $P(\mu)=I_{\mu}(s,q-s+1)$ when $1\le s\le q$. Differentiating under the integral shows that
\begin{equation}\label{eq:beta_density_identity_app}
\frac{d}{d\mu} I_{\mu}(a,b)\ =\ \frac{\mu^{a-1}(1-\mu)^{b-1}}{B(a,b)}.
\end{equation}
We will invoke \eqref{eq:beta_density_identity_app} with $(a,b)=(s,q-s+1)$ to obtain the closed form for $P'(\mu)$.

When referring to equilibrium objects, we keep the notation from the main text. The entrant’s discounted objective under $(q,s,B,H)$ is
\[
\Pi_C(\mu)\ =\ \alpha\,\mu\,\big[q+H P(\mu)\big]\ +\ B\,P(\mu)\ -\ c(\mu),
\]
with the first-order condition (FOC) for an interior optimum given by
\begin{equation}\label{eq:foc_app}
\alpha\big[q+H P(\mu^\star)\big]\ +\ \alpha\,\mu^\star\,H\,P'(\mu^\star)\ +\ B\,P'(\mu^\star)\ =\ c'(\mu^\star).
\end{equation}
In the generalized engine of Section~\ref{sec:general_engines}, replace $q+H P(\mu)$ by $q+H_0+\Delta H\,P(\mu)$ and $H$ by $\Delta H$ in \eqref{eq:foc_app}. The planner’s FOC for the first-best benchmark is
\begin{equation}\label{eq:planner_foc_app}
q+H P(\mu^{FB})\ +\ \mu^{FB} H P'(\mu^{FB})\ =\ c'(\mu^{FB}),
\end{equation}
or the obvious variant with $(H_0,\Delta H)$.

Finally, we restate and slightly sharpen the regularity condition used in Section~\ref{sec:mechanisms_main}.

\noindent\textbf{Assumption A1 (Regularity and single crossing).} The cost function $c$ is $C^2$ and strictly convex on $[\underline{\mu},\overline{\mu}]$. Moreover, there exists a constant $\kappa_{\min}>0$ such that $c''(\mu)\ge \kappa_{\min}$ for all $\mu\in[\underline{\mu},\overline{\mu}]$, and
\begin{equation}\label{eq:single_crossing_condition}
\kappa_{\min}\ >\ \alpha H\,\Big(2\,\sup_{\mu\in(0,1)} P'(\mu)\ +\ \sup_{\mu\in(0,1)} \mu\,|P''(\mu)|\Big)\ +\ B\,\sup_{\mu\in(0,1)} |P''(\mu)|.
\end{equation}
Assumption A1 implies that the “gap” function $\Delta(\mu)\equiv c'(\mu)-\big[\alpha(q+H P(\mu))+\alpha \mu H P'(\mu)+B P'(\mu)\big]$ is strictly increasing on $[\underline{\mu},\overline{\mu}]$.

Condition \eqref{eq:single_crossing_condition} is a convenient sufficient condition. It uses the facts, established below, that $P'$ and $P''$ are continuous and bounded on $(0,1)$ for fixed integer $(q,s)$ with $1\le s\le q$. Intuitively, Assumption A1 requires the cost curvature to dominate the combined curvature of the benefit term induced by the selection frontier. In common specifications such as quadratic or strongly convex costs, the bound is easy to satisfy.

\subsection*{A.2 Properties of the binomial tail and its derivatives}

We collect the analytic properties of $P(\cdot)$ that are used throughout.

\begin{lemma}[Derivative identity and basic properties]\label{lem:Pprime_identity}
Fix integers $q\ge 1$ and $s\in\{1,\ldots,q\}$. The binomial-tail function
\[
P(\mu)=\sum_{k=s}^{q}\binom{q}{k}\mu^k(1-\mu)^{q-k},\qquad \mu\in(0,1),
\]
is continuously differentiable on $(0,1)$, strictly increasing, and satisfies
\begin{equation}\label{eq:Pprime_closed_form_app}
P'(\mu)\ =\ \frac{\mu^{s-1}(1-\mu)^{q-s}}{B\!\big(s,\,q-s+1\big)}\ \equiv\ f_{\mathrm{Beta}(s,\,q-s+1)}(\mu).
\end{equation}
If $1<s<q$, then $P'$ is unimodal with unique mode at $\mu^\dagger=(s-1)/(q-1)$, $\lim_{\mu\downarrow 0} P'(\mu)=0$, and $\lim_{\mu\uparrow 1} P'(\mu)=0$. Moreover, $P$ and $P'$ extend continuously to $[0,1]$ with $P(0)=\mathbb{1}\{s=0\}$ and $P(1)=1$.
\end{lemma}

\noindent\emph{Proof.} The series representation is a finite sum of $C^\infty$ functions on $(0,1)$, hence $P$ is $C^1$ and $P'$ can be computed by term-wise differentiation. Alternatively, use the identity $P(\mu)=I_{\mu}(s,q-s+1)$ and \eqref{eq:beta_density_identity_app} to obtain \eqref{eq:Pprime_closed_form_app}. Because the Beta density is nonnegative on $(0,1)$, $P$ is weakly increasing; in fact strictly increasing because $P'(\mu)>0$ for $\mu\in(0,1)$ when $1\le s\le q$. For $1<s<q$, the Beta density is unimodal with mode $(s-1)/(s+q-s+1-2)=(s-1)/(q-1)$ and decays to zero at the endpoints; see any standard reference. Continuity at the boundaries follows from dominated convergence. \qed

\begin{lemma}[Bounds and Lipschitz continuity]\label{lem:bounds}
For fixed integers $(q,s)$ with $1\le s\le q$, there exist finite constants $M_1,M_2>0$ depending only on $(q,s)$ such that
\[
\sup_{\mu\in(0,1)} P'(\mu)\ \le\ M_1,\qquad \sup_{\mu\in(0,1)} |P''(\mu)|\ \le\ M_2.
\]
Consequently, $P$ is globally Lipschitz on $[0,1]$ and $P'$ is globally Lipschitz on any compact subinterval of $(0,1)$.
\end{lemma}

\noindent\emph{Proof.} The closed form \eqref{eq:Pprime_closed_form_app} implies $P'$ is continuous on $(0,1)$ and extends continuously to $[0,1]$ with limits zero at the boundaries when $1<s<q$; in the edge cases $s\in\{1,q\}$ the expression reduces to a power function and is still bounded. Continuity on a compact set implies a finite supremum, yielding $M_1$. Differentiating \eqref{eq:Pprime_closed_form_app} explicitly gives
\[
P''(\mu)\ =\ \frac{(s-1)\mu^{s-2}(1-\mu)^{q-s}-(q-s)\mu^{s-1}(1-\mu)^{q-s-1}}{B(s,q-s+1)},
\]
which is continuous on $(0,1)$ and integrable near the endpoints for $s\in\{1,\ldots,q\}$; the supremum is thus finite, yielding $M_2$. Lipschitz continuity follows from the bounded derivative on compact sets. \qed

\begin{lemma}[Hazard-type ratio and diagnosticity]\label{lem:hazard}
Define $\Lambda(\mu)\equiv P'(\mu)/P(\mu)$ for $\mu\in(0,1)$ with $P(\mu)\in(0,1)$. Then $\Lambda(\cdot)$ is continuous wherever defined and finite on compact subsets of $\{\mu: P(\mu)\in(0,1)\}$. If $1<s<q$, $\Lambda(\mu)$ is large precisely when the integer threshold $s$ is aligned with the central mass of the quality distribution; in particular, for $\mu$ near $\mu^\dagger=(s-1)/(q-1)$, $\Lambda(\mu)$ attains relatively high values.
\end{lemma}

\noindent\emph{Proof.} Continuity follows from the quotient of continuous functions with nonvanishing denominator. The qualitative statement uses the unimodality of $P'$ and the S-shape of $P$ for $1<s<q$: around the middle of the support, $P$ is away from 0 and 1 while $P'$ is near its peak, hence their ratio is relatively large. \qed

\noindent\textbf{Remarks.} (i) The proofs above extend verbatim to over-dispersed families used in robustness checks (e.g., Beta-Binomial), with $P'(\cdot)$ replaced by the corresponding density. (ii) All suprema in Assumption A1 are finite by Lemma~\ref{lem:bounds}.

\subsection*{A.3 Existence and uniqueness of the best response (Proposition~\ref{prop:exist_unique})}

We now establish the existence and uniqueness of $\mu^\star$ solving the private optimality condition \eqref{eq:foc_app} under Assumption A1. The argument proceeds in three steps: derive the FOC from first principles and verify interior optimality conditions; show that the “gap” function $\Delta(\mu)$ is strictly increasing; and conclude uniqueness by the intermediate value theorem.

\begin{lemma}[Derivation of the FOC and sufficiency]\label{lem:foc_derivation}
Assume $c$ is $C^2$ and strictly convex. Then $\Pi_C(\mu)=\alpha\,\mu\,[q+H P(\mu)]+B P(\mu)-c(\mu)$ is continuously differentiable on $(\underline{\mu},\overline{\mu})$ with derivative
\[
\Pi_C'(\mu)\ =\ \alpha\big[q+H P(\mu)\big]\ +\ \alpha\,\mu\,H\,P'(\mu)\ +\ B\,P'(\mu)\ -\ c'(\mu).
\]
If Assumption A1 holds, $\Pi_C(\cdot)$ is strictly quasi-concave on $[\underline{\mu},\overline{\mu}]$ and any solution to $\Pi_C'(\mu)=0$ is the unique maximizer.
\end{lemma}

\noindent\emph{Proof.} Differentiability follows from the product rule and Lemma~\ref{lem:Pprime_identity}. The derivative equals the difference between marginal benefit and marginal cost, yielding \eqref{eq:foc_app}. To argue strict quasi-concavity, observe that
\[
\Pi_C''(\mu)\ =\ \alpha H\Big(2\,P'(\mu)+\mu\,P''(\mu)\Big)\ +\ B\,P''(\mu)\ -\ c''(\mu).
\]
By Assumption A1, $c''(\mu)$ dominates the positive and negative components induced by $P'$ and $P''$ uniformly in $\mu$, so $\Pi_C''(\mu)<0$ for all $\mu$. Strict concavity is stronger than needed but sufficient to ensure any stationary point is the unique global maximizer. \qed

\begin{lemma}[Monotonicity of the gap function]\label{lem:gap_monotone}
Define $\Delta(\mu)\equiv c'(\mu)-\big[\alpha(q+H P(\mu))+\alpha \mu H P'(\mu)+B P'(\mu)\big]$ on $[\underline{\mu},\overline{\mu}]$. Under Assumption A1, $\Delta(\cdot)$ is strictly increasing and continuous.
\end{lemma}

\noindent\emph{Proof.} Continuity follows from the continuity of all terms. Differentiating,
\[
\Delta'(\mu)\ =\ c''(\mu)\ -\ \alpha H\Big(2\,P'(\mu)+\mu\,P''(\mu)\Big)\ -\ B\,P''(\mu).
\]
By Assumption A1, $c''(\mu)\ge \kappa_{\min}$ and the remaining terms are bounded in absolute value by the right-hand side of \eqref{eq:single_crossing_condition}, which $\kappa_{\min}$ strictly exceeds. Hence $\Delta'(\mu)>0$ for all $\mu$. \qed

\begin{proof}[Proof of Proposition~\ref{prop:exist_unique}]
By Lemma~\ref{lem:foc_derivation}, any interior maximizer satisfies $\Delta(\mu)=0$. By Lemma~\ref{lem:gap_monotone}, $\Delta$ is strictly increasing and continuous on $[\underline{\mu},\overline{\mu}]$. If $\Delta(\underline{\mu})>0$, then $\Pi_C'(\mu)<0$ for all $\mu\in[\underline{\mu},\overline{\mu}]$ and the unique maximizer is $\mu^\star=\underline{\mu}$. If $\Delta(\overline{\mu})<0$, then $\Pi_C'(\mu)>0$ for all $\mu$ and the unique maximizer is $\mu^\star=\overline{\mu}$. Otherwise, the intermediate value theorem guarantees a unique $\mu^\star\in(\underline{\mu},\overline{\mu})$ with $\Delta(\mu^\star)=0$. Strict concavity from Lemma~\ref{lem:foc_derivation} implies this point is the unique global maximum. \qed
\end{proof}

\noindent\textbf{Discussion.} Proposition~\ref{prop:exist_unique} formalizes the intuition that the entrant’s best response is well-behaved under weak curvature conditions. Two features bear emphasis. First, the sufficient condition in Assumption A1 is conservative: it requires a uniform lower bound on $c''$ that dominates the supremum of selection-induced curvature; in many applications, a localized version (evaluated near the equilibrium) is enough. Second, the proof highlights precisely where properties of $P$ matter: boundedness and continuity of $P'$ and $P''$ ensure that the benefit curvature cannot explode, while the positive sign and unimodality of $P'$ deliver the familiar S-shape of the benefit schedule that, when combined with convex costs, yields a single crossing.

\subsection*{A.4 Monotone comparative statics (Theorem~\ref{thm:mcs})}

We prove that the unique best response $\mu^\star$ is weakly increasing in each policy parameter $(q,H,B,\alpha)$ under Assumption~A1. Define the private first-order condition as an equation in the choice variable and parameters:
\[
G(\mu;\,q,H,B,\alpha)\ \equiv\ \alpha\big[q+H P(\mu)\big]\ +\ \alpha\,\mu\,H\,P'(\mu)\ +\ B\,P'(\mu)\ -\ c'(\mu)\ =\ 0.
\]
By Lemma~\ref{lem:foc_derivation}, any interior maximizer satisfies $G(\mu^\star;\cdot)=0$. By Lemmas~\ref{lem:Pprime_identity}-\ref{lem:bounds} and Assumption~A1, $G$ is continuously differentiable and its derivative with respect to $\mu$ is
\[
\frac{\partial G}{\partial \mu}(\mu;\cdot)\ =\ \alpha H\Big(2\,P'(\mu)+\mu\,P''(\mu)\Big)\ +\ B\,P''(\mu)\ -\ c''(\mu)\ <\ 0\quad\text{for all }\mu\in(\underline{\mu},\overline{\mu}).
\]
Hence the \emph{implicit function theorem} applies: in any neighborhood where $\mu^\star$ is interior there exists a unique, continuously differentiable function $\mu^\star(q,H,B,\alpha)$ solving $G(\mu^\star;\cdot)=0$, with partial derivatives
\[
\frac{\partial \mu^\star}{\partial \theta}\ =\ -\,\frac{\displaystyle \frac{\partial G}{\partial \theta}}{\displaystyle \frac{\partial G}{\partial \mu}}\Bigg|_{\mu=\mu^\star}\qquad \text{for }\ \theta\in\{q,H,B,\alpha\}.
\]
We compute signs term by term. For all $\mu\in(\underline{\mu},\overline{\mu})$,
\[
\frac{\partial G}{\partial q}=\alpha>0,\qquad 
\frac{\partial G}{\partial H}=\alpha\Big(P(\mu)+\mu\,P'(\mu)\Big)>0,\qquad 
\frac{\partial G}{\partial B}=P'(\mu)>0,\qquad 
\frac{\partial G}{\partial \alpha}=q+H P(\mu)+\mu H P'(\mu)>0,
\]
using $P(\mu)\in(0,1)$ and $P'(\mu)>0$ (Lemma~\ref{lem:Pprime_identity}). Since $\partial G/\partial \mu<0$, each fraction is nonnegative:
\[
\frac{\partial \mu^\star}{\partial q}\ \ge\ 0,\qquad 
\frac{\partial \mu^\star}{\partial H}\ \ge\ 0,\qquad 
\frac{\partial \mu^\star}{\partial B}\ \ge\ 0,\qquad
\frac{\partial \mu^\star}{\partial \alpha}\ \ge\ 0.
\]
This establishes the interior case in Theorem~\ref{thm:mcs}.

To cover boundary solutions, recall the “gap” function $\Delta(\mu)=c'(\mu)-[\alpha(q+H P(\mu))+\alpha \mu H P'(\mu)+B P'(\mu)]$. Under Assumption~A1, $\Delta(\cdot)$ is continuous and strictly increasing (Lemma~\ref{lem:gap_monotone}). Let $\theta$ denote any one of $(q,H,B,\alpha)$ and write $\Delta(\mu\,;\theta)$ to emphasize dependence. For fixed $\mu$, $\Delta(\mu\,;\theta)$ is weakly \emph{decreasing} in $\theta$ because it subtracts a term that is increasing in $\theta$. Therefore, as $\theta$ increases, the unique zero of $\Delta(\cdot\,;\theta)$ (if interior) moves weakly to the right; if no zero exists (i.e., $\Delta(\underline{\mu};\theta)\ge 0$ or $\Delta(\overline{\mu};\theta)\le 0$), the maximizer is at a boundary and the same monotone movement in the sign pattern implies the selected boundary point is weakly increasing in $\theta$ as well. In either case, the best-response correspondence is a singleton and is weakly increasing in each parameter. This completes the proof of Theorem~\ref{thm:mcs}.

Intuitively, the result follows the economic decomposition in the main text. Increasing $q$ adds certain exposure in testing; increasing $H$ raises both the expected continuation and the sensitivity of graduation to quality; increasing $B$ concentrates monetary incentives where selection is most diagnostic; and increasing $\alpha$ lets creators internalize a larger share of engagement. Because costs are convex and the selection frontier is well behaved, these shifts raise the marginal private return at every $\mu$, so the unique crossing with marginal cost occurs at a (weakly) higher quality.

\subsection*{A.5 Front-loading optimality (Theorem~\ref{thm:frontload})}

We formalize the timing result when testing is scheduled as raw impressions across calendar slots. Fix an undiscounted count $Q\in\mathbb{N}$ of guaranteed testing impressions and an increasing schedule $\tau=(t_1,\ldots,t_Q)$ of delivery times. The \emph{discounted testing mass} induced by $\tau$ is
\[
q(\tau)\ =\ \sum_{j=1}^{Q}\gamma^{\,t_j-1},\qquad \gamma\in(0,1).
\]
The graduation statistic is the success count $S\sim \mathrm{Binomial}(Q,\mu)$ and the pass event is $\{S\ge s\}$ with $s=\lceil Q\bar{\mu}\rceil$. Hence
\[
P_Q(\mu)\ =\ \Pr[S\ge s],\qquad P'_Q(\mu)\ =\ \frac{\mu^{\,s-1}(1-\mu)^{\,Q-s}}{B(s,Q-s+1)},
\]
which depend on $Q$ and $s$ but \emph{not} on the schedule $\tau$ because outcomes are i.i.d.\ conditional on $\mu$. The private first-order condition under schedule $\tau$ is
\begin{equation}\label{eq:foc_tau}
\alpha\Big[q(\tau)+H P_Q(\mu^\star)\Big]\ +\ \alpha\,\mu^\star\,H\,P'_Q(\mu^\star)\ +\ B\,P'_Q(\mu^\star)\ =\ c'(\mu^\star).
\end{equation}

Let $\tau^\uparrow$ denote the earliest feasible schedule that assigns the $Q$ impressions to the first $Q$ eligible slots, i.e., $t_j=j$ absent capacity constraints. We first show that $\tau^\uparrow$ maximizes the discounted testing mass.

\begin{lemma}[Earliest pacing maximizes discounted mass]\label{lem:majorization}
For any schedule $\tau$ with $Q$ impressions and discount factor $\gamma\in(0,1)$,
\[
q(\tau^\uparrow)\ \ge\ q(\tau),
\]
with equality if and only if $\tau=\tau^\uparrow$ up to permutation of equal times.
\end{lemma}

\noindent\emph{Proof.} The sequence $(\gamma^{t-1})_{t\ge 1}$ is strictly decreasing in $t$. By the rearrangement inequality, the sum of $Q$ terms chosen from a strictly decreasing sequence is maximized by taking the $Q$ largest terms, i.e., $t_1=1,\ldots,t_Q=Q$. Equivalently, any delay that replaces some $\gamma^{k-1}$ with $\gamma^{k}$ strictly reduces the sum. \qed

We now prove Theorem~\ref{thm:frontload}. Fix $Q$, $\bar{\mu}$ (hence $s$), and $(H,B,\alpha)$. Consider two schedules $\tau$ and $\tilde{\tau}$ with the same $Q$. Because $P_Q$ and $P'_Q$ do not depend on timing, the only difference between \eqref{eq:foc_tau} under the two schedules is the term $\alpha q(\cdot)$. By Lemma~\ref{lem:majorization}, $q(\tau^\uparrow)\ge q(\tau)$ for all $\tau$. Define the corresponding gap functions
\[
\Delta_{\tau}(\mu)\ \equiv\ c'(\mu)\ -\ \Big(\alpha\big[q(\tau)+H P_Q(\mu)\big]\ +\ \alpha\,\mu\,H\,P'_Q(\mu)\ +\ B\,P'_Q(\mu)\Big).
\]
As in Lemma~\ref{lem:gap_monotone}, each $\Delta_{\tau}$ is continuous and strictly increasing in $\mu$ under Assumption~A1 (with $(P,P')$ replaced by $(P_Q,P'_Q)$). Moreover, for all $\mu$,
\[
\Delta_{\tau^\uparrow}(\mu)\ \le\ \Delta_{\tau}(\mu),
\]
with strict inequality when $q(\tau^\uparrow)>q(\tau)$. Therefore, the unique zero of $\Delta_{\tau^\uparrow}(\cdot)$, if interior, occurs at a weakly larger $\mu$ than the unique zero of $\Delta_{\tau}(\cdot)$; if the solution is at a boundary, the same monotone shift implies the boundary solution under $\tau^\uparrow$ is weakly larger. Hence $\mu^\star(\tau^\uparrow)\ge \mu^\star(\tau)$ for all $\tau$ with the same $Q$, establishing the incentive part of Theorem~\ref{thm:frontload}.

To compare platform objectives, note that for any schedule $\tau$ and corresponding equilibrium $\mu^\star(\tau)$,
\[
W\big(\mu^\star(\tau);\tau,H,B\big)\ =\ \mu^\star(\tau)\,\Big(q(\tau)+H P_Q\big(\mu^\star(\tau)\big)\Big)\ -\ B\,P_Q\big(\mu^\star(\tau)\big).
\]
Because $\mu^\star(\tau^\uparrow)\ge \mu^\star(\tau)$ and $q(\tau^\uparrow)\ge q(\tau)$ while $P_Q(\cdot)$ is increasing, it follows that
\[
W\big(\mu^\star(\tau^\uparrow);\tau^\uparrow,H,B\big)\ \ge\ W\big(\mu^\star(\tau);\tau,H,B\big),
\]
with strict inequality unless both $q(\tau)$ and $\mu^\star(\tau)$ coincide with their front-loaded counterparts. This proves Theorem~\ref{thm:frontload}.

\noindent\textbf{Remarks.} (i) If per-period feasibility caps allow at most one guaranteed impression per period, define $\tau^\uparrow$ as “as early as allowed by the cap.” Lemma~\ref{lem:majorization} carries through and the dominance result is unchanged. (ii) If the continuation engine is of the generalized form in \eqref{eq:exposure_twoarm}, replace $q(\tau)$ with $q(\tau)+H_0$ and $H$ with $\Delta H$; the proof is identical because the selection terms $P_Q$ and $P'_Q$ are invariant to timing. (iii) If discounting varies deterministically over time, replacing $\gamma^{t-1}$ by a decreasing weight sequence $\{w_t\}_{t\ge 1}$ yields the same conclusion by rearrangement: front-loading maximizes $\sum_{j} w_{t_j}$ for fixed $Q$ and leaves $P_Q$, $P'_Q$ unchanged.
\subsection*{A.6 First-best implementability (Theorem~\ref{thm:fb_implement})}

We provide a complete proof of Theorem~\ref{thm:fb_implement} in the block-continuation specification (testing window $q$, single continuation scale $H$). Recall the planner’s benchmark chooses $\mu$ to maximize $\mu\,[q+H P(\mu)]-c(\mu)$ and hence satisfies
\begin{equation}\label{eq:planner_foc_rep}
q+H P(\mu^{FB})\;+\;\mu^{FB} H P'(\mu^{FB})\;=\;c'(\mu^{FB}),
\end{equation}
cf. \eqref{eq:planner_foc_app}. The entrant’s private FOC under policy $(q,\bar{\mu},B,H,\alpha)$ is
\begin{equation}\label{eq:private_foc_rep}
\alpha\big[q+H P(\mu^\star)\big]\;+\;\alpha\,\mu^\star\,H\,P'(\mu^\star)\;+\;B\,P'(\mu^\star)\;=\;c'(\mu^\star),
\end{equation}
cf. \eqref{eq:foc_app}. Define
\begin{equation}\label{eq:Bstar_rep}
B^\star \;\equiv\; \frac{\big[q+H P(\mu^{FB})+\mu^{FB} H P'(\mu^{FB})\big]\,(1-\alpha)}{P'(\mu^{FB})}\ \ \ge 0.
\end{equation}

\begin{proof}[Proof of Theorem~\ref{thm:fb_implement}]
Evaluate \eqref{eq:private_foc_rep} at $\mu^\star=\mu^{FB}$ and substitute \eqref{eq:Bstar_rep}:
\[
\alpha\!\big[q+H P(\mu^{FB})\big]\;+\;\alpha\,\mu^{FB}\! H P'(\mu^{FB})\;+\;(1-\alpha)\!\big[q+H P(\mu^{FB})+\mu^{FB}\! H P'(\mu^{FB})\big].
\]
Collecting terms yields $q+H P(\mu^{FB})+\mu^{FB}H P'(\mu^{FB})$, which equals $c'(\mu^{FB})$ by \eqref{eq:planner_foc_rep}. Hence $\mu^{FB}$ satisfies the private FOC under $B^\star$. By Lemma~\ref{lem:foc_derivation} and Assumption~A1, the private objective is strictly concave, so the stationary point is unique: $\mu^\star=\mu^{FB}$. Nonnegativity of $B^\star$ follows from $P'(\mu^{FB})>0$ (the informative case) and $\alpha\in(0,1]$. \qed
\end{proof}

\paragraph{Comparative statics of $B^\star$.}
The expression \eqref{eq:Bstar_rep} is transparent: $B^\star$ shrinks with the revenue share and expands with the planner’s marginal value at the target. Differentiating \eqref{eq:Bstar_rep} while holding $(q,\bar{\mu},H)$ fixed and viewing $\mu^{FB}$ as the solution to \eqref{eq:planner_foc_rep}:
\[
\frac{\partial B^\star}{\partial \alpha}\ =\ -\,\frac{q+H P(\mu^{FB})+\mu^{FB}H P'(\mu^{FB})}{P'(\mu^{FB})}\ \le\ 0,
\]
with equality only at the degenerate point where the bracket vanishes (which cannot happen at an interior first-best). Thus, as monetization approaches full revenue sharing ($\alpha\uparrow 1$), the bounty becomes unnecessary. The dependence on $(q,H,\bar{\mu})$ works through both the direct bracket and $\mu^{FB}$; locally, increases in $q$ or $H$ raise the bracket and hence raise $B^\star$ unless offset by movements in $\mu^{FB}$ that increase $P'(\mu^{FB})$ sufficiently. The key design implication remains: if the bar is tuned so that $P'(\mu^{FB})$ is large, $B^\star$ is small because each dollar is spent at a steep margin.

\paragraph{Expected payout at the first best.}
Because the bounty is paid only upon graduation, expected spend equals $B^\star P(\mu^{FB})$. Substituting \eqref{eq:Bstar_rep} gives
\begin{equation}\label{eq:expected_payout_rep}
\mathbb{E}\big[\text{payout}\,\big|\mu^{FB}\big]\ =\ \frac{\big[q+H P(\mu^{FB})+\mu^{FB} H P'(\mu^{FB})\big]\,(1-\alpha)\,P(\mu^{FB})}{P'(\mu^{FB})}.
\end{equation}
Expression \eqref{eq:expected_payout_rep} makes precise the budget logic in the main text: for a given normative target, spend falls when the leverage ratio $P'(\mu^{FB})/P(\mu^{FB})$ rises, i.e., when the pass event is more diagnostic relative to its level.

\paragraph{Edge and boundary cases.}
If $P'(\mu^{FB})=0$ (e.g., the bar sits at an extreme so the pass probability is locally flat), then any finite bounty cannot align incentives at $\mu^{FB}$. In such cases, either the target itself is pathological (the first best lies at a corner where selection is uninformative) or the policy should move the bar to restore diagnosticity. If $P(\mu^{FB})\in\{0,1\}$ the expected payout is zero, but so is its incentive potency; again, the remedy is to reposition the bar into the informative region where $P\in(0,1)$ and $P'>0$.

\subsection*{A.7 Targeting efficiency of hit-based bounties}

We formalize the local cost comparison between a marginal increase in a hit-based bounty and a marginal increase in a uniform testing subsidy that pays per successful testing outcome. Fix a baseline policy $(q,\bar{\mu},H,\alpha)$ and an interior equilibrium $\mu^\star$. Consider two infinitesimal perturbations that increase the left-hand side of \eqref{eq:private_foc_rep} by the same amount $\varepsilon>0$ at $\mu^\star$.

\paragraph{Hit-based increment.}
An increase $\Delta B$ changes the marginal condition by $(\Delta B)P'(\mu^\star)$; to achieve $\varepsilon$, set $\Delta B=\varepsilon/P'(\mu^\star)$. The induced expected spend (to first order) is
\begin{equation}\label{eq:cost_hit}
\text{Cost}_{\text{hit}}\ =\ (\Delta B)\,P(\mu^\star)\ =\ \varepsilon\,\frac{P(\mu^\star)}{P'(\mu^\star)}.
\end{equation}

\paragraph{Uniform testing subsidy.}
Let $\sigma$ be a per-success payment during testing only. Expected testing-stage payments equal $\sigma\,\mathbb{E}[S]=\sigma\,\mu q$, and the marginal condition shifts by $\sigma q$ (a constant in $\mu$). To achieve $\varepsilon$, set $\sigma=\varepsilon/q$. The expected spend (to first order) is
\begin{equation}\label{eq:cost_flat_rep}
\text{Cost}_{\text{flat}}\ =\ \sigma\,\mu^\star q\ =\ \varepsilon\,\mu^\star.
\end{equation}

\begin{proposition}[Local dominance condition]\label{prop:dominance}
At a given interior equilibrium $\mu^\star$ with $P(\mu^\star)\in(0,1)$ and $P'(\mu^\star)>0$, a hit-based increment is locally more cost-effective than a uniform testing subsidy if and only if
\begin{equation}\label{eq:dominance_cond_rep}
\frac{P'(\mu^\star)}{P(\mu^\star)}\ \ge\ \frac{1}{\mu^\star}.
\end{equation}
\end{proposition}

\begin{proof}
Compare \eqref{eq:cost_hit} and \eqref{eq:cost_flat_rep}. The hit-based perturbation is cheaper iff $\varepsilon\,P(\mu^\star)/P'(\mu^\star)\le \varepsilon\,\mu^\star$, which simplifies to \eqref{eq:dominance_cond_rep}.
\end{proof}

\paragraph{Interpretation and scope.}
Condition \eqref{eq:dominance_cond_rep} states that targeting wins locally when the pass event is sufficiently diagnostic relative to its level, precisely when the leverage ratio $P'(\mu^\star)/P(\mu^\star)$ is high. This occurs when the bar is placed so that pass rates are neither vanishingly small nor trivially large and the Beta density $P'$ is near its peak. If \eqref{eq:dominance_cond_rep} fails (e.g., because the bar is far from the mass of entrants), the right remedy is to retune the bar; otherwise both instruments buy weak incentives per dollar. Two caveats: (i) a per-\emph{impression} subsidy during testing has zero marginal effect on the first-order condition and therefore no incentive value; it should not be compared to hit-based transfers. (ii) A per-success subsidy that also pays in continuation mimics a temporary increase in the revenue share; it can substitute for $B$ but dilutes targeting by paying in low-information states.

\subsection*{A.8 General exploitation engines and multi-winner cohorts}

We extend the implementability and targeting results to the generalized continuation landscape in which “pass” and “fail” lead to distinct expected discounted pull counts $H_1$ and $H_0$, with spread $\Delta H=H_1-H_0\ge 0$ (Section~\ref{subsec:general_engines}). The exposure aggregator is
\[
q+H_0+\Delta H\,P(\mu),
\]
and the private FOC becomes
\begin{equation}\label{eq:foc_general_rep}
\alpha\big[q+H_0+\Delta H\,P(\mu^\star)\big]\;+\;\alpha\,\mu^\star\,\Delta H\,P'(\mu^\star)\;+\;B\,P'(\mu^\star)\;=\;c'(\mu^\star).
\end{equation}
The planner’s FOC replaces $H$ by $\Delta H$ and adds $H_0$ to the level term:
\begin{equation}\label{eq:planner_general_rep}
q+H_0+\Delta H\,P(\mu^{FB})\;+\;\mu^{FB}\,\Delta H\,P'(\mu^{FB})\;=\;c'(\mu^{FB}).
\end{equation}

\begin{proposition}[Implementability with generalized continuation]\label{prop:implement_general}
Under Assumption~A1, the bounty
\begin{equation}\label{eq:Bstar_general_rep}
B^\star\ =\ \frac{\big[q+H_0+\Delta H\,P(\mu^{FB})+\mu^{FB}\,\Delta H\,P'(\mu^{FB})\big]\,(1-\alpha)}{P'(\mu^{FB})}
\end{equation}
implements $\mu^\star=\mu^{FB}$.
\end{proposition}

\begin{proof}
The proof is identical to Theorem~\ref{thm:fb_implement}: substitute \eqref{eq:Bstar_general_rep} into \eqref{eq:foc_general_rep} at $\mu^\star=\mu^{FB}$ to recover \eqref{eq:planner_general_rep}. Strict concavity yields uniqueness.
\end{proof}

\paragraph{Targeting and budgets under generalized continuation.}
Expected spend at the first best equals $B^\star P(\mu^{FB})$ with $B^\star$ from \eqref{eq:Bstar_general_rep}. The local cost comparison in Proposition~\ref{prop:dominance} is unchanged because it depends only on how the perturbations enter the private FOC at $\mu^\star$; the spread $\Delta H$ scales the slope terms in the same way for both instruments. Intuitively, when engines are “softer” (large $H_0$ and small $\Delta H$), the certain component of return is larger but the decisive power of graduation is smaller; incentives then rely more on $q+H_0$ and on positioning the bar so that $P'$ is steep. When engines are “sharper” (large $\Delta H$), the slope term is strong and hit-based targeting is especially cash-efficient.

\paragraph{Multi-winner cohorts.}
Suppose $K$ seats are intended to be filled after testing from a large pool of entrants. A practical mechanism is to graduate any entrant whose success count $S$ exceeds an integer $s_K$ chosen so that, in expectation, $K$ seats are filled (equivalently, a percentile rule on $S$). For a focal entrant with quality $\mu$, the pass probability and slope are
\[
P^{(K)}(\mu)\ =\ \sum_{k=s_K}^{q}\binom{q}{k}\mu^k(1-\mu)^{q-k},\qquad
{P^{(K)}}'(\mu)\ =\ \frac{\mu^{\,s_K-1}(1-\mu)^{\,q-s_K}}{B(s_K,q-s_K+1)}.
\]
All results above carry through with $(P,P')$ replaced by $(P^{(K)},{P^{(K)}}')$ and with $H$ or $(H_0,\Delta H)$ interpreted \emph{per seat}. In particular, implementability uses the same construction with ${P^{(K)}}'$ in the denominator, and the targeting advantage of hit-based payments is governed by the leverage ratio ${P^{(K)}}'(\mu^\star)/P^{(K)}(\mu^\star)$. The intuition is immediate: a larger cohort size $K$ lowers the implied threshold $s_K$, raises pass rates, and flattens the slope around the margin; a smaller $K$ does the opposite. Choosing $K$ and $s_K$ jointly to keep the pass event diagnostic is therefore part of the same design problem.

 Appendix~\ref{app:proofs} has established (i) existence and uniqueness of the best response under weak curvature, (ii) monotone comparative statics in all levers, (iii) the timing dominance of front-loading for fixed testing counts, and (iv) first-best implementability with a closed-form hit-based bounty, including generalized continuation and multi-winner cohorts. The proofs hinge on two analytic facts that are easy to monitor in practice: the binomial tail $P(\cdot)$ and its slope $P'(\cdot)$. Designing policies that keep $P$ in the interior and $P'$ large around the observed equilibrium is both the mathematical sufficient condition for clean comparative statics and the operational recipe for strong, cost-effective incentives.

\section{Alternative Signal Structures and Robustness}\label{app:signals}

This appendix generalizes the outcome process that underlies the testing stage and shows that the main results depend only on two primitives: the \emph{pass probability} $P(\mu)$ and its \emph{slope} with respect to true quality, $P'(\mu)$. We move beyond the Bernoulli-Binomial baseline of the main text in two directions that occur frequently on platforms. First, we allow the per-trial success probability to be a smooth, increasing transformation of quality (capturing saturation, nonlinearities, or calibration effects). Second, we allow success probabilities to vary across the $q$ testing impressions (capturing heterogeneity across slots, contexts, or audiences), which yields a Poisson-Binomial count. In each case we derive $P'(\mu)$ in closed form, restate the private and planner optimality conditions, and verify that existence/uniqueness, monotone comparative statics, front-loading, and implementability carry over with minimal changes. Throughout, we keep the notation and standing assumptions from Appendix~\ref{app:proofs}.

\subsection*{B.1 Binomial counts with an increasing link from quality to success}

In many applications, pre-entry “quality” is not itself a probability but a latent index that maps to a per-impression success chance through an increasing, differentiable link $\psi:[\underline{\mu},\overline{\mu}]\to(0,1)$. Think of $\psi(\mu)$ as the calibrated likelihood that a randomly matched user engages, given the creator’s true quality $\mu$ and the platform’s measurement stack. During testing, the entrant receives $q$ independent draws with per-trial success probability $p=\psi(\mu)$; let $S\sim\mathrm{Binomial}(q,p)$ be the number of observed successes, and let $s=\lceil q\bar{\mu}\rceil$ be the integer threshold implied by the posted bar. The pass probability is the tail
\[
P(\mu)\;=\;\Pr[S\ge s\mid p=\psi(\mu)]\;=\;I_{\psi(\mu)}\!\big(s,\;q-s+1\big),
\]
where $I_{x}(a,b)$ is the regularized incomplete Beta function. Differentiating by the chain rule and using the identity from Appendix~\ref{app:proofs} (Lemma~\ref{lem:Pprime_identity}) yields a simple closed form for the slope:
\begin{equation}\label{eq:Pprime_link}
P'(\mu)\;=\;\psi'(\mu)\,\frac{\psi(\mu)^{\,s-1}\,\big(1-\psi(\mu)\big)^{\,q-s}}{B\!\left(s,\,q-s+1\right)}\;=\;\psi'(\mu)\,\cdot\,q\,\binom{q-1}{s-1}\,\psi(\mu)^{\,s-1}\big(1-\psi(\mu)\big)^{\,q-s}.
\end{equation}
Equation \eqref{eq:Pprime_link} shows that a nonlinearity between quality and measurable success probability scales the baseline slope by the factor $\psi'(\mu)$. Intuitively, $P'(\mu)$ can be decomposed into (i) how quickly measurable success chances move with true quality ($\psi'(\mu)$) and (ii) how sensitive the pass event is to changes in those measurable chances (the Beta density term). Two immediate consequences follow.

First, all results in Sections~\ref{sec:mechanisms_main}-\ref{sec:resources_bwk} hold verbatim after replacing $P'(\mu)$ by \eqref{eq:Pprime_link}. Existence and uniqueness rely on boundedness and continuity of $P'$ and $P''$; these properties are inherited from the Beta density and the smoothness of $\psi$. Monotone comparative statics use only the sign of $\partial G/\partial\mu$ in the implicit-function argument; since $\psi'(\mu)>0$ by assumption, signs are unchanged. Front-loading depends on the fact that timing does not affect the distribution of the \emph{count} $S$ given $q$; again unchanged. Implementability (Theorem~\ref{thm:fb_implement}) goes through with the same one-line formula, now interpreted with $P$ and $P'$ defined via $\psi$:
\[
B^\star\;=\;\frac{\big[q+H P(\mu^{FB})+\mu^{FB} H P'(\mu^{FB})\big]\,(1-\alpha)}{P'(\mu^{FB})}\,.
\]
Second, the leverage ratio that governs targeting, $\Lambda(\mu)\equiv P'(\mu)/P(\mu)$, factors as
\[
\Lambda(\mu)\;=\;\psi'(\mu)\,\frac{ \psi(\mu)^{\,s-1}\big(1-\psi(\mu)\big)^{\,q-s} }{ B\!\left(s,\,q-s+1\right)\,I_{\psi(\mu)}(s,q-s+1)}\,.
\]
For a fixed bar $(q,s)$, the Beta-density fraction is maximized when $\psi(\mu)$ lies near $(s-1)/(q-1)$. The link contributes the multiplicative term $\psi'(\mu)$: if calibration compresses probabilities near the center of the support (small $\psi'$), the observed pass event becomes less diagnostic for a given movement in true quality; conversely, if $\psi'$ is large at the margin, small quality improvements translate into sizable changes in measurable success and the same bar produces sharper targeting. Managerially speaking, when a surface exhibits saturation or clipping (e.g., success rates top out mechanically), the platform can either retune the bar to lie in a region where $\psi'$ is larger or invest in measurement that restores sensitivity of observable outcomes to true quality.

\subsection*{B.2 Heterogeneous slots: Poisson-Binomial counts and influence weights}

Testing impressions are rarely homogeneous: audiences, placements, and contexts vary within the $q$-slot window. A flexible way to model such heterogeneity replaces a single success probability with a vector of slot-specific probabilities
\[
\mathbf{p}(\mu)\;=\;\big(p_1(\mu),\ldots,p_q(\mu)\big),\qquad p_t(\mu)\in(0,1),\ \ p_t'(\mu)\ge 0,
\]
so that the number of successes is a Poisson-Binomial count $S=\sum_{t=1}^q Y_t$, where $Y_t\sim\mathrm{Bernoulli}(p_t(\mu))$ are independent across $t$. The pass probability is the tail of the Poisson-Binomial distribution,
\[
P(\mu)\;=\;\Pr\!\left[\sum_{t=1}^q Y_t\ \ge\ s\ \Big|\ \mathbf{p}(\mu)\right].
\]
A key identity gives the slope in terms of \emph{influence weights} that have a direct probabilistic interpretation.

\begin{lemma}[Derivative of a Poisson-Binomial tail]\label{lem:PB_derivative}
Let $S_{-t}\equiv\sum_{u\ne t} Y_u$ be the success count excluding slot $t$. Then
\begin{equation}\label{eq:PB_derivative_identity}
\frac{d}{d\mu}\,P(\mu)\;=\;\sum_{t=1}^{q} p_t'(\mu)\,\Pr\!\big[\,S_{-t}=s-1\ \big|\ \mathbf{p}(\mu)\big].
\end{equation}
\end{lemma}

\noindent\emph{Proof.} Differentiate $P(\mu)=\sum_{A\subseteq\{1,\ldots,q\}\,:\,|A|\ge s}\ \prod_{t\in A} p_t(\mu)\ \prod_{t\notin A} (1-p_t(\mu))$ term-wise and collect the coefficient of $p_t'(\mu)$. The only terms that survive are those in which the derivative acts on $p_t$ while the remaining $s-1$ successes are provided by slots in $A\setminus\{t\}$; summing over such sets yields $\Pr[S_{-t}=s-1]$. \qed

Identity \eqref{eq:PB_derivative_identity} is intuitive: the marginal effect of increasing the success chance in slot $t$ equals the probability that the other $q-1$ slots sum to $s-1$, i.e., the event in which a success in slot $t$ flips the outcome from fail to pass. Two corollaries connect the heterogeneous-slot model back to the baseline and to design levers.

First, if all slots share the same link $p_t(\mu)\equiv\psi(\mu)$, then $\Pr[S_{-t}=s-1]=\Pr[\mathrm{Bin}(q-1,\psi(\mu))=s-1]$ and $p_t'(\mu)=\psi'(\mu)$ for all $t$. Lemma~\ref{lem:PB_derivative} reduces to
\[
P'(\mu)\;=\;q\,\psi'(\mu)\,\Pr\!\big[\mathrm{Bin}(q-1,\psi(\mu))=s-1\big],
\]
which is exactly the link-function generalization of \eqref{eq:Pprime_link} (using the combinatorial identity $q\,\Pr[\mathrm{Bin}(q-1,p)=s-1]=p^{s-1}(1-p)^{q-s}/B(s,q-s+1)$). Second, when slots differ, \eqref{eq:PB_derivative_identity} expresses $P'(\mu)$ as a \emph{weighted sum} of slot-specific sensitivities $p_t'(\mu)$ with weights $\Pr[S_{-t}=s-1]$ that depend on the rest of the window. This decomposition makes several design implications transparent.

\emph{Diagnosticity and scheduling.} The influence weights do not depend on the \emph{order} of the slots, only on the multiset of success chances in the window. Thus, as in the homogeneous case, for a fixed collection $\{p_t(\mu)\}_{t=1}^q$ the timing of guaranteed impressions does not affect the distribution of the count statistic $S$ and hence does not affect the selection terms $P(\mu)$ and $P'(\mu)$. The front-loading result from Theorem~\ref{thm:frontload} therefore remains valid: moving testing impressions earlier (increasing the discounted mass $q(\tau)$) weakly strengthens incentives without changing the diagnosticity of the pass event. If guaranteed impressions can be targeted to \emph{slots} with larger $p_t'(\mu)$ or larger influence weights, then the platform can increase $P'(\mu)$ for a given window by preferentially assigning guarantees to those contexts, an interpretation consistent with giving early opportunities in audiences where small quality improvements are most likely to flip the decision.

\emph{Comparative statics and implementability.} With $P'(\mu)$ characterized by \eqref{eq:PB_derivative_identity}, the private and planner optimality conditions retain the same structure as in \S\ref{subsec:best_response} and \S\ref{subsec:implement_fb}. Existence and uniqueness follow from the same curvature argument because $P'$ and $P''$ are bounded and continuous whenever each $p_t$ is $C^1$ with bounded derivative. The implicit-function signs are unchanged: $P(\mu)\in(0,1)$ and $P'(\mu)>0$ continue to hold, so increasing $(q,H,B,\alpha)$ rotates the private marginal-benefit schedule upward and the best response moves weakly right. The implementability bounty keeps its one-line form,
\[
B^\star\;=\;\frac{\big[q+H P(\mu^{FB})+\mu^{FB} H P'(\mu^{FB})\big]\,(1-\alpha)}{P'(\mu^{FB})},
\]
now with $P'$ computed via \eqref{eq:PB_derivative_identity}. Targeting logic remains the same: the expected spend at the first best is $B^\star P(\mu^{FB})$, which shrinks when the leverage ratio $P'(\mu^{FB})/P(\mu^{FB})$ is large. In heterogeneous windows, the leverage ratio can be increased (holding the bar fixed) by concentrating guaranteed impressions in the contexts that have high influence weights, or (holding guarantees fixed) by adjusting the bar so that typical entrants sit where the mixture of slot-specific weights is largest.

\emph{Practical estimation.} The identity \eqref{eq:PB_derivative_identity} has an applied advantage: it suggests an influence-function estimator of $P'(\mu)$ from routine telemetry. For a cohort near the current equilibrium, one can estimate $p_t'(\mu)$ by small, randomized content-quality nudges (e.g., prompts that marginally improve thumbnails or titling) and estimate $\Pr[S_{-t}=s-1]$ by leave-one-out counts within the window. Multiplying and summing across slots yields a nonparametric estimate of the slope that feeds directly into the bounty calculation and the balanced exploration dashboard. Because the same decomposition underlies the homogeneous case (with $p_t'(\mu)\equiv\psi'(\mu)$), this estimator collapses to the familiar binomial formula when slots are exchangeable.

\noindent\textbf{Summary.} Allowing a nonlinear link from quality to measurable success chances and allowing heterogeneity across slots preserves the spine of the analysis. The two objects that matter, $P(\mu)$ and $P'(\mu)$, are easily re-expressed: \eqref{eq:Pprime_link} for link-transformed binomial counts and \eqref{eq:PB_derivative_identity} for Poisson-Binomial counts. All main results continue to hold with these definitions, and the managerial guidance is unchanged: keep the pass event diagnostic around the induced equilibrium, front-load the guaranteed window to maximize its certain return, and concentrate cash where the slope is steepest.

\subsection*{B.3 Exchangeable over-dispersion via latent propensity mixing}

Empirically, early outcomes often display more variability than a Binomial model can rationalize. A convenient and widely used way to capture such \emph{over-dispersion} is to assume that, within the testing window, engagement draws are i.i.d.\ \emph{conditional} on a latent success propensity that varies across entrants or contexts. Concretely, suppose there exists a smooth, increasing map $\psi(\mu,\varepsilon)\in(0,1)$ such that conditional on $(\mu,\varepsilon)$ each testing impression succeeds with probability $p=\psi(\mu,\varepsilon)$ and the $q$ outcomes are independent. The latent factor $\varepsilon$ collects persistent audience or context heterogeneity (e.g., niche-audience fit, time-of-day mix) and is independent of $\mu$ with known distribution $F_\varepsilon$. This \emph{exchangeable} representation nests the Beta-Binomial as a special case (take $\psi(\mu,\varepsilon)\equiv \varepsilon$ with $\varepsilon\sim \mathrm{Beta}(\alpha(\mu),\beta(\mu))$), but it is more general and avoids committing to particular special functions.

Under this specification, the pass probability is a mixture of Binomial tails:
\[
P(\mu)\;=\;\mathbb{E}_{\varepsilon}\!\left[\,\Pr\!\big[S\ge s\ \big|\ p=\psi(\mu,\varepsilon)\big]\,\right]\;=\;\mathbb{E}_{\varepsilon}\!\left[\,I_{\psi(\mu,\varepsilon)}\!\big(s,\;q-s+1\big)\,\right],
\]
where $I_{x}(a,b)$ is the regularized incomplete Beta function (Appendix~\ref{app:proofs}). Differentiating under the expectation (by dominated convergence, justified by boundedness of the Beta density and smoothness of $\psi$) yields a transparent expression for the slope:
\begin{equation}\label{eq:Pprime_mixture}
P'(\mu)\;=\;\mathbb{E}_{\varepsilon}\!\left[\,\frac{\partial I_{x}(s,q-s+1)}{\partial x}\Big|_{x=\psi(\mu,\varepsilon)}\cdot \frac{\partial \psi(\mu,\varepsilon)}{\partial \mu}\,\right]\;=\;\mathbb{E}_{\varepsilon}\!\left[\,\psi_{\mu}(\mu,\varepsilon)\cdot \frac{\psi(\mu,\varepsilon)^{\,s-1}\big(1-\psi(\mu,\varepsilon)\big)^{\,q-s}}{B(s,q-s+1)}\,\right].
\end{equation}
Equation \eqref{eq:Pprime_mixture} shows that all incentive-relevant objects survive the mixture intact: the slope $P'(\mu)$ is the \emph{average} of the baseline Beta density evaluated at the latent success propensity, scaled by the sensitivity of that propensity to true quality. Two design consequences follow.

First, the main results carry over verbatim when $P'(\mu)$ is interpreted via \eqref{eq:Pprime_mixture}. Existence and uniqueness rely on boundedness and continuity of $P'$ and $P''$, which are inherited from the bounded Beta kernel and the smoothness of $\psi$; monotone comparative statics use only the sign of $\partial G/\partial\mu$ and the positivity of the kernel; front-loading depends on the fact that the distribution of the \emph{count} given the window size is order-invariant under exchangeability; and implementability retains the one-line bounty with $P$ and $P'$ computed under the mixture. Operationally, the platform does not need to identify the mixture components; it needs to monitor the \emph{net} slope of the pass event around the induced equilibrium.

Second, over-dispersion \emph{flattens} the pass frontier in the central region. To see the intuition without heavy calculus, consider the core kernel $k(p)\equiv p^{s-1}(1-p)^{q-s}/B(s,q-s+1)$ inside the expectation in \eqref{eq:Pprime_mixture}. When $1\le s\le q-1$, $k(\cdot)$ is log-concave on $(0,1)$; averaging a log-concave function over a mean-preserving spread in $p$ reduces its peak value by Jensen’s inequality. Thus, holding fixed the average sensitivity $\mathbb{E}[\psi_\mu(\mu,\varepsilon)]$, variability in $\psi(\mu,\varepsilon)$ dampens the maximum attainable $P'(\mu)$ near the center of the window, making pass events \emph{less} diagnostic for a given movement in true quality. The practical remedy aligns with the main text: either increase the testing window size (which sharpens the Beta kernel), or reposition the bar so that the typical entrant sits closer to the region where the mixture places more mass on steep parts of $k(\cdot)$, or target guaranteed impressions toward contexts $(\varepsilon)$ where $\psi_\mu(\mu,\varepsilon)$ is larger.

A familiar parametric case illustrates the calculus-free takeaway. Let the latent propensity be $p=\psi(\mu,\varepsilon)=m(\mu)$ (a mean map) with probability one, which collapses the model to the homogeneous link case in \S B.1 and yields the baseline slope $P'(\mu)=m'(\mu)\,k(m(\mu))$. Introducing a mean-preserving spread in $p$ (with the same $m(\mu)$) replaces $k(m(\mu))$ by $\mathbb{E}[k(p)]\le k(m(\mu))$, with strict inequality unless the spread is degenerate. In short, persistent heterogeneity in per-trial success chances makes the pass event less informative around the mean, so more early opportunities or a bar tuned toward the steeper region are required to sustain the same incentive power.

Finally, this exchangeable representation clarifies when the \emph{order} of testing could matter. If outcomes are conditionally i.i.d.\ given a latent propensity (de Finetti mixing), the distribution of the count $S$ depends only on the multiset of propensities and not on their order, so the front-loading optimality in the main text continues to hold: moving guaranteed impressions earlier increases the \emph{discounted} testing mass without changing $P$ or $P'$. If outcomes exhibit serial dependence that cannot be written as a mixture (e.g., transient bursts that change success chances as a function of recent outcomes), then $P$ could depend on schedule; this case falls outside the exchangeable scope and requires engine-specific modeling. In practice, many measurement stacks smooth or bucket outcomes precisely to approach exchangeability within short testing windows, which is the regime our results target.

\subsection*{B.4 Measurement noise, misclassification, and label quality}

Even with careful instrumentation, early outcomes can be noisy in ways that blunt the diagnosticity of the pass event. Two ubiquitous perturbations are \emph{misclassification} (false positives or negatives in recorded successes) and \emph{label smoothing} (systematic dampening or clipping of measured engagement). Both can be folded into our framework with simple transformations of the per-trial success chance, and both yield clean, actionable corrections to the design formulas.

Consider misclassification with a false-negative rate $\eta_0\in[0,1)$ and a false-positive rate $\eta_1\in[0,1)$ that are independent of $\mu$ and stable over the testing window. If the true per-trial success chance is $p$, the \emph{observed} success chance is
\[
\tilde{p}\;=\;\Pr[\text{observed success}]\;=\;(1-\eta_0-\eta_1)\,p\ +\ \eta_1.
\]
In the homogeneous link case of \S B.1 with $p=\psi(\mu)$, the pass probability becomes $\tilde{P}(\mu)=I_{\tilde{p}}(s,q-s+1)$ and its slope scales as
\begin{equation}\label{eq:Pprime_noise}
\tilde{P}'(\mu)\;=\;\frac{d I_{x}(s,q-s+1)}{dx}\Big|_{x=\tilde{p}}\cdot \frac{d \tilde{p}}{d\mu}
\;=\;(1-\eta_0-\eta_1)\,\frac{\tilde{p}^{\,s-1}(1-\tilde{p})^{\,q-s}}{B(s,q-s+1)}\cdot \psi'(\mu).
\end{equation}
Relative to the noise-free slope $P'(\mu)=\psi'(\mu)\,p^{s-1}(1-p)^{q-s}/B(\cdot)$, two effects appear: a constant multiplicative attenuation $(1-\eta_0-\eta_1)$, and a shift in the evaluation point from $p$ to $\tilde{p}$. The first effect is the dominant one for design: misclassification uniformly shrinks the responsiveness of the pass event to true quality. The second is typically smaller for moderate noise and can be controlled by retuning the bar to keep pass rates in the diagnostic interior.

The same scaling logic extends to the exchangeable and heterogeneous-slot settings. If the observed per-trial chance is $\tilde{p}(\mu,\varepsilon)=(1-\eta_0-\eta_1)\,\psi(\mu,\varepsilon)+\eta_1$, then the mixture slope becomes
\[
\tilde{P}'(\mu)\;=\;\mathbb{E}_{\varepsilon}\!\left[\,\psi_{\mu}(\mu,\varepsilon)\cdot (1-\eta_0-\eta_1)\cdot \frac{\tilde{p}(\mu,\varepsilon)^{\,s-1}\big(1-\tilde{p}(\mu,\varepsilon)\big)^{\,q-s}}{B(s,q-s+1)}\,\right],
\]
which is exactly the clean attenuation plus evaluation-point shift described above. Two immediate prescriptions follow.

\emph{Bounty calibration under noisy labels.} Because the implementability bounty scales inversely with $P'(\mu^{FB})$ (Theorem~\ref{thm:fb_implement}), misclassification that reduces the slope by a factor $(1-\eta_0-\eta_1)$ requires an offsetting \emph{increase} in the posted bounty by the factor $1/(1-\eta_0-\eta_1)$ to preserve alignment at the same target. In practice, a rough upper bound on $(\eta_0+\eta_1)$ (e.g., from audit samples) suffices to de-bias the bounty and maintain incentives. If misclassification varies by context, the analysis suggests prioritizing guaranteed impressions in better-measured slots, which both raises $P'(\mu)$ and cuts expected spend at the same target.

\emph{Bar retuning and pass-rate drift.} Misclassification also shifts pass rates: for a fixed bar, moving from $p$ to $\tilde{p}$ generally pushes cohorts toward the center (false negatives make high-$p$ items look less successful; false positives make low-$p$ items look more successful). Because targeting power depends on pass events being neither too rare nor too common, the dashboard should monitor the \emph{observed} pass rate and retune the bar to restore diagnosticity if noise levels change. The leverage ratio $\Lambda(\mu)=P'(\mu)/P(\mu)$ is again the sufficient statistic: when it drifts down due to worsening label quality, either improve measurement, move the bar, or temporarily raise $q$ or $B$ until measurement recovers.

A closely related perturbation is label smoothing or clipping, where the recorded outcome is a dampened version of the underlying binary success (for instance, because the measurement stack thresholds continuous dwell signals or caps credit for rapid repeats). Such smoothing can be modeled as a compression of the link, $\tilde{p}=\phi(\psi(\mu))$ with $\phi'(p)\in(0,1)$ near the margin. The slope formula becomes $\tilde{P}'(\mu)=\phi'(\psi(\mu))\,\psi'(\mu)\cdot \text{BetaDensity}(\tilde{p})$, which is the same attenuation structure: weaker label sensitivity translates one-for-one into weaker incentive leverage and a proportionally larger bounty to reach the same target. From a budgeting standpoint, dollars spent on improving label sensitivity (raising $\phi'$ by better instrumentation or denoising) are \emph{substitutes} for dollars spent on bounties: both purchase slope at the graduation margin.

Finally, a comment on adversarial manipulation. If a fraction $\zeta$ of entrants can inflate observed success on a subset of impressions (e.g., bot traffic) without affecting long run value, the design response is not to abandon thresholds but to \emph{change the pass event}. The analysis in Appendix~\ref{app:proofs} goes through unchanged when the test statistic is built from higher-persistence signals (e.g., delayed likes, deep-view milestones) for which manipulation has lower yield. In formulas, replace the per-trial chance $p$ by a chance $p^{\mathrm{pers}}$ tied to a signal with higher persistence; all slopes and budgets recompute mechanically. The take-away here is simple: the pass event is a design object, choose it so that its slope with respect to \emph{desirable} quality is large and its slope with respect to manipulative effort is small, and then use the same early-exposure and bounty toolkit on that event.

\subsection*{B.5 Temporal dependence, novelty decay, and large-$q$ approximations}

The analysis so far assumed that conditional on primitives the $q$ testing outcomes are exchangeable (the Poisson-Binomial and mixture cases) or, at minimum, that the pass event depends on the window \emph{only} through the multiset of per-trial success probabilities. In many surfaces, however, audience receptivity drifts within the testing window: novelty may be strongest immediately after launch and decay thereafter; time-of-day composition may shift systematically; early algorithmic placements may command higher intent. This subsection shows that the core prescriptions are unchanged under a broad class of \emph{deterministic} drifts, and it provides a normal-approximation formula for $P(\mu)$ and $P'(\mu)$ that is useful for calibration when $q$ is moderate to large.

\paragraph{Deterministic multiplicative drift.}
Let the per-trial success chance in slot $t$ be $p_t(\mu)=\theta_t \,\psi(\mu)$ with $\theta_t\in(0,1]$ capturing deterministic audience drift (e.g., novelty weights) common to all entrants, and $\psi(\mu)$ an increasing link from true quality to baseline success chance. The number of successes is $S=\sum_{t=1}^{q} Y_t$ with $Y_t\sim\mathrm{Bernoulli}(p_t(\mu))$ independent across $t$. The pass probability is a Poisson-Binomial tail, and by Lemma~\ref{lem:PB_derivative} its slope is
\begin{equation}\label{eq:PB_drift_slope}
P'(\mu)\;=\;\sum_{t=1}^{q} \underbrace{\theta_t\,\psi'(\mu)}_{\text{per-slot sensitivity}} \;\cdot\; \underbrace{\Pr\!\big[S_{-t}=s-1\ \big|\ \{p_u(\mu)\}_{u\neq t}\big]}_{\text{influence weight of slot $t$}}.
\end{equation}
When $\{\theta_t\}$ is nonincreasing in $t$ (novelty decay), the largest per-slot sensitivities $\theta_t\,\psi'(\mu)$ occur at the earliest times. For a given \emph{cardinality} $Q$ of guaranteed impressions, selecting the $Q$ earliest slots maximizes the right-hand side of \eqref{eq:PB_drift_slope} by the rearrangement logic that pairs the largest sensitivities with the largest influence weights.\footnote{Formally, when $\theta_t$ is common across entrants and the bar $s$ is fixed, the family of influence weights $\{\Pr[S_{-t}=s-1]\}_t$ is \emph{majorized} by the vector that places more mass on slots with higher $p_t$; see Marshall and Olkin’s majorization theory for Poisson-Binomial tails. The heuristic is straightforward: in windows where some $p_t$ are larger, it is more likely that the other $q-1$ slots already sum to $s-1$, so a marginal increase in a high-$p_t$ slot is especially consequential.} Thus novelty decay strengthens, rather than weakens, the case for front-loaded guarantees: among all $Q$-slot schedules, the earliest schedule (largest $\sum_{j=1}^{Q}\theta_{t_j}$) weakly maximizes both $P(\mu)$ and $P'(\mu)$ for every $\mu$ and therefore weakly increases the induced effort and platform objective relative to any delayed schedule.

\paragraph{Large-$q$ normal approximation.}
For calibration and diagnostics it is often convenient to approximate the Poisson-Binomial tail by a normal CDF when $q$ is moderate to large. Let
\[
m(\mu)\;\equiv\;\sum_{t=1}^q p_t(\mu),\qquad v(\mu)\;\equiv\;\sum_{t=1}^q p_t(\mu)\big(1-p_t(\mu)\big),
\]
and define the standardized threshold $z(\mu)\equiv \frac{s- m(\mu)}{\sqrt{v(\mu)}}$. A continuity-corrected DeMoivre-Laplace approximation gives
\begin{equation}\label{eq:CLT_tail}
P(\mu)\ \approx\ 1-\Phi\!\Big(z(\mu)-\tfrac{1}{2}/\sqrt{v(\mu)}\Big),
\end{equation}
where $\Phi$ and $\phi$ are the standard normal CDF and PDF. Differentiating \eqref{eq:CLT_tail} yields an explicit slope:
\begin{equation}\label{eq:CLT_slope}
P'(\mu)\ \approx\ \phi\!\big(z(\mu)\big)\,\cdot\,\Bigg(\frac{m'(\mu)}{\sqrt{v(\mu)}}\ + \ \frac{(s-m(\mu))\,v'(\mu)}{2\,v(\mu)^{3/2}}\Bigg),
\end{equation}
with
\[
m'(\mu)=\sum_{t=1}^q p_t'(\mu),\qquad v'(\mu)=\sum_{t=1}^q p_t'(\mu)\,(1-2p_t(\mu)).
\]
In the multiplicative-drift case $p_t(\mu)=\theta_t\psi(\mu)$, $m'(\mu)=\psi'(\mu)\sum_t \theta_t$ and $v'(\mu)=\psi'(\mu)\sum_t \theta_t(1-2\theta_t\psi(\mu))$, which makes \eqref{eq:CLT_slope} especially transparent. The approximation captures the two channels emphasized in the main text: moving the mean early-success count relative to the bar (the $m'(\mu)$ term) and changing the dispersion of early outcomes (the $v'(\mu)$ term). It is accurate when $q$ is modestly large, when no single slot dominates ($\max_t p_t$ not too close to 0 or 1), and when $s$ is not in an extreme tail. In practice, plotting the empirical pass-probability curve against the normal approximation is an inexpensive diagnostic; the slope formula \eqref{eq:CLT_slope} then feeds directly into the bounty calculation and the balanced exploration dashboard.

\paragraph{Implications.}
Deterministic drift that is common across entrants can be absorbed into the $p_t(\mu)$ vector without altering any of the structural results; it only sharpens the operational recommendation to deploy guarantees as early as feasible and, when targeting is possible, into high-sensitivity contexts. Normal approximations provide closed-form surrogates for $P$ and $P'$ that are robust enough for planning and sensitivity analysis, especially when budgets or governance require “back-of-the-envelope” estimates before running pilots.

\subsection*{B.6 Estimating $P(\mu)$ and $P'(\mu)$ from telemetry: a practical recipe}

Designing guarantees, bars, and bounties requires estimates of the pass probability and its slope around the current equilibrium. This subsection outlines a practical, low-assumption workflow that converts routine telemetry into the two inputs needed by the formulas in the main text. The emphasis is on methods that an applied team can run with standard analytics tooling; asymptotic details are kept to a minimum.

\paragraph{Step 1: Choose a quality proxy and bin entrants.}
Because true pre-entry quality is latent, begin with a pre-policy proxy that is available at entry and predictive of early success (e.g., an offline content-quality score, a historical creator-level metric, or a lightweight model trained on pre-policy cohorts). Sort recent entrants by this proxy and form coarse bins (e.g., deciles). For each bin, compute the empirical pass rate $\widehat{P}$ under the current $(q,s)$ and the average value of the proxy.

\paragraph{Step 2: Fit a smooth pass-probability curve.}
Regress the pass indicator on the proxy using a flexible but monotone link (e.g., logistic regression with splines, isotonic regression, or a generalized additive model). This yields a smooth estimate $\widehat{P}(\cdot)$ of the pass probability as a function of the proxy. Evaluate $\widehat{P}$ at the cohort median to obtain a point estimate of the current pass rate and read off the estimated slope $\widehat{P}'$ at that point by automatic differentiation or by finite differences. If the proxy is roughly linear in true quality near the median, $\widehat{P}'$ is a good stand-in for $P'(\mu^\star)$; if not, rescale the proxy locally (e.g., z-score within a neighborhood) to make the units interpretable.

\paragraph{Step 3: Cross-check with the influence-weights estimator.}
Independently estimate $P'(\mu^\star)$ using the identity in Lemma~\ref{lem:PB_derivative}. Within the testing window, compute the empirical distribution of leave-one-out counts $\widehat{\Pr}[S_{-t}=s-1]$ by slot. Next, induce small, randomized nudges that marginally improve observable aspects of quality (e.g., prompting a clearer thumbnail or a better title) and estimate the per-slot sensitivity $\widehat{p_t'}$ from the resulting change in observed per-slot success rates. Multiply and sum to obtain
\[
\widehat{P}'_{\text{IF}}\;=\;\sum_{t=1}^{q} \widehat{p_t'} \cdot \widehat{\Pr}[S_{-t}=s-1].
\]
Agreement between the curve-slope estimate from Step 2 and the influence-weights estimate $\widehat{P}'_{\text{IF}}$ is a strong sanity check; discrepancies point to misspecification (e.g., heterogeneity or nonlinearity in the proxy) or measurement drift.

\paragraph{Step 4: Plug in and propagate uncertainty.}
Compute the implementability bounty and budget quantities using the plug-in formulas with $(\widehat{P},\widehat{P}')$. For example,
\[
\widehat{B}^\star\;=\;\frac{\big[q+\widehat{H}\, \widehat{P}+\widehat{\mu}^{FB}\, \widehat{H}\, \widehat{P}'\big]\,(1-\widehat{\alpha})}{\widehat{P}'},\qquad
\widehat{\text{Spend}}\;=\;\widehat{B}^\star\,\widehat{P}.
\]
Obtain uncertainty bands by nonparametric bootstrap over entrants (resampling creators with replacement and recomputing all four steps) or by delta-method linearization when using the smooth curve from Step 2. Report at least percentile intervals for $\widehat{P}$, $\widehat{P}'$, $\widehat{B}^\star$, and expected spend; in dashboards, show how $\widehat{B}^\star$ would change under a conservative attenuation of the slope (e.g., multiply $\widehat{P}'$ by $0.8$ to reflect potential label noise as in \S B.4).

\paragraph{Step 5: Stress-test bar placement and schedule.}
Using the fitted pass-probability curve, simulate counterfactual pass rates and slopes under nearby bars $s\pm 1$ and under modest changes in testing size $q\pm 1$. If novelty decay is present (Step 1 can diagnose this by plotting per-slot success rates), prioritize schedules that move guarantees earlier and recompute the implied $P$ and $P'$ via the normal approximation \eqref{eq:CLT_tail}-\eqref{eq:CLT_slope} or via direct Poisson-Binomial evaluation. Choose the policy that restores the leverage ratio $\widehat{\Lambda}=\widehat{P}'/\widehat{P}$ to the interior region where targeting is potent and expected spend per unit incentive is low.

\paragraph{What not to do.}
Two practices routinely degrade incentives and budgets. First, tuning bars solely to hit a fixed pass-rate target without monitoring the \emph{slope} encourages regimes with low diagnosticity; a 50\% pass rate with a flat slope is not a good target. Always pair pass rates with slopes. Second, paying flat per-impression subsidies during testing has no effect on the first-order condition and only consumes cash; if dollars must be spent for fairness or promotional reasons, route them through guarantees ($q$) or through the hit-based bounty where they purchase slope.

\noindent\textbf{Summary.}
Temporal drift and large-$q$ regimes can be folded into the same two primitives that drive the theory: $P(\mu)$ and $P'(\mu)$. Deterministic novelty decay strengthens the case for early guarantees and, when targeting is available, for placing them in high-sensitivity contexts. Normal approximations provide quick, interpretable surrogates for pass probabilities and slopes. Finally, a simple telemetry workflow, fit a pass curve, cross-check with influence weights, plug in, bootstrap, and stress-test, delivers the quantities needed to set guarantees, bars, and bounties in a way that is rigorous, auditable, and easy to operationalize.

\section{General Exploitation Engines and Continuation Mapping}\label{app:engines}

This appendix clarifies how the main-text exposure aggregator
\[
\Xi(\mu)\;=\;q\;+\;H_0\;+\;\Delta H\cdot P(\mu),\qquad \Delta H\equiv H_1-H_0\ge 0,
\]
arises from a wide class of post-testing allocation engines that operate with posterior-based priorities (indices, sampling rules, or tiered bands). The goal is twofold. First, we provide a clean reduction from the mechanics of a realistic ranking engine to two scalars $(H_0,\Delta H)$ that summarize \emph{expected discounted continuation exposure} conditional on \emph{failing} versus \emph{passing} the testing bar. Second, we describe how to estimate these objects in practice by simulation or replay on logged data, emphasizing simple diagnostics that ensure the mapping is reliable.

Throughout, we keep the testing stage exactly as in the main text: the platform grants $q$ guaranteed impressions and applies a graduation test that \emph{passes} an entrant if and only if the observed success count $S$ exceeds an integer threshold $s=\lceil q\bar{\mu}\rceil$. Write $P(\mu)=\Pr[S\ge s]$ and recall that $P'(\mu)$ is the Beta density (Appendix~\ref{app:proofs}). Let $D\in\{\text{pass},\text{fail}\}$ be the outcome of the test. After $D$ is realized, the engine takes over and allocates impressions stochastically over the remaining horizon based on the entrant’s posterior and on the competing inventory.

\subsection*{C.1 From engines to $(H_0,\Delta H)$: a structural reduction}

We begin with a general engine that, in each period $t\ge q+1$, assigns a probability $x_t\in[0,1]$ that the entrant is surfaced in the focal slot (or, more generally, the expected share of relevant impressions). The state that determines $x_t$ may include (i) the entrant’s posterior over $\mu$ after the testing window and any subsequent outcomes, (ii) the posteriors of other items in the candidate set, and (iii) time or context variables. Define the discounted continuation exposure as
\[
H(D)\ \equiv\ \mathbb{E}\!\left[\ \sum_{t\ge q+1}\gamma^{t-1}\,x_t\ \bigg|\ D\ \right],\qquad H_1\equiv H(\text{pass}),\ \ H_0\equiv H(\text{fail}),
\]
where the expectation averages over the engine’s randomness, the stochastic evolution of posteriors, and any randomness in the competitive set. The following proposition states that the exposure aggregator in the main text is exact whenever post-testing allocation depends on $D$ only through its induced posterior and whenever subsequent play is conditionally independent of the \emph{counterfactual} outcomes that were not observed in testing.

\begin{proposition}[Band-separable continuation]\label{prop:band_separable}
Suppose the engine satisfies:
\begin{enumerate}
\item[(i)] (\emph{Posterior sufficiency}) Given the entrant’s posterior at $t=q+1$, the distribution of future allocations $\{x_t\}_{t\ge q+1}$ is conditionally independent of the \emph{path} of testing outcomes beyond the pass/fail event $D$.
\item[(ii)] (\emph{No treatment on the untreated}) The testing window contains the only guaranteed exposures; after $t=q$, allocation is governed by the engine with no additional policy shocks.
\end{enumerate}
Then
\[
\mathbb{E}\!\left[\ \sum_{t\ge q+1}\gamma^{t-1}\,x_t\ \right]\;=\;H_0\;+\;(H_1-H_0)\cdot P(\mu)\;=\;H_0+\Delta H\cdot P(\mu).
\]
Consequently, the total expected discounted exposure equals $\Xi(\mu)=q+H_0+\Delta H\,P(\mu)$.
\end{proposition}

\noindent\emph{Proof (sketch).} By the law of iterated expectations,
\[
\mathbb{E}\!\left[\sum_{t\ge q+1}\gamma^{t-1}x_t\right]\ =\ \sum_{d\in\{\text{pass},\text{fail}\}} \Pr[D=d]\cdot \mathbb{E}\!\left[\sum_{t\ge q+1}\gamma^{t-1}x_t\ \Big|\ D=d\right]\ =\ P(\mu)\,H_1+\big(1-P(\mu)\big)\,H_0.
\]
Assumption (i) ensures that $H(d)$ is a well-defined functional of the posterior induced by $D=d$ and not of the full realization of testing outcomes; (ii) ensures that no other policy terms enter continuation. \qed

\begin{figure}[!ht]
  \centering
  \includegraphics[width=.55\linewidth]{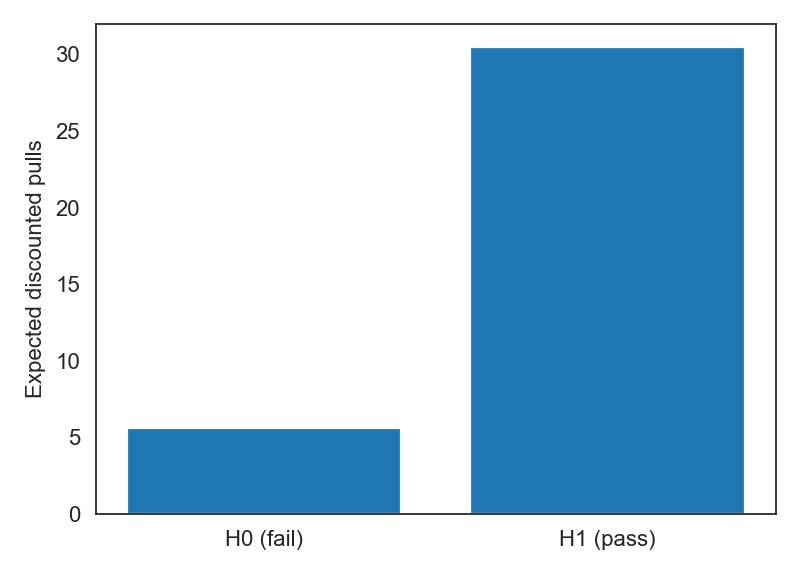}
  \caption{\textbf{Continuation spread from Thompson sampling.}
  One focal slot competes with $20$ fixed competitors. Conditional on pass, the expected discounted pulls $H_1$ exceed $H_0$ after fail; the spread $\Delta H=H_1-H_0$ is the prize in the main text.}
  \label{fig:appC-H}
\end{figure}

Proposition~\ref{prop:band_separable} nests most engines used in practice:

\emph{Tiered-band engines.} The platform assigns items to bands (e.g., “promising” vs. “tail”) based on posterior thresholds; the entrant’s expected discounted exposure within each band is $H_1$ or $H_0$. Crossing the testing bar is the admission rule into the higher band, so the reduction is exact.

\emph{Index policies (Gittins/UCB).} At $t=q+1$ the entrant’s posterior initializes an index; subsequent allocations depend on this posterior and on competing indices but not on the unrealized counterfactual test path. $H_1$ (resp. $H_0$) is the expected discounted pull count starting from the posterior conditional on $D=\text{pass}$ (resp. $D=\text{fail}$).

\emph{Thompson sampling.} Similar logic: $H_1$ and $H_0$ are the expected discounted pull counts when the initial posterior is the one induced by passing or failing the bar. Because the test window is short, these two posteriors are well separated precisely when $s$ is in the diagnostic interior.

Two remarks help with interpretation. First, $H_0$ behaves like an additive contribution to the guaranteed window $q$: even failed entrants are occasionally sampled by the engine, especially under exploration-heavy rules. Second, $\Delta H$ is the \emph{prize spread}, how much more continuation a pass buys over a fail in expectation. It is this spread that fuels the slope term in the entrant’s FOC, and its magnitude can be shaped both by the engine (more distinct priority bands) and by the bar (steeper separation between the two posteriors).

\subsection*{C.2 Computing and estimating $(H_0,\Delta H)$}

We turn to practical methods for obtaining $(H_0,\Delta H)$. The key is to work \emph{conditional on the posteriors} that result from testing, because those are sufficient states for all engines listed above.

\paragraph{Posterior after testing.}
Fix a Beta prior $\mathrm{Beta}(a_0,b_0)$ for expository concreteness (other conjugate families are analogous). After $q$ trials with $S$ successes, the posterior is $\mathrm{Beta}(a_0+S,\,b_0+q-S)$. Conditional on $D=\text{pass}$, $S\ge s$ and the posterior is a mixture over $S\in\{s,\ldots,q\}$ with weights proportional to $\binom{q}{S}\mu^S(1-\mu)^{q-S}$. The \emph{ex ante} continuation values $H_1$ and $H_0$ integrate over these $S$-mixtures; in practice we approximate by (i) conditioning on the \emph{typical} $S$ near the tail mean, or (ii) Monte Carlo over $S$ with the Binomial weights and over $\mu$ with the prior.

\paragraph{Direct simulation (“engine replay”).}
When the engine is available in a sandbox (or can be faithfully simulated), the most reliable approach is Monte Carlo:
\begin{enumerate}
\item Draw a large number of $(\mu,S)$ pairs by sampling $\mu$ from the prior and $S\sim\mathrm{Bin}(q,\mu)$. Separate the draws into $S\ge s$ (pass) and $S<s$ (fail).
\item For each draw, initialize the entrant’s posterior as $\mathrm{Beta}(a_0+S,b_0+q-S)$ at time $t=q+1$ and simulate the engine forward for a long horizon $T$ (or with truncation at discount $\gamma$), logging whether the entrant is shown in each period and the realized outcomes that update its posterior. Keep competitors’ posteriors fixed to a representative sample or draw them from their own steady-state distribution.
\item Compute the discounted pull count $\sum_{t=q+1}^{q+T}\gamma^{t-1}\mathbf{1}\{\text{shown at }t\}$ and average across the pass draws (to get $H_1$) and the fail draws (to get $H_0$).
\end{enumerate}
This procedure yields $(H_0,\Delta H)$ for the current engine and cohort composition. It transparently handles Thompson sampling, UCB variants, or banded ranking rules and allows stress tests (e.g., more aggressive priority spreads).

\paragraph{Offline replay on logs.}
If the engine cannot be simulated but historical logs contain (i) the ranking scores that would have been used to allocate the focal slot and (ii) the realized exposure decision, one can approximate $H_1$ and $H_0$ by \emph{counterfactual reweighting}:
\begin{enumerate}
\item Identify past entrants with $S\ge s$ and with $S<s$ under the same test window. For each, reconstruct the posterior path the engine would have observed (using logged outcomes).
\item Using the logged ranking scores and exposure decisions, estimate a parametric allocation model $x_t=\sigma(\text{score vectors}_t)$ (e.g., a multinomial logit or a calibrated link) that maps scores into exposure probabilities.
\item For each posterior path, predict $\hat{x}_t$ under the estimated allocation model and compute $\sum_{t\ge q+1}\gamma^{t-1}\hat{x}_t$. Average across pass cases to get $\hat{H}_1$ and across fail cases to get $\hat{H}_0$.
\end{enumerate}
Two cautions apply. First, if the historical engine explored differently in pass and fail cases, the allocation model must include the exploration policy to avoid bias. Second, logs that reflect different bars or windows should be filtered to match the policy under evaluation.

\paragraph{Closed-form surrogates for canonical engines.}
In simple environments, $(H_0,\Delta H)$ admits quick surrogates:

\emph{Two-band priority rule.} Suppose passing places the entrant in a high-priority queue that is sampled with probability $\pi_H$ each period and failing places it in a low-priority queue sampled with probability $\pi_L$ (competitor mix held fixed). Then
\[
H_1\approx \frac{\pi_H}{1-\gamma},\qquad H_0\approx \frac{\pi_L}{1-\gamma},\qquad \Delta H\approx \frac{\pi_H-\pi_L}{1-\gamma}.
\]
These expressions are exact if the per-period sampling probabilities are constant and independent of the entrant’s subsequent outcomes; they remain accurate when posterior updates change probabilities slowly relative to discounting.

\emph{Thompson sampling against a stable incumbent.} With a single incumbent of known quality $\bar{\mu}$ and one focal slot per period, the probability that the entrant is shown at $t$ equals $\Pr[\tilde{\mu}_t\ge \bar{\mu}]$ where $\tilde{\mu}_t$ is a Beta draw from the entrant’s posterior at $t$. If the posterior mean starts above $\bar{\mu}$ in the pass case and below in the fail case, a convenient approximation uses the Beta CDF at $\bar{\mu}$ along a deterministic posterior-mean path and sums geometrically; the result is larger by an order of magnitude in the pass case, precisely the prize spread $\Delta H$.

\paragraph{Diagnostics and governance.}
Whichever method is used, three diagnostics make the mapping credible. First, \emph{stability}: re-estimate $(H_0,\Delta H)$ on adjacent cohorts; large swings signal that the engine or competitor mix has changed in ways not captured by the procedure. Second, \emph{band separation}: verify that the posterior means after pass versus fail are well separated under the chosen $(q,s)$; small separations imply a small $\Delta H$ and weak slope incentives. Third, \emph{consistency}: plug the estimated $(H_0,\Delta H)$ into the main-text formulas and check whether predicted graduation and continuation statistics align with realized telemetry in a holdout period.

\noindent\textbf{Implications for design.} The $(H_0,\Delta H)$ decomposition reframes engine tuning as incentive tuning. Increasing baseline resurfacing of failed entrants raises $H_0$ and acts like more guaranteed impressions; increasing priority for passed entrants raises $\Delta H$ and strengthens the slope term that rewards quality at the margin. The testing bar interacts with both: for a fixed engine, choosing $(q,s)$ that produces a clear separation of posteriors maximizes the \emph{effective} spread $\Delta H$ seen by typical entrants, lowering the cash bounty needed for alignment.

 The remaining half of Appendix~\ref{app:engines} extends the mapping to multi-winner (\emph{top-K}) cohorts and rotating carousels, provides a lightweight approximation for $(H_0,\Delta H)$ under UCB-style indices with competitor drift, and details an online calibration loop that co-tunes $(q,s)$ and $(H_0,\Delta H)$ to keep the pass event diagnostic while maintaining budget discipline.

\subsection*{C.3 Multi-winner cohorts and rotating carousels}

Many production surfaces promote more than one item at a time. Feeds show several cards, shelves display a row of thumbnails, and carousels rotate a fixed number of slots. The reduction in Proposition~\ref{prop:band_separable} extends with minimal change once continuation is summarized \emph{per period} by the probability that an entrant appears among the promoted set. Let $K$ be the number of concurrent seats and let $\pi_t(\theta_t)$ denote the probability that an entrant with period-$t$ priority index $\theta_t$ is included in the top-$K$ set given the contemporaneous distribution of competitors’ indices. Define the per-period \emph{effective weight} of seat $r\in\{1,\ldots,K\}$ as $w_r\in(0,1]$ to capture position effects (e.g., higher click-through in earlier slots). The period-$t$ expected share of impressions is then
\[
x_t\;=\;\sum_{r=1}^{K} w_r\,\Pr\!\big[\text{entrant is ranked at position }r\ \big|\ \theta_t,\ \text{competitors}\big].
\]
Aggregating over positions yields $x_t = \omega\cdot \pi_t(\theta_t)$ with $\omega\equiv\sum_{r=1}^K w_r$ whenever the rank distribution is smooth in $\theta_t$ (e.g., Plackett-Luce, softmax). The continuation pair becomes
\[
H_1\;=\;\mathbb{E}\!\Big[\sum_{t\ge q+1}\gamma^{t-1}\omega\,\pi_t(\theta_t^{\text{pass}})\Big],\qquad
H_0\;=\;\mathbb{E}\!\Big[\sum_{t\ge q+1}\gamma^{t-1}\omega\,\pi_t(\theta_t^{\text{fail}})\Big],
\]
and the prize spread is $\Delta H=H_1-H_0$. Two approximations make these objects transparent and easy to calibrate.

A constant-index benchmark assumes that, over the discount horizon, an entrant’s index is well summarized by its post-testing initialization $\theta_0$ (posterior mean plus any engine-specific optimism bonus). Then $\pi_t(\theta_t)\approx \pi(\theta_0)$ and
\[
H_1\ \approx\ \frac{\omega\,\pi(\theta_0^{\text{pass}})}{1-\gamma},\qquad
H_0\ \approx\ \frac{\omega\,\pi(\theta_0^{\text{fail}})}{1-\gamma},\qquad
\Delta H\ \approx\ \frac{\omega\,[\pi(\theta_0^{\text{pass}})-\pi(\theta_0^{\text{fail}})]}{1-\gamma}.
\]
This “frozen-index” surrogate is particularly accurate when testing windows are short and the initial posterior gap created by passing versus failing is large relative to subsequent within-horizon learning.

A relaxation-to-steady-state benchmark allows the inclusion probability to drift toward a cohort-level steady state $\pi_\infty$ at rate $\lambda>0$: $\pi_t(\theta_t)\approx \pi_\infty - (\pi_\infty-\pi_0)\,e^{-\lambda\,(t-q)}$, with $\pi_0$ equal to the inclusion probability implied by the post-testing index. Summing the resulting geometric series gives
\[
H\big(\pi_0,\pi_\infty,\lambda\big)\;=\;\omega\left[\frac{\pi_\infty}{1-\gamma}\ -\ \frac{\pi_\infty-\pi_0}{1-\gamma e^{-\lambda}}\right],
\]
and therefore
\[
\Delta H\ \approx\ H\!\big(\pi_0^{\text{pass}},\pi_\infty,\lambda\big)\ -\ H\!\big(\pi_0^{\text{fail}},\pi_\infty,\lambda\big)\ =\ \omega\,\frac{\pi_0^{\text{pass}}-\pi_0^{\text{fail}}}{1-\gamma e^{-\lambda}}.
\]
The sole new parameter, $\lambda$, measures how quickly inclusion probabilities mean-revert as posteriors update; it is readily estimated from logs by regressing realized inclusion on its own lag and a constant within a short post-testing window. In rotating carousels that serialize $K$ seats over a few time ticks, the same formulas apply with $\gamma$ interpreted over the carousel cycle and $\omega$ equal to the cycle‑averaged position weight.

These approximations preserve the managerial interpretation. The baseline $H_0$ acts like extra guaranteed exposure from ambient resurfacing in the promoted set, while the spread $\Delta H$ captures how much more the engine will include an entrant who cleared the bar. Engines with more distinct priority bands (or larger $K$ with steep position weights) have larger spreads and therefore stronger slope incentives for the same testing policy.

\subsection*{C.4 UCB-style indices with competitor drift: a lightweight surrogate}

Upper-confidence-bound (UCB) engines allocate by an optimism‑corrected estimate: an entrant with posterior mean $\hat{\mu}_t$ and pull count $n_t$ receives index $\theta_t=\hat{\mu}_t + b_t$, where $b_t=c\,\sqrt{(\log t)/(n_t\vee 1)}$ for some exploration constant $c>0$. Passing the bar raises both $\hat{\mu}_{q+1}$ and $n_{q+1}$ relative to failing; subsequently, the optimism term $b_t$ decays as the item is pulled. Competitors’ indices drift as their own counts and means evolve. Directly simulating this coupled system is feasible (and recommended when precision matters), but a tractable surrogate gives intuition and closed forms.

Two simplifications produce an accurate first pass. First, approximate the \emph{selection frontier} formed by competitors’ indices by a slowly varying threshold $\bar{\theta}_t$ such that the entrant is included in the promoted set when $\theta_t\ge \bar{\theta}_t$. Second, replace the hard threshold by a smooth link $\pi_t(\theta_t)\approx \sigma\!\big(\kappa(\theta_t-\bar{\theta}_t)\big)$ with slope parameter $\kappa>0$ (a logistic smoothing of the inclusion event). After testing, take $\theta_{q+1}^{\text{pass}}=\hat{\mu}_{q+1}^{\text{pass}}+b_{q+1}$ and $\theta_{q+1}^{\text{fail}}=\hat{\mu}_{q+1}^{\text{fail}}+b_{q+1}$, where the common optimism bonus reflects the same $q$ trials but different realized $S$; write $\pi_0^{\text{pass}}=\sigma(\kappa(\theta_{q+1}^{\text{pass}}-\bar{\theta}_{q+1}))$ and analogously for fail. As the engine operates, $\hat{\mu}_t$ drifts toward the true mean while $b_t$ shrinks roughly like $1/\sqrt{n_t}$. Because $n_t$ itself grows in proportion to inclusion, a standard linearization yields an exponential relaxation of inclusion probabilities:
\[
\pi_t\ \approx\ \pi_\infty\ -\ \big(\pi_\infty-\pi_0\big)\,e^{-\lambda\,(t-q)},\qquad \lambda\ \approx\ \kappa\,\frac{\partial \theta}{\partial n}\cdot \pi_\infty,
\]
where $\partial \theta/\partial n \approx -\tfrac{1}{2}c\,(\log t)^{1/2} n^{-3/2}$ evaluated near the post-testing count gives a negative slope that translates into positive mean reversion for inclusion. The same relaxation-sum as above then gives
\[
H_1\ \approx\ \omega\left[\frac{\pi_\infty}{1-\gamma}-\frac{\pi_\infty-\pi_0^{\text{pass}}}{1-\gamma e^{-\lambda}}\right],\qquad
H_0\ \approx\ \omega\left[\frac{\pi_\infty}{1-\gamma}-\frac{\pi_\infty-\pi_0^{\text{fail}}}{1-\gamma e^{-\lambda}}\right],
\]
and
\[
\Delta H\ \approx\ \omega\,\frac{\pi_0^{\text{pass}}-\pi_0^{\text{fail}}}{1-\gamma e^{-\lambda}}.
\]
Competitor drift can be folded into $\bar{\theta}_t$ by letting it follow a stable AR(1), $\bar{\theta}_{t+1}=\rho\,\bar{\theta}_t+(1-\rho)\bar{\theta}^\star+\varepsilon_t$, and absorbing the resulting slow motion into the steady-state inclusion $\pi_\infty$ (higher $\bar{\theta}^\star$ implies smaller $\pi_\infty$). In practice, $\pi_0^{\text{pass}}$, $\pi_0^{\text{fail}}$, $\pi_\infty$, and $\lambda$ are estimated directly from logs: fit a smooth curve of post-testing inclusion probability against the post-testing index, read off the two initial probabilities, track the cohort‑average inclusion over the next few dozen periods to estimate $\lambda$, and infer $\pi_\infty$ from the plateau. These four numbers, together with $\omega$ and $\gamma$, determine $(H_0,\Delta H)$ to first order and therefore suffice for bounty and budget calculations.

Two cautions guide use. First, when the promoted set is extremely small (e.g., $K=1$) and the frontier is very sharp, the logistic smoothing may require a large $\kappa$; the formulas remain valid but numerical stability benefits from capping $\kappa$ based on observed variance of score differences. Second, when the engine applies additional resets or “freshness boosts,” those should be treated as part of the policy and reflected in the estimated inclusion path; failing to do so understates $\lambda$ and overstates the spread.

\subsection*{C.5 An online calibration loop that co-tunes $(q,s)$ and $(H_0,\Delta H)$}

In steady operation, the testing policy and the engine should be tuned together so that the pass event is diagnostic, the prize spread is meaningful, and budgets are met. The following loop is designed for weekly execution and uses only quantities that the platform can monitor in dashboards.

Start with current $(q,s)$, measured $(\widehat{H}_0,\widehat{\Delta H})$ from simulation or replay, and observed cohort pass statistics $(\widehat{P},\widehat{P}')$ near the induced equilibrium. Compute the implied implementability bounty
\[
\widehat{B}^{\star}\;=\;\frac{\big[q+\widehat{H}_0+\widehat{\Delta H}\,\widehat{P}+\widehat{\mu}^{FB}\,\widehat{\Delta H}\,\widehat{P}'\big]\,(1-\widehat{\alpha})}{\widehat{P}'},
\]
and the expected per-entrant spends: impressions $\widehat{q}$ (delivered, discounted) and cash $\widehat{B}^{\star}\widehat{P}$. Compare these with budgets $(R,M)$ and compute shadow-price gaps using the balanced exploration rule. At the implementability bounty, the marginal value of an extra discounted early impression equals the entrant’s target quality, while the marginal value of an extra expected payout dollar equals one in magnitude. Thus, a simple and robust control is:
\[
\Delta q\ \propto\ \min\!\left\{\,R-\widehat{q},\ \ \tau_q\,\big(\widehat{\mu}^{FB}-\widehat{\lambda}_{\mathrm{imp}}\big)\right\},\qquad
\Delta B\ \propto\ \min\!\left\{\,M-\widehat{B}^{\star}\widehat{P},\ \ \tau_B\,\big(1-\widehat{\lambda}_{\$}\big)\right\},
\]
with small step sizes $(\tau_q,\tau_B)$ and current estimates of shadow prices $(\widehat{\lambda}_{\mathrm{imp}},\widehat{\lambda}_{\$})$ inferred from budget tightness (e.g., by dual ascent on the Lagrangian in the main text). Update $(q,B)$ by $(q+\Delta q,\ \max\{0,B+\Delta B\})$.

Co-tune the bar to keep the pass event diagnostic by targeting a leverage corridor for the ratio $\widehat{\Lambda}=\widehat{P}'/\widehat{P}$ (or, equivalently, a corridor for the pair of pass rate and slope). If $\widehat{\Lambda}$ falls below the corridor, raise $s$ by one; if it exceeds, lower $s$ by one. Each move prompts a fresh measurement of $(\widehat{P},\widehat{P}')$ the following week. In parallel, adjust the engine to maintain a healthy prize spread: when $\widehat{\Delta H}$ drifts low, strengthen band separation (e.g., widen score gaps for high-priority lanes or increase the promoted set’s sampling frequency) until the inclusion path after a pass sits materially above the path after a fail. Because $(q,s)$ and $(H_0,\Delta H)$ interact, damp adjustments with exponential smoothing and change at most one unit per week to avoid oscillations.

Three safeguards make the loop reliable. First, segment the loop: maintain separate $(q,s,B)$ and $(\widehat{H}_0,\widehat{\Delta H})$ per surface or cohort with materially different monetization $\alpha$ or audience dynamics, while enforcing common shadow prices across segments so that resources flow to where marginal value is highest. Second, cap bar moves to keep pass rates within a governance window (e.g., 25-70\%); outside that window, the diagnosticity-per-dollar curve is typically flat and targeted bounties lose their advantage. Third, validate each change with a small, randomized holdout to detect nonstationarities (e.g., a competing launch that shifts the frontier), and roll back if the measured leverage ratio or prize spread degrades.

Once stabilized, the loop converges to a regime in which (i) guarantees are front-loaded and sized so that effort has a meaningful certain payoff, (ii) the bar sits where the pass event is most informative for the observed cohort, (iii) the engine’s prize spread is large enough that crossing the bar matters, and (iv) the posted bounty is just large enough to close any residual wedge from incomplete revenue sharing. In this regime, the mapping $\Xi(\mu)=q+H_0+\Delta H\,P(\mu)$ is not merely an analytic convenience: it is a live operational contract between the testing policy and the engine that sustains strong, cost-effective incentives week after week.

\section{Resource-Constrained Design and Balanced Exploration Algorithm}\label{app:resources}

This appendix develops the formal backbone of Section~\ref{sec:resources_bwk}. We write the platform’s constrained design problem at the per-entrant level, introduce shadow prices for the two resources (discounted early impressions and expected bounty spend), and derive explicit expressions for the marginal gains that enter the equal-marginal-value rule. The goal is to make the decomposition into \emph{direct} and \emph{indirect} effects fully transparent: one extra early impression raises value immediately by surfacing the entrant once more, and it also nudges investment upward; one extra expected payout dollar buys targeting power where the graduation probability is steep, and it, too, nudges investment. The formulas below quantify these channels and yield clean Kuhn-Tucker conditions for interior and corner solutions.

Throughout we take the generalized continuation representation from Appendix~\ref{app:engines} as primitive: after testing, a fail state earns expected discounted continuation $H_0$ and a pass state earns $H_1$, with spread $\Delta H\equiv H_1-H_0\ge 0$. The exposure aggregator is
\[
\Xi(\mu)\;=\;q\;+\;H_0\;+\;\Delta H\,P(\mu),
\]
and the platform’s per-entrant objective (net of bounty payments) is
\[
W(\mu; q,B)\;=\;\mu\,\Xi(\mu)\;-\;B\,P(\mu)\;=\;\mu\,[\,q+H_0+\Delta H P(\mu)\,]\;-\;B P(\mu).
\]
The entrant chooses quality $\mu$ to maximize
\[
\Pi_C(\mu; q,B)\;=\;\alpha\,\mu\,[\,q+H_0+\Delta H P(\mu)\,]\;+\;B\,P(\mu)\;-\;c(\mu),
\]
with unique best response $\mu^\star(q,B)$ characterized by the private first-order condition
\begin{equation}\label{eq:FOC-private-D}
G(\mu; q,B)\;\equiv\;\alpha\big[q+H_0+\Delta H P(\mu)\big]\;+\;\alpha\,\mu\,\Delta H\,P'(\mu)\;+\;B\,P'(\mu)\;-\;c'(\mu)\;=\;0,
\end{equation}
under Assumption~A1 (Appendix~\ref{app:proofs}). We suppress $(\bar{\mu},s,\alpha,H_0,\Delta H)$ in notation when no confusion arises. The pass probability $P(\mu)$ and slope $P'(\mu)$ are as defined in the main text and Appendix~\ref{app:proofs}.

\subsection*{D.1 Constrained objective and shadow prices}

The platform allocates two scarce resources per entrant: discounted early impressions (the guaranteed window $q$) and expected bounty spend (equal to $B\,P(\mu^\star)$). Let $R$ and $M$ denote the per-entrant budgets for these resources over a planning window. The per-entrant constrained program mirrors \eqref{eq:perentrant_program}:
\begin{equation}\label{eq:program-D}
\max_{q\ge 0,\;B\ge 0}\;\; W\!\left(\mu^\star(q,B);\,q,B\right)
\quad\text{s.t.}\quad
q\;\le\;R,\qquad B\,P\!\left(\mu^\star(q,B)\right)\;\le\;M.
\end{equation}
Introduce shadow prices $\lambda_{\mathrm{imp}}\ge 0$ (per discounted early impression) and $\lambda_{\$}\ge 0$ (per expected payout dollar). The Lagrangian is
\begin{equation}\label{eq:Lagrangian-D}
\mathcal{L}(q,B;\lambda_{\mathrm{imp}},\lambda_{\$})
\;=\;W\!\left(\mu^\star;q,B\right)\;-\;\lambda_{\mathrm{imp}}\,(q-R)\;-\;\lambda_{\$}\,\big(B P(\mu^\star)-M\big),
\end{equation}
where $\mu^\star$ abbreviates $\mu^\star(q,B)$. At an interior optimum the first-order conditions read
\[
\frac{\partial \mathcal{L}}{\partial q}\;=\;0,\qquad \frac{\partial \mathcal{L}}{\partial B}\;=\;0,
\]
with complementary-slackness and nonnegativity conditions on $(\lambda_{\mathrm{imp}},\lambda_{\$})$. We now expand the two derivatives by the chain rule to obtain the marginal-benefit objects that appear in the balanced exploration rule.

\subsection*{D.2 Marginal gains from $q$ and $B$ (direct and indirect channels)}

Write $\mu_q\equiv \partial \mu^\star/\partial q$ and $\mu_B\equiv \partial \mu^\star/\partial B$. By the implicit-function theorem applied to \eqref{eq:FOC-private-D},
\begin{equation}\label{eq:mu-derivs}
\mu_q\;=\;-\frac{\partial G/\partial q}{\partial G/\partial \mu}\;=\;\frac{\alpha}{D(\mu^\star)},\qquad
\mu_B\;=\;-\frac{\partial G/\partial B}{\partial G/\partial \mu}\;=\;\frac{P'(\mu^\star)}{D(\mu^\star)},
\end{equation}
where the \emph{gap curvature}
\begin{equation}\label{eq:denominator-D}
D(\mu^\star)\;\equiv\;c''(\mu^\star)\;-\;\alpha\,\Delta H\Big(2\,P'(\mu^\star)+\mu^\star P''(\mu^\star)\Big)\;-\;B\,P''(\mu^\star)\;>\;0
\end{equation}
by Assumption~A1. Thus, as established in Appendix~A.4, both $\mu_q$ and $\mu_B$ are nonnegative: more guaranteed exposure or a larger bounty induces (weakly) higher quality.

Differentiate $W(\mu^\star;q,B)$ with respect to $q$:
\[
\frac{dW}{dq}\;=\;\underbrace{\frac{\partial W}{\partial q}}_{\text{direct}}\;+\;\underbrace{\frac{\partial W}{\partial \mu}\,\mu_q}_{\text{indirect via investment}}.
\]
The direct term is the certain extra engagement from one more early impression:
\[
\frac{\partial W}{\partial q}\;=\;\mu^\star.
\]
The indirect term multiplies the induced change in quality by the \emph{planner-side marginal value of quality at the current policy},
\begin{equation}\label{eq:MV-planner}
\frac{\partial W}{\partial \mu}\;=\;q+H_0+\Delta H\,P(\mu^\star)\;+\;\mu^\star\,\Delta H\,P'(\mu^\star)\;-\;B\,P'(\mu^\star).
\end{equation}
Combining and substituting \eqref{eq:mu-derivs} gives the marginal gain from one additional discounted early impression:
\begin{equation}\label{eq:MBq}
\mathrm{MB}_q\;\equiv\;\frac{dW}{dq}\;=\;\mu^\star\;+\;\frac{\alpha}{D(\mu^\star)}\Big[q+H_0+\Delta H\,P(\mu^\star)+\mu^\star \Delta H\,P'(\mu^\star)-B\,P'(\mu^\star)\Big].
\end{equation}
The first term is transparent; the second says that $q$ is especially valuable when the current policy leaves a large wedge between the planner’s marginal value of quality and zero (the bracketed term) and when creators are responsive (large $\alpha$ and small $D(\mu^\star)$).

Differentiate $W$ with respect to $B$:
\[
\frac{dW}{dB}\;=\;\underbrace{\frac{\partial W}{\partial B}}_{\text{direct}}\;+\;\underbrace{\frac{\partial W}{\partial \mu}\,\mu_B}_{\text{indirect via investment}}.
\]
The direct term is negative: higher bounty pays more upon graduation,
\[
\frac{\partial W}{\partial B}\;=\;-\,P(\mu^\star).
\]
The indirect term is positive when the planner’s marginal value of quality (the bracket in \eqref{eq:MV-planner}) is positive, because $\mu_B\ge 0$:
\begin{equation}\label{eq:MBB-instrument}
\mathrm{MB}^{\text{(instr)}}_B\;\equiv\;\frac{dW}{dB}\;=\;-\,P(\mu^\star)\;+\;\frac{P'(\mu^\star)}{D(\mu^\star)}\Big[q+H_0+\Delta H\,P(\mu^\star)+\mu^\star \Delta H\,P'(\mu^\star)-B\,P'(\mu^\star)\Big].
\end{equation}
Expression \eqref{eq:MBB-instrument} is the \emph{marginal gain per unit of the \emph{instrument} $B$}. The resource the planner budgets, however, is expected payout $B\,P(\mu^\star)$. One unit increase in $B$ raises expected payout by
\[
\frac{d}{dB}\big[\,B\,P(\mu^\star)\,\big]\;=\;P(\mu^\star)\;+\;B\,P'(\mu^\star)\,\mu_B\;=\;P(\mu^\star)\;+\;B\,\frac{(P'(\mu^\star))^2}{D(\mu^\star)}.
\]
Hence the \emph{marginal gain per expected payout dollar} is the ratio
\begin{equation}\label{eq:MB-dollar}
\mathrm{MB}_{\$}\;\equiv\;\frac{dW/dB}{d\,[B P(\mu^\star)]/dB}\;=\;\frac{-\,P(\mu^\star)\;+\;\dfrac{P'(\mu^\star)}{D(\mu^\star)}\Big[q+H_0+\Delta H\,P(\mu^\star)+\mu^\star \Delta H\,P'(\mu^\star)-B\,P'(\mu^\star)\Big]}{P(\mu^\star)\;+\;B\,\dfrac{(P'(\mu^\star))^2}{D(\mu^\star)}}.
\end{equation}

Two observations make \eqref{eq:MBq}-\eqref{eq:MB-dollar} easy to use. First, both expressions are increasing in the \emph{diagnostic leverage} of the pass event at the current equilibrium, through $P'(\mu^\star)$ and (for the bounty) through the denominator that scales with $B (P')^2/D$. This formalizes the intuition that dollars and impressions are most valuable when the graduation margin is steep. Second, the bracket in \eqref{eq:MBq}-\eqref{eq:MBB-instrument} is the planner’s marginal value of $\mu$ net of the transfer term $B P'(\mu^\star)$. When that bracket is small (because the policy is already close to the benchmark that aligns private and social incentives), the \emph{indirect} channel contributes little and the marginal value of resources is dominated by the direct terms ($\mu^\star$ for impressions, $-P(\mu^\star)$ for dollars). When the bracket is large, the induced quality response materially boosts both marginal values.

\paragraph{Equal-marginal-value rule and Kuhn-Tucker characterization.}
Differentiating the Lagrangian \eqref{eq:Lagrangian-D} yields the necessary conditions
\begin{equation}\label{eq:KKT-equal}
\mathrm{MB}_q\;=\;\lambda_{\mathrm{imp}},\qquad \mathrm{MB}_{\$}\;=\;\lambda_{\$},
\end{equation}
at any interior optimum with both constraints binding. At corners the inequalities reverse in the familiar way:
\[
q=0\ \Rightarrow\ \mathrm{MB}_q\le \lambda_{\mathrm{imp}},\qquad
B=0\ \Rightarrow\ \mathrm{MB}_{\$}\le \lambda_{\$},
\]
and analogously when a budget slackens (shadow price zero). Condition \eqref{eq:KKT-equal} is the \emph{balanced exploration rule}: raise $q$ until its marginal gain per discounted impression equals the impression shadow price, and raise $B$ until its marginal gain per expected dollar equals the cash shadow price. The ratio form,
\[
\frac{\mathrm{MB}_q}{\lambda_{\mathrm{imp}}}\;=\;\frac{\mathrm{MB}_{\$}}{\lambda_{\$}}\;=\;1,
\]
is convenient when shadow prices are known up to a common scale (e.g., from portfolio-level budget constraints).

\textbf{Two instructive special cases:}

\emph{(i) Near-alignment regimes.}
Suppose the policy is tuned so that the planner’s marginal value at $\mu^\star$ is small, i.e., the bracket in \eqref{eq:MV-planner} evaluated at $\mu^\star$,
\[
K(\mu^\star)\;\equiv\;q+H_0+\Delta H\,P(\mu^\star)+\mu^\star\,\Delta H\,P'(\mu^\star)\;-\;B\,P'(\mu^\star),
\]
is near zero. Using \eqref{eq:MB-dollar}, write
\[
\mathrm{MB}_{\$}(B)\;=\;\frac{-\,P(\mu^\star)\;+\;\dfrac{P'(\mu^\star)}{D(\mu^\star)}\,K(\mu^\star)}
{P(\mu^\star)\;+\;B\,\dfrac{(P'(\mu^\star))^{2}}{D(\mu^\star)}}
\;=\;\frac{A - d\,B}{P(\mu^\star)+d\,B},
\]
with $A\equiv -P(\mu^\star)+\dfrac{P'(\mu^\star)}{D(\mu^\star)}K(\mu^\star)$ and $d\equiv \dfrac{(P'(\mu^\star))^{2}}{D(\mu^\star)} > 0$. Under near alignment $K(\mu^\star)\approx 0$, so $A\approx -P(\mu^\star)$ and
\[
\mathrm{MB}_{\$}(B)\ \approx\ \frac{-P(\mu^\star)-d\,B}{P(\mu^\star)+d\,B}\ \equiv\ -1\quad\text{for all }B.
\]
In words, when private and planner marginal values are already aligned at the equilibrium, an extra \emph{expected payout dollar} is approximately a pure transfer and the per-dollar marginal gain is very close to $-1$ throughout.

\emph{(ii) Quadratic costs.} If $c(\mu)=\tfrac{\kappa}{2}(\mu-\underline{\mu})^2$, then $D(\mu^\star)=\kappa-\alpha\,\Delta H(2P'+\mu^\star P'')-B P''$. When $q$ and $B$ are small (early pilots), $D(\mu^\star)\approx \kappa$ and
\[
\mu_q\ \approx\ \frac{\alpha}{\kappa},\qquad \mu_B\ \approx\ \frac{P'(\mu^\star)}{\kappa},\qquad
\mathrm{MB}_q\ \approx\ \mu^\star\;+\;\frac{\alpha}{\kappa}\big[q+H_0+\Delta H P+\mu^\star \Delta H P'-B P'\big],
\]
with an analogous simplification for $\mathrm{MB}_{\$}$. These linearized forms are helpful for back-of-the-envelope sizing in dashboards.

Equations \eqref{eq:MBq} and \eqref{eq:MB-dollar} translate directly into the design choices the platform controls. The more diagnostic the pass event (large $P'$), the more each instrument “buys” in terms of incentives; the more creators internalize engagement ($\alpha$ large), the more $q$ induces quality; the larger the prize spread $\Delta H$, the more both instruments leverage the graduation margin; and the tighter the budgets (large shadow prices), the more aggressively the policy should concentrate resources where the slope is steep. The Kuhn-Tucker conditions provide a practical stop rule: if $\mathrm{MB}_q$ falls below the impression shadow price, stop increasing $q$; if $\mathrm{MB}_{\$}$ falls below the cash shadow price, stop raising the bounty; if one instrument’s marginal value is well above its shadow price while the other’s is below, rebalance accordingly.

\subsection*{D.3 A balanced exploration algorithm: from formulas to weekly tuning}

The Kuhn-Tucker conditions \eqref{eq:KKT-equal} provide a target but not yet an operational procedure. This subsection turns \eqref{eq:MBq}-\eqref{eq:MB-dollar} into a simple weekly loop that co-tunes $(q,B)$ to budgets while keeping the policy near the implementability benchmark. The inputs are (i) estimates of the pass probability and slope around the current equilibrium, $\widehat{P}$ and $\widehat{P}'$ (Appendix~\ref{app:signals}), (ii) estimates of the continuation landscape $(\widehat{H}_0,\widehat{\Delta H})$ (Appendix~\ref{app:engines}), (iii) the revenue share $\widehat{\alpha}$ and the curvature proxy for costs (either a parametric $c''$ or a local second-difference estimate), and (iv) per-entrant budgets $(R,M)$ for discounted early impressions and expected bounty spend.

\emph{Step 1: Warm-start the equilibrium.} Given current $(q,B)$, solve the private FOC \eqref{eq:FOC-private-D} for $\mu^\star$. Because the left-hand side is strictly decreasing in $\mu$ (Assumption~A1), a bracketed root-finder converges rapidly; warm-starting from last week’s solution avoids recomputation. Record $\widehat{\mu}^\star$ and evaluate $\widehat{P}$ and $\widehat{P}'$ at that point.

\emph{Step 2: Compute marginal gains.} Form the curvature denominator $\widehat{D}=c''(\widehat{\mu}^\star)-\widehat{\alpha}\,\widehat{\Delta H}\big(2\widehat{P}'+\widehat{\mu}^\star \widehat{P}''\big)-B\,\widehat{P}''$ (with $\widehat{P}''$ from a smooth pass-curve fit; if unavailable, set $\widehat{P}''\!=\!0$ as a conservative approximation). Plug into \eqref{eq:MBq} and \eqref{eq:MB-dollar} to obtain $\widehat{\mathrm{MB}}_q$ and $\widehat{\mathrm{MB}}_{\$}$.

\emph{Step 3: Update the instruments (primal) and shadow prices (dual).} Let $(\lambda_{\mathrm{imp}},\lambda_{\$})$ denote current shadow-price iterates. A robust primal-dual step is
\begin{align*}
q^{+} &\leftarrow \Pi_{[0,R]}\!\Big(q\ +\ \eta_q\,(\widehat{\mathrm{MB}}_q-\lambda_{\mathrm{imp}})\Big),\\
B^{+} &\leftarrow \Pi_{[0,\overline{B}]}\!\Big(B\ +\ \eta_B\,(\widehat{\mathrm{MB}}_{\$}-\lambda_{\$})\Big),\\[2pt]
\lambda_{\mathrm{imp}}^{+} &\leftarrow \Big[\lambda_{\mathrm{imp}}\ +\ \rho\,\big(q^{+}-R\big)\Big]_{+},\qquad
\lambda_{\$}^{+} \leftarrow \Big[\lambda_{\$}\ +\ \rho\,\big(B^{+}\widehat{P}-M\big)\Big]_{+},
\end{align*}
where $\Pi$ projects to the feasible interval, $[\,\cdot\,]_+$ clips at zero, $\overline{B}$ is an administrative cap, and $(\eta_q,\eta_B,\rho)$ are small step sizes (e.g., $10^{-2}$-$10^{-1}$) with exponential smoothing across weeks to avoid oscillations. Intuitively, increase an instrument when its marginal gain exceeds the current shadow price; increase a shadow price when its budget is exceeded.

\emph{Step 4: Keep the pass event diagnostic and the schedule front-loaded.} Independently of the budget loop, maintain a corridor for diagnostic leverage $\Lambda=\widehat{P}'/\widehat{P}$ (or equivalently for the pair of pass rate and slope). If $\Lambda$ drifts low (pass events too rare or too common), adjust the bar by one unit $s\mapsto s\pm 1$; if novelty decay is present, ensure guarantees are scheduled as early as eligibility permits (Appendix~\ref{app:signals}, B.5; Appendix~\ref{app:proofs}, A.5). Recompute $(\widehat{P},\widehat{P}')$ after any bar or schedule change.

\emph{Step 5: Optional implementability snap-to.} When governance prefers the first-best benchmark, replace $B$ by the implementability bounty
\[
\widehat{B}^{\star}\ =\ \frac{\big[q+\widehat{H}_0+\widehat{\Delta H}\,\widehat{P}+\widehat{\mu}^{FB}\,\widehat{\Delta H}\,\widehat{P}'\big]\,(1-\widehat{\alpha})}{\widehat{P}'},
\]
clipped to satisfy the cash budget ($\widehat{B}^{\star}\widehat{P}\le M$). Near this point the bracket in \eqref{eq:MV-planner} is small, which simplifies marginal-gain monitoring: $\mathrm{MB}_q\approx \widehat{\mu}^{FB}$ and $\mathrm{MB}_{\$}\in(-1,0)$.

\emph{Convergence and stability.} Under standard conditions for subgradient methods (diminishing step sizes; slowly moving telemetry), the iterates converge to a neighborhood of the Kuhn-Tucker set. Two practical safeguards tighten the loop: cap weekly changes to one unit in $q$ or $s$ and to a small fraction (e.g., 10\%) for $B$, and validate each move in a small randomized holdout to detect nonstationarities (e.g., an engine update).

\subsection*{D.4 Estimating shadow prices, trading resources, and segment allocation}

\emph{Shadow prices and the value of budgets.} At the constrained optimum, the envelope theorem gives
\[
\frac{\partial W^{\star}}{\partial R}\;=\;\lambda_{\mathrm{imp}}\ \ \ge 0,\qquad
\frac{\partial W^{\star}}{\partial M}\;=\;\lambda_{\$}\ \ \ge 0,
\]
so the shadow prices are interpretable as the marginal value of one more discounted early impression and one more expected payout dollar, respectively. They can be estimated two ways: (i) \emph{internally} from the dual iterates in the primal-dual loop; or (ii) \emph{experimentally} by running small budget perturbations (e.g., $\pm 5\%$) in otherwise identical cohorts and measuring the induced change in constrained welfare, $\Delta W/\Delta R$ and $\Delta W/\Delta M$.

\emph{An exchange rate between attention and cash.} When attention is cut by $\Delta R<0$ and one wishes to offset the loss by increasing the cash budget $\Delta M>0$ with minimal impact on constrained welfare, a first-order rule is
\[
\lambda_{\mathrm{imp}}\,\Delta R\ +\ \lambda_{\$}\,\Delta M\ \approx\ 0
\quad\Longrightarrow\quad
\Delta M\ \approx\ -\,\frac{\lambda_{\mathrm{imp}}}{\lambda_{\$}}\ \Delta R.
\]
The ratio $\lambda_{\mathrm{imp}}/\lambda_{\$}$ is the \emph{resource exchange rate}: expected payout dollars required to compensate one lost discounted impression. It is high when pass events are weakly diagnostic (cash is inefficient) and low when the slope is steep.

\emph{Sensitivity to monetization.} A change in monetization $\alpha$ shifts both instruments’ potency. Near the implementability bounty, $B^\star$ scales with $(1-\alpha)$ (Appendix~\ref{app:proofs}, A.6). Thus, when $\alpha$ increases (e.g., a surface starts sharing more ad revenue), cash needs fall, and attention becomes the binding resource; when $\alpha$ decreases, a larger bounty is required to maintain alignment, and the cash shadow price rises. Dashboards should therefore track $(\widehat{\alpha},\widehat{B}^{\star}\widehat{P})$ jointly and precompute how much $M$ must rise or can fall to keep the same target $\widehat{\mu}^{FB}$.

\emph{Corners and failure modes.} Two edge cases deserve explicit treatment. First, if the measured slope $\widehat{P}'$ is small at the current equilibrium, both $\mathrm{MB}_q$’s indirect term and $\mathrm{MB}_{\$}$ collapse; targeted dollars do little work. The remedy is to retune the bar (move $s$ toward the cohort’s center) or, if measurement is the culprit, to improve label sensitivity (Appendix~\ref{app:signals}, B.4) before spending cash. Second, if $H_0$ is large and $\Delta H$ is small (a “soft” engine with weak band separation), slope incentives are mechanically weak; increase the prize spread (engine tuning) or lean more on $q$ until the spread is restored.

\emph{Segment allocation with common shadow prices.} Let $s\in\mathcal{S}$ index segments with shares $\pi_s$ and primitives $(\alpha_s,c_s,\bar{\mu}_s,H_{0,s},\Delta H_s)$. For a fixed pair of shadow prices $(\lambda_{\mathrm{imp}},\lambda_{\$})$, each segment chooses $(q_s,B_s)$ to satisfy
\[
\mathrm{MB}_{q,s}(q_s,B_s)\;=\;\lambda_{\mathrm{imp}},\qquad \mathrm{MB}_{\$,s}(q_s,B_s)\;=\;\lambda_{\$},
\]
using the segment-specific analogues of \eqref{eq:MBq} and \eqref{eq:MB-dollar}. Aggregating the resource uses $\sum_s \pi_s q_s\le R$ and $\sum_s \pi_s B_s P_s(\mu^\star_s)\le M$ then pins down $(\lambda_{\mathrm{imp}},\lambda_{\$})$. In discrete operations (integer $q_s$ and bounded $B_s$ grids), a simple greedy scheduler approximates the solution: (i) allocate one guaranteed impression at a time to the segment with the largest current $\mathrm{MB}_{q,s}$ until $R$ is exhausted; (ii) allocate one expected payout dollar at a time to the segment with the largest current $\mathrm{MB}_{\$,s}$ until $M$ is exhausted; (iii) recompute slopes and marginal gains after each allocation to respect diminishing returns. This “two water-fills” heuristic converges quickly and respects common shadow prices.

\emph{What to report and monitor.} To keep the loop auditable, dashboards should display: current $(q,s,B)$ by segment; delivered discounted mass $\hat{q}$; pass rate $\widehat{P}$ and slope $\widehat{P}'$ at the cohort median; continuation spread $\widehat{\Delta H}$; marginal gains $\widehat{\mathrm{MB}}_q$ and $\widehat{\mathrm{MB}}_{\$}$; shadow prices $(\widehat{\lambda}_{\mathrm{imp}},\widehat{\lambda}_{\$})$; and implied exchange rate $\widehat{\lambda}_{\mathrm{imp}}/\widehat{\lambda}_{\$}$. Movement in these quantities should justify all policy changes: bar moves when $\widehat{\Lambda}$ drifts, bounty moves when $\widehat{\alpha}$ or $M$ moves, and early-slot moves when $\widehat{\lambda}_{\mathrm{imp}}$ or novelty decay changes.

\noindent\textbf{Summary.} The resource-constrained calculus yields implementable rules. Compute the two marginal gains, compare them to their shadow prices, and adjust $(q,B)$ until equal-marginal-value holds under the current budgets. Keep the pass event diagnostic, the schedule front-loaded, and the engine’s prize spread meaningful. With these pieces in place, dollars and impressions are deployed precisely where their next unit buys the most incentive.
\end{document}